\documentclass[12pt,reqno]{amsart}
\usepackage{amsthm,amsfonts,amssymb,euscript,color}

\newtheorem{proposition}{Proposition}
\newtheorem{theorem}{Theorem}
\newtheorem{remark}{Remark}

\newtheorem{conjecture}{Conjecture}

\usepackage[margin=1in,dvips]{geometry}
\usepackage{bbm}
\usepackage{graphicx}
\usepackage{enumerate}
\def\ub {\underline{u}}
\def\th {\theta}
\def\Lb {\underline{L}}
\def\Hb {\underline{H}}
\def\chib {\underline{\chi}}
\def\chih {\hat{\chi}}
\def\chibh {\hat{\underline{\chi}}}
\def\omegab {\underline{\omega}}
\def\etab {\underline{\eta}}
\def\betab {\underline{\beta}}
\def\alphab {\underline{\alpha}}

\def\hot{\widehat{\otimes}}

\def\rhoc{\check{\rho}}
\def\sigmac{\check{\sigma}}

\def\kappab{\underline{\kappa}}

\def\a {\alpha}
\def\b {\beta}
\def\ab {\alphab}
\def\bb {\betab}
\def\nab {\nabla}
\def\ep {\epsilon}
\def\om {\omega}
\def\omb {\omegab}
\def\Om {\Omega}
\def\rd {\partial}

\def\f {\frac}

\newcommand\pr[1]{\frac{\partial}{\partial #1}}
\newcommand{\bea}{\begin{eqnarray}}
\newcommand{\eea}{\end{eqnarray}}
\def\beaa{\begin{eqnarray*}}
\def\eeaa{\end{eqnarray*}}

\renewcommand{\div}{\mbox{div }}
\newcommand{\curl}{\mbox{curl }}
\newcommand{\trchb}{\mbox{tr} \chib}
\def\trch{\mbox{tr}\chi}

\def\ls{\lesssim}

\newcommand{\eps}{{\epsilon} \mkern-8mu /\,}
\newcommand{\Ls}{{\mathcal L} \mkern-10mu /\,}

\begin{document}

\title{Weak null singularities in general relativity}

\author{Jonathan Luk}
\address{Department of Mathematics, Stanford University,~Stanford~CA~94305-2125,~USA}
\email{jluk@stanford.edu}

\begin{abstract}
We construct a class of spacetimes (without symmetry assumptions) satisfying the vacuum Einstein equations with singular boundaries on two null hypersurfaces intersecting in the future on a 2-sphere. The metric of these spacetimes extends continuously beyond the singularities while the Christoffel symbols fail to be square integrable in a neighborhood of any point on the singular boundaries. The construction shows moreover that the singularities are stable in a suitable sense. These singularities are stronger than the impulsive gravitational spacetimes considered by Luk-Rodnianski and {conjecturally they are present} in the interior of {generic} black holes arising from gravitational collapse. 
\end{abstract}

\maketitle

\tableofcontents

\section{Introduction}

In this paper, we study the existence and stability of weak null singularities in general relativity without symmetry assumptions. More precisely, a weak null singularity is a singular null boundary of a spacetime $(\mathcal M, g)$ solving the Einstein equations
$$Ric_{\mu \nu}-\frac 12 R g_{\mu \nu} = T_{\mu \nu}$$
such that the Christoffel symbols blow up and are not square integrable while the metric is continuous up to the boundary. This can be interpreted as a terminal singularity of the spacetime as it cannot be made sense of as a weak solution\footnotemark\label{footnote.weak}\footnotetext{One can define a weak solution to the Einstein equations by requiring $\int_{\mathcal M} Ric(X,Y) -\frac 12 R g(X,Y)-T(X,Y) dVol =0$ in the weak sense for all compactly supported smooth vector fields $X$ and $Y$. After integration by parts, the minimal regularity required for the spacetime for this to be defined is that the Christoffel symbols are square-integrable. See the discussion {on p.13} of \cite{Chr}.} to the Einstein equations along the singular boundary. While the singularity is sufficiently strong to be terminal, {it is at the same time} sufficiently weak such that the metric in an appropriate coordinate system is continuous up to the boundary.

The study of weak null singularities began with the attempts to understand the (in)stability of the Cauchy horizon in the black hole interior of Reissner-Nordstr\"om spacetimes. Reissner-Nordstr\"om spacetimes are the unique two-parameter family of asymptotically flat (with two ends), spherically symmetric, static solutions to the Einstein-Maxwell equations. Their Penrose diagrams\footnote{for $0<{|Q|}<M$} are given by Figure 1. As seen in the Penrose diagram, the Reissner-Nordstr\"om solution possesses a smooth Cauchy horizon $\mathcal C\mathcal H^+$ in the interior of the black hole such that the spacetime can be extended non-uniquely as a smooth solution to the Einstein-Maxwell system. This feature is also shared\footnote{for $0<|a|<M$} by the Kerr family of solutions to the vacuum Einstein equations, which can also be {depicted} by a Penrose diagram given by Figure 1. According to the strong cosmic censorship conjecture (see Section \ref{SCCC} below), the Reissner-Nordstr\"om and Kerr spacetimes are expected to be non-generic and the smooth Cauchy horizons are expected to be unstable.
\begin{figure}[htbp]
\begin{center}
 
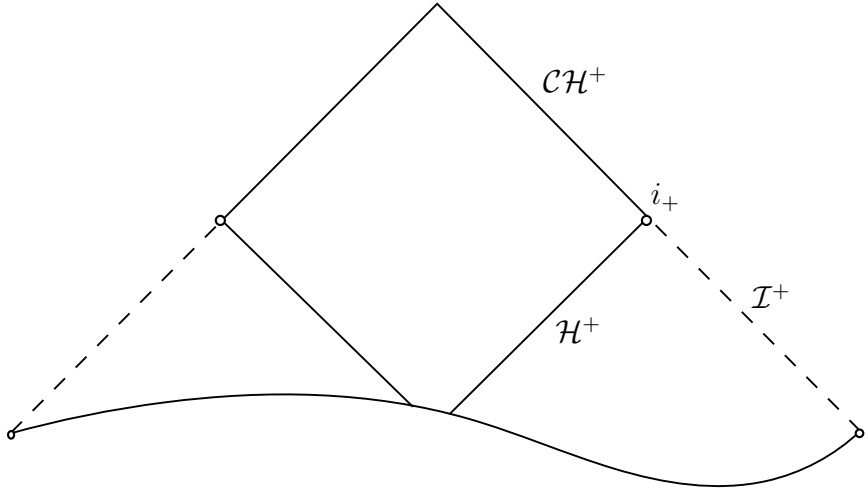
 
\caption{The Penrose diagram of Reissner-Nordstr\"om spacetimes}
\end{center}
\end{figure}

In a seminal work, Dafermos \cite{D1, D2} showed that for a spacetime solution to the spherically symmetric Einstein-Maxwell-real scalar field system, if an appropriate upper and lower bound for the scalar field is \emph{assumed} on the event horizon, then in a neighborhood of timelike infinity, the black hole terminates in a weak null singularity. The necessary upper bound was shown to hold for non-extremal black hole spacetimes arising from asymptotically flat initial data by Dafermos-Rodnianski \cite{DRPrice}. In particular, this implies that near timelike infinity, the terminal boundary of the Cauchy development does not contain a spacelike portion.

In a more recent work \cite{D3}, Dafermos showed that if, in addition to assuming the two black hole exterior regions settle to Reissner-Nordstr\"om with appropriate rates, the initial data are moreover globally close to that of Reissner-Nordstr\"om, then the maximal Cauchy development of the data possesses the same Penrose diagram as Reissner-Nordstr\"om. In particular, the spacetime terminates in a global bifurcate weak null singularity and the singular boundary does not contain any spacelike portion.

The works {\cite{D1, D2, D3} were} in part motivated by the physics literature on the instability of Cauchy horizons, weak null singularities and the strong cosmic censorship conjecture. It will be discussed below in Section \ref{SCCC}.

While the works of Dafermos \cite{D1, D2, D3} are restricted to the class of spherically symmetric spacetimes, they nonetheless suggest the genericity of weak null singularities in the black hole interior{, at least ``in a neighborhood of timelike infinity''}. In particular, they motivate the following conjecture for the vacuum Einstein equations
\begin{equation}\label{E.Eqn}
Ric_{\mu \nu}=0:
\end{equation}
\begin{conjecture}\label{Ori.conj}
\begin{enumerate}
\item {Consider the characteristic initial value problem with smooth characteristic initial data on a pair of null hypersurfaces $H_0$ and $\Hb_0$ intersecting on a $2$-sphere. Suppose that $H_0$ is an affine complete null hypersurface on which the data approach that of the event horizon of a Kerr solution (with $0<|a|<M$) at a sufficiently fast polynomial rate}\footnote{{In particular, this applies if an asymptotically flat spacetime has an exterior region which approaches a subextremal Kerr solution at a sufficiently fast polynomial rate. This also holds in the case where the Cauchy hypersurface has only one asymptotically flat end. In that case, numerical work in spherical symmetry \cite{HodPiran} suggests that the singular boundary may also contain a non-empty spacelike portion, in addition to the null portion.}}, then the development $(\mathcal M, g)$ of the initial data possesses a null boundary ``emanating from timelike infinity $i_+$'' through which the spacetime is extendible with a continuous metric {(see shaded region in Figure 2)}. Moreover, {given an appropriate ``lower bound'' on the $H_0$,} this piece of null boundary is generically a weak null singularity with non-square-integrable Christoffel symbols.
\item (Ori, see discussion in \cite{D3}) If the data for $(\mathcal M, g)$ {on a complete $2$-ended asymptotically flat Cauchy hypersurface are} globally a small perturbation of two-ended Kerr initial data (with $0<{|a|}<M$), then the maximal Cauchy development possesses a global bifurcate future null boundary $\partial \mathcal M$. Moreover, for generic such perturbations of Kerr, $\partial \mathcal M$ is a global bifurcate weak null singularity which intersects every futurely {causally} incomplete geodesic.
\end{enumerate}
\end{conjecture}

\begin{figure}[htbp]
\begin{center}
 
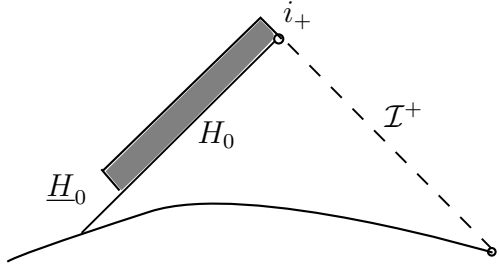
 
\caption{{Region of existence in Conjecture \ref{Ori.conj}.1}}
\end{center}
\end{figure}

If Conjecture \ref{Ori.conj} is true, then in particular there exist \emph{local stable} weak null singularities for the \emph{vacuum} Einstein equations \emph{without symmetry assumptions}. We show in this paper that there is in fact a large class of such singularities, parameterized by {singular} initial data. More specifically, we solve a characteristic initial value problem with singular initial data and construct a class of stable bifurcate weak null singularities. 

To motivate the strength of the singularity considered in this paper, we first recall the strength of the spherically symmetric weak null singularities in a neighborhood of Reissner-Nordstr\"om studied in \cite{D2}. The instability of the Reissner-Nordstr\"om Cauchy horizon is in fact already suggested by a linear analysis (see {\cite{CH, McN, Penrose.2}}). For a spherically symmetric solution to the linear wave equation which has a polynomially decaying (in the Eddington-Finkelstein coordinates) tail\footnote{with upper and lower bounds} along the event horizon, there is a singularity in a ($C^0$)-regular coordinate system near the Cauchy horizon of the strength\footnote{{This statement regarding the \emph{linear} wave equation can be inferred using the methods in \cite{D1} for the nonlinear coupled Einstein-Maxwell-scalar field system.}}
\bea
|\partial_{\ub} \phi|\sim (\ub_*-\ub)^{-1}{\log^{-p}(\frac{1}{\ub_*-\ub})},\label{D-strength}
\eea
for some $p>1$ as $\ub\to \ub_*$. In particular, along an outgoing null curve, $\partial_{\ub}\phi$ is integrable but not $L^q$-integrable for any $q>1$. In the spacetimes constructed by Dafermos \cite{D1, D2}, it was shown moreover that even in the nonlinear setting, $\partial_{\ub}\phi$ is also singular but remains integrable. A more precise analysis will show that in fact the spherically symmetric scalar field in the nonlinear setting of \cite{D2} also blows up at a rate given by \eqref{D-strength}.

Returning to the problem of constructing stable weak null singularities in vacuum, our construction is based on solving a characteristic initial value problem with singular data. We will in fact construct spacetimes not only with one weak null singularity, but instead contain two weak null singularities terminating at a bifurcate sphere. More precisely, the data on the initial characteristic hypersurface $H_0$ (resp. $\Hb_0$) is determined by the traceless part of the null second fundamental form $\chih$ (resp. $\chibh$). We consider singular initial data satisfying in particular
$$|\chih|\sim (\ub_*-\ub)^{-1}{\log^{-p}(\frac{1}{\ub_*-\ub})},\quad\mbox{for some }p>1,$$
and
$$|\chibh|\sim (u_*-u)^{-1}{\log^{-p}(\frac{1}{u_*-u})},\quad\mbox{for some }p>1.$$
This singularity is consistent with the strength of the weak null singularities in \eqref{D-strength}.

The following is a first version of the main result of this paper (see Figure 3). We refer the readers to the statement of Theorems \ref{extthm}, \ref{C0extthm} and \ref{blowupthm} for a more precise formulation of the theorem.
\begin{theorem}[Main theorem, first version]\label{thm.first.ver}
For a class of singular characteristic initial data {without any symmetry assumptions} for the vacuum Einstein equations 
$$Ric_{\mu\nu}=0$$
with the singular profile as above (see precise requirements on the data in Section \ref{sec.main.result}) and for $\epsilon$ sufficiently small and $u_*$, $\ub_*\leq \ep$, there exists a unique smooth spacetime $(\mathcal M, g)$ endowed with a double null foliation $(u,\ub)$ in $0\leq u<u_*$, $0\leq \ub< \ub_*$, which satisfies the vacuum Einstein equations with the given data. Associated to $(\mathcal M, g)$, there exists a coordinate system $(u,\ub,\th^1,\th^2)$ such that the metric extends continuously to the boundary but the Christoffel symbols are not in $L^2$.
\end{theorem}

\begin{figure}[htbp]
\begin{center}
 
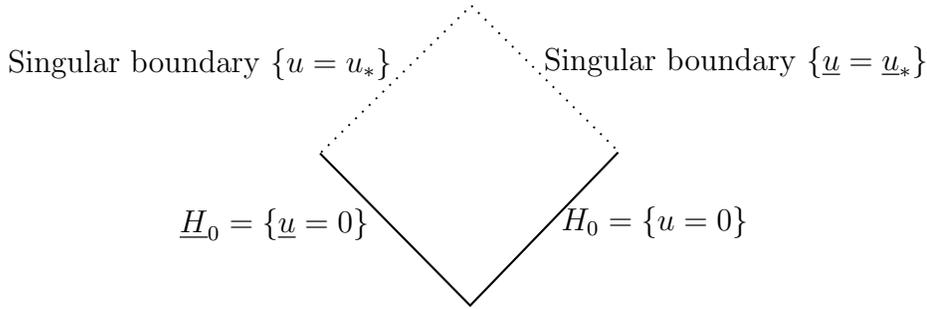
 
\caption{{Region of existence in Theorem \ref{thm.first.ver}}}
\end{center}
\end{figure}

\begin{remark}
This class of stable local weak null singularities that we construct in particular provides the first construction of weak null singularities {of such strength} for the vacuum Einstein equations.\footnote{We recall the Birkhoff's Theorem which states that the only spherically symmetric vacuum spacetimes are the Minkowski and Schwarzschild solutions. Thus to construct stable examples of weak null singularities in vacuum, one necessarily works outside the class of spherically symmetric spacetimes.}
\end{remark}

Theorem \ref{thm.first.ver} allows singularities on both initial null hypersurface and is valid in the region where $u_*$ and $\ub_*$ are sufficiently small. In the context of the interior of black holes, this corresponds to the darker shaded region in Figure 4. The existence theorem clearly implies an existence result when the data are only singular on one of the initial null hypersurfaces. In that context, we can in fact combine the methods in this paper with that in \cite{L} to show that the domain of existence can be extended so that only one of the characteristic length scales is required to be small. More precisely, we allow that data on $H_0$ such that
$$|\chih| \sim (\ub_*-\ub)^{-1}{\log^{-p}(\frac{1}{\ub_*-\ub})},\quad\mbox{for some }p>1,$$
on $0\leq \ub <\ub_*\leq C$  and the data on $\Hb_0$ are smooth on $0\leq u\leq u_*\leq \ep$. Then for $\epsilon$ sufficiently small, the spacetime $(\mathcal M, g)$ remains smooth in $0\leq u<u_*$, $0\leq \ub< \ub_*$ (see for example the lightly shaded region in Figure 4). We will omit the details of the proof of this result.

Theorem \ref{thm.first.ver}, which proves the existence and stability of the conjecturally generic weak null singularities, can be viewed as a first step towards Conjecture \ref{Ori.conj}. A next step is an analogue of \cite{D2} for the vacuum Einstein equations without symmetry assumptions, i.e., to solve the characteristic initial value problem inside the black hole with data prescribed on the event horizon that is approaching Kerr at appropriate rates. This requires an understanding of the \emph{formation} of weak null singularities from smooth data on the event horizon {(see part (1) of Conjecture \ref{Ori.conj})}. A full resolution of Conjecture \ref{Ori.conj}.2, however, requires in addition an understanding of the decay rates of gravitational radiation along the event horizon for generic perturbations of Kerr spacetime. This latter problem is intimately tied to the problem of the nonlinear stability of Kerr spacetimes, which continues to be one of the most important and challenging open problems in mathematical general relativity. Nevertheless, significant progress has been made for the corresponding \emph{linear} problem in the past decade. We refer the readers to the survey of Dafermos-Rodnianski \cite{DRsurvey} for more about this linear problem.

\begin{figure}[htbp]
\begin{center}
 
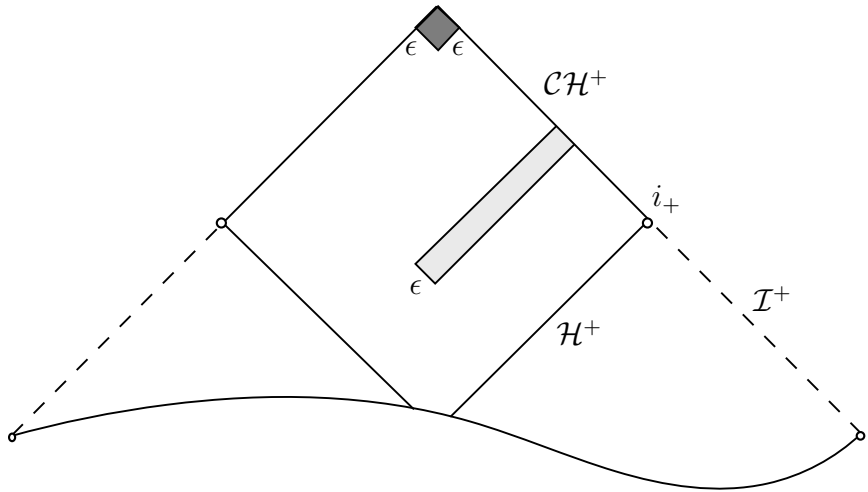
 
\caption{Domains of existence}
\end{center}
\end{figure}

The approach for the main theorem applies equally well to the Einstein-Maxwell-scalar field system without symmetry assumptions\footnote{This can be easily seen by decomposing the Maxwell field and the gradient of the scalar field in terms of the null frame below. The components in this decomposition obey equations that can be put in the same schematic form as in Section \ref{schnot}. Therefore, the Maxwell field and the scalar field and their derivatives satisfy estimates similar to those for the Ricci coefficients and curvature components.}. Thus, we show that the weak null singularity of Dafermos \cite{D2}, which arises from appropriately decaying data on the event horizon, is stable against \emph{non-spherically symmetric} perturbations on the hypersurface $\Sigma$ sufficiently far within the black hole region (see Figure 5). 

\begin{figure}[htbp]
\begin{center}
 
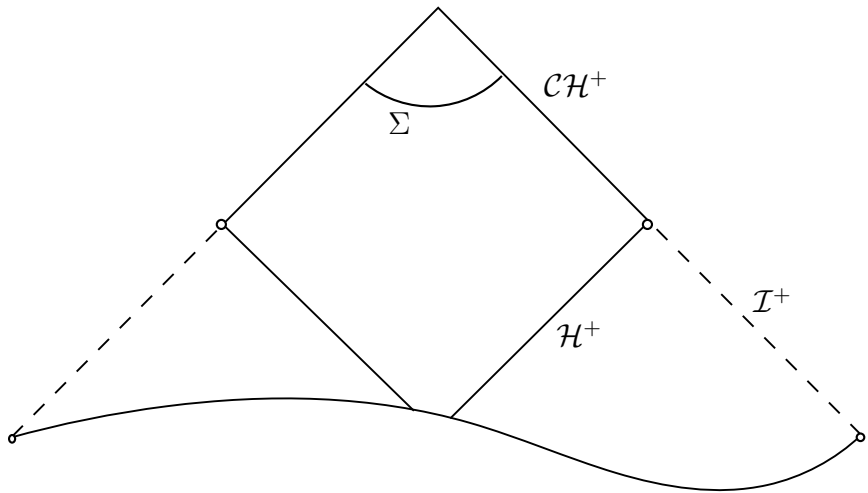
 
\caption{Perturbations in the black hole interior of Dafermos spacetimes}
\end{center}
\end{figure}

\subsection{Weak Null Singularities and Strong Cosmic Censorship Conjecture}\label{SCCC}

The study of weak null singularities can be viewed in the larger context of Penrose's celebrated strong cosmic censorship conjecture in general relativity. The conjecture states that for \emph{generic} asymptotically flat initial data for ``reasonable'' Einstein-matter systems, the maximal Cauchy development is future inextendible as a suitably regular Lorentzian manifold. This would guarantee general relativity to be a deterministic theory.

As pointed out above, the Kerr and Reissner-Nordstr\"om families of solutions (of the Einstein vacuum and Einstein-Maxwell equations respectively) have maximal Cauchy developments that are extendible as larger smooth spacetimes unless the angular momentum or the charge vanishes. This is connected with the existence of a smooth Cauchy horizon in the black hole interior such that the spacetime can be extended beyond as a smooth solution. According to the strong cosmic censorship conjecture, this is expected to be non-generic.

On the other hand, the situation for the Schwarzschild spacetime is more preferable from the point of view of the deterministic nature of the theory. The maximal development of the Schwarzschild spacetime terminates with a spacelike singularity at which the Hawking mass and the curvature scalar invariants blow up. In particular, the spacetime cannot be extended in $C^2$.

The early motivation for the strong cosmic censorship conjecture besides the desirability of a deterministic theory is a linear heuristic argument by Penrose \cite{Penrose.2} suggesting that the Reissner-Nordstr\"om Cauchy horizon is unstable. This was also confirmed by the numerical work by Simpson and Penrose \cite{SP}. It is thus conjectured that a small global perturbation would lead to a singularity in the interior of the black hole in such a way that the maximal Cauchy development is future inextendible.

However, the nature of the singular boundary in the interior of black holes was not well-understood\footnote{In particular, it was believed that a perturbation of the Reissner-Nordstr\"om Cauchy horizon would lead to a Schwarzschild type singularity.} until the first study of weak null singularity carried out by Hiscock \cite{Hiscock}. In an attempt to understand the instability of the Reissner-Nordstr\"om Cauchy horizon, he considered the Vaidya model allowing for a self-gravitating ingoing null dust. In this model, an explicit solution can be found and he showed that various components of the Christoffel symbols blow up. This, however, was called a whimper singularity as the Hawking mass and the curvature scalar invariants remain bounded.

In subsequent works, Poisson-Israel \cite{PI1, PI2} added an outgoing null dust to the model considered by Hiscock. While explicit solutions were not available, they were able to deduce that the second outgoing null dust would cause the Hawking mass to blow up at the null singularity. It was then thought of as a stronger singularity than that of Hiscock. 

However, from the point of view of partial differential equations, it is more natural to view this singularity at the level of {the} non-square-integrability of the Christoffel symbols, which is exactly the threshold such that the spacetime cannot be defined as a weak solution to the Einstein equations. From this perspective, the singularity of Poisson-Israel is as strong as that of Hiscock and both singularities can be viewed as terminal boundaries for the spacetimes in question.

While the Christoffel symbols blow up {at the} Cauchy horizon, one can also think that the Cauchy horizon is ``stable'' in the sense that no singularity arises before the ``original Cauchy horizon''. In particular, there is no spacelike portion of the singular boundary in a neighborhood of timelike infinity. This is thus contrary to the case of the Schwarzschild spacetime. This weak null singularity picture has been further explored and justified in many numerical works (see \cite{BDIM, BS, Burko}).

As we described before, the aforementioned picture of the interior of black holes was finally established by Dafermos in the context of the spherically symmetric Einstein-Maxwell-scalar field system \cite{D1}. This is the main motivation for our present work in which we initiate the study of weak null singularities of similar strength in vacuum without any symmetry assumptions. 

{Finally, we note that a class of \emph{analytic} spacetimes with slightly weaker singularities have been previously constructed in \cite{FO}. While this class of spacetime is more restrictive, as discussed in \cite{FO}, it nonetheless admits the full ``functional degrees of freedom'' of the Einstein equations.}

\subsection{Comparison with Impulsive Gravitational Waves}

As pointed out by Dafermos \cite{D3}, the weak null singularities that we consider in this paper share many similarities with impulsive gravitational waves. The latter are vacuum spacetimes admitting null hypersurfaces which support delta function singularities in the Riemann curvature tensor. Explicit examples were first constructed by Penrose \cite{Penrose72}, Khan-Penrose \cite{KhanPenrose} and Szekeres \cite{Szekeres1}. In these spacetimes, while the Christoffel symbols are not continuous, they remain bounded. Therefore, in contrast with the weak null singularities that we consider here, these impulsive gravitational waves are not terminal singularities. In fact, the solution to the vacuum Einstein equation extends beyond the singularity and is smooth except across the singular hypersurface. Nevertheless, both scenarios represent singularities propagating along null hypersurfaces and from a mathematical point of view, the proofs of the existence theory for these singularities share many common features.

In recent joint works with Rodnianski \cite{LR, LR2}, we initiated the rigorous mathematical study for general impulsive gravitational waves without symmetry assumptions. We constructed the impulsive gravitational waves via solving the characteristic initial problem such that the initial data admit curvature delta singularities supported on an embedded 2-sphere. One of the new ideas in the proof is the use of renormalized energy estimates for the curvature components, i.e., instead of controlling the spacetime curvature components in $L^2$, we subtract off an $L^\infty$ correction from some curvature components. This allowed us to derive a closed system of $L^2$ estimates which is completely independent of the singular curvature components.

In \cite{LR2}, when the interaction of impulsive gravitational waves was studied, we also extended the analysis to include a class of spacetimes such that when measured in the worst direction, the Christoffel symbols are merely in $L^2$. {We proved an existence and uniqueness theorem for spacetimes with such low regularity and showed that the spacetime solution can be extended beyond the singularities. Notice that this result is in fact sharp: This is because if the Christoffel symbols fail to be square-integrable, the spacetime cannot be extended as a weak solution to the Einstein equations (see footnote \ref{footnote.weak}).}

By contrast, the {spacetimes} considered in this paper {have} Christoffel symbols {which} are\footnote{In fact, we allow initial data to be in $L^p$ only for $p=1${, but not for any $p>1$}.} not in $L^2$. Even though the {weak null} singularities {are terminal singularities in the sense} that there cannot be an existence theory beyond the{m}, the theory developed in \cite{LR, LR2} can be extended to control the spacetime \emph{up to the singularity}. Moreover, our main theorem, which allows for two weak null singularities terminating at their intersection, can be viewed as an extension of the result in \cite{LR2} on the interaction of two impulsive gravitational waves. In particular, the renormalized energy of \cite{LR, LR2} plays an important role in the proof of our main theorem. However, even after renormalization, the renormalized curvature is still singular (i.e., not in $L^2$) and has to be dealt with using an additional weighted estimate.

\subsection{Description of the Main Results}\label{sec.main.result}

Our setup is the characteristic initial value problem with initial data given on two null hypersurfaces $H_0$ and $\Hb_0$ intersecting at a 2-sphere $S_{0,0}$ (see Figure 6). We will follow the general notations in \cite{KN, Chr, KlRo}.

\begin{figure}[htbp]
\begin{center}
 
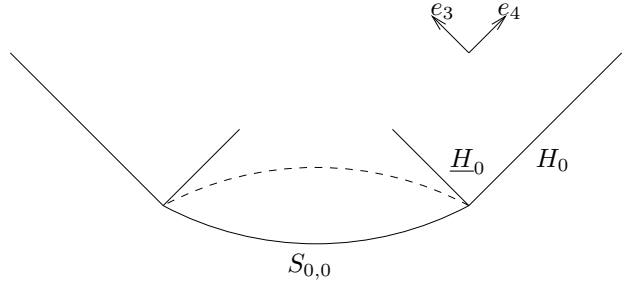
 
\caption{The Basic Setup}
\end{center}
\end{figure}

We introduce a null frame $\{e_1,e_2,e_3,e_4\}$ adapted to a double null foliation $(u,\ub)$ (see Section \ref{secdnf}). Denote the constant $u$ hypersurfaces by $H_u$, the constant $\ub$ hypersurfaces by $\Hb_{\ub}$ and their intersections by $S_{u,\ub}=H_u\cap \Hb_{\ub}$. Decompose the Riemann curvature tensor with respect to the null frame $\{e_1,e_2,e_3,e_4\}$:
\begin{equation*}
\begin{split}
\a_{AB}&=R(e_A, e_4, e_B, e_4),\quad \, \,\,   \ab_{AB}=R(e_A, e_3, e_B, e_3),\\
\b_A&= \frac 1 2 R(e_A,  e_4, e_3, e_4) ,\quad \bb_A =\frac 1 2 R(e_A,  e_3,  e_3, e_4),\\
\rho&=\frac 1 4 R(e_4,e_3, e_4,  e_3),\quad \sigma=\frac 1 4  \,^*R(e_4,e_3, e_4,  e_3){.}
\end{split}
\end{equation*}
We define also the Gauss curvature of the 2-spheres associated to the double null foliation to be $K$.
\noindent Define also the following Ricci coefficients with respect to the null frame:
\begin{equation*}
\begin{split}
&\chi_{AB}=g(D_A e_4,e_B),\, \,\, \quad \chib_{AB}=g(D_A e_3,e_B),\\
&\eta_A=-\frac 12 g(D_3 e_A,e_4),\quad \etab_A=-\frac 12 g(D_4 e_A,e_3){,}\\
&\omega=-\frac 14 g(D_4 e_3,e_4),\quad\,\,\, \omegab=-\frac 14 g(D_3 e_4,e_3),\\
&\zeta_A=\frac 1 2 g(D_A e_4,e_3){.}
\end{split}
\end{equation*}
Let $\chih$ (resp. $\chibh$) be the traceless part of $\chi$ (resp. $\chib$). 

The data on $H_0$ are given on $0\leq \ub < \ub_*$ such that $\chi$ becomes singular as $\ub \to \ub_*$. Similarly, the data on $\Hb_0$ is given on $0\leq u < u_*$ such that $\chib$ becomes singular as $u \to u_*$.

More precisely, let $f_1:{[0,\ub_*)} \to \mathbb R$ be a {smooth} function such that ${f_1}(x)\geq 0$ is decreasing and 
$$\int_0^{{\ub_*}} \frac{1}{f_1(x)^2} dx < \infty.$$
(resp. let $f_2:[0,{u_*}) \to \mathbb R$ be a {smooth} function such that ${f_2}(x)\geq 0$ is decreasing and 
$$\int_0^{{u_*}} \frac{1}{f_2(x)^2} dx < \infty.)$$
For example, $f_1$ can be taken to be $f_1(x)=({\ub_*}-x)^{\frac 12}\log^p (\frac{{1}}{{\ub_*}-x})$ for $p>\frac 12$.

Our main theorem shows local existence for a class of singular initial data with\footnote{{We assume also bounds for the angular derivatives which are consistent with this singular profile (see the precise statement in Theorem \ref{extthm}).}} 
$$|\chi(0,\ub)|\ls f_1(\ub)^{-2},\quad |\chib(u,0)|\ls f_2(u)^{-2}.$$
We construct a (unique) solution $(\mathcal M,g)$ to the vacuum Einstein equations in the region $u<u_*$, $\ub<\ub_*$, where $u_*$, $\ub_*\leq \ep$, and
\bea
\int_0^{\ub_*} f_1(\ub)^{-2} d\ub,\mbox{ }\int_0^{u_*} f_2(u)^{-2} du\leq \ep^2.\label{fsmall}
\eea
Here, $(u,\ub)$ is a double null foliation for $(\mathcal M,g)$ and the metric $g$ takes the form
$$g=-2\Omega^2(du\otimes d\ub+d\ub \otimes du)+\gamma_{AB} (d\th^A-b^A du)\otimes (d\th^B-b^B du)$$
in the $(u,\ub,\th^1,\th^2)$ coordinate system (to be defined in Section \ref{coordinates}). Define also $\nab$ to be the induced Levi-Cevita connection on the $2$-spheres of constant $u$ and $\ub$, i.e., $S_{u,\ub}$, and $\nab_3$, $\nab_4$ to be the projections of the covariant derivatves $D_3$, $D_4$ to the tangent space of $S_{u,\ub}$. Our main theorem (Theorem \ref{thm.first.ver}) can be stated precisely as a combination of Theorems \ref{extthm}, \ref{C0extthm} and \ref{blowupthm}. The first main result is the following theorem, which shows the existence of a spacetime up to the (potentially singular) null boundaries:
\begin{theorem}\label{extthm}
Consider the characteristic initial value problem for
\begin{equation}\label{v.EE}
Ric_{\mu\nu}=0
\end{equation}
with data that are smooth on $H_0\cap\{0\leq \ub< \ub_*\}$ and $\Hb_0\cap\{0\leq u< u_*\}$ such that
\begin{itemize}
\item There exists an atlas such that in each coordinate chart with local coordinates $(\th^1,\th^2)$, the initial metric $\gamma_0$ on $S_{0,0}$ obeys
$$d\leq \det\gamma_0\leq D, $$
and
$${\sum_{i_1+i_2\leq 6}|(\frac{\partial}{\partial\th^1})^{i_1}(\frac{\partial}{\partial\th^2})^{i_2}} \gamma_{BC}|\leq D.$$

\item The metric on $H_0$ and $\Hb_0$ satisfies the gauge conditions
$$\Omega=1\quad\mbox{on }H_0\mbox{ and }\Hb_0$$
and
{$$b^A=0\quad\mbox{on }\Hb_0.$$}

\item The Ricci coefficients on the initial hypersurface $H_0$ verify
$$\sum_{i\leq 5}\sup_{\ub}||{f_1}^2(\ub)\nab^i\chi||_{L^2(S_{0,\ub})}\leq D,$$
$$\sum_{i\leq 4}\sup_{\ub}||\nab^i\zeta||_{L^2(S_{0,\ub})}\leq D,$$
$$\sum_{i\leq 4}\sup_{\ub}||\nab^i\trchb||_{L^2(S_{0,\ub})}\leq D.$$

\item The Ricci coefficients on the initial hypersurface $\Hb_0$ verify
$$\sum_{i\leq 5}\sup_{u}||{f_2}^2(u)\nab^i\chib||_{L^2(S_{u,0})}\leq D,$$
$$\sum_{i\leq 4}\sup_{u}||\nab^i\zeta||_{L^2(S_{u,0})}\leq D,$$
$$\sum_{i\leq 4}\sup_{u}||\nab^i\trch||_{L^2(S_{u,0})}\leq D.$$

\end{itemize}
Then for $\ep$ sufficiently small (depending only on $d$ and $D$) and $u_*$, $\ub_*\leq \ep$, $||{f_1}(\ub)^{-1}||_{L^2_{\ub}}$, $||{f_2}(u)^{-1}||_{L^2_{u}}<\ep$, there exists a unique spacetime $(\mathcal M, g)$ endowed with a double null foliation $(u,\ub)$ in $0\leq u< u_*$ and $0\leq \ub <\ub_*$, which is a solution to the vacuum Einstein equations \eqref{v.EE} with the given data. Moreover, the spacetime remains smooth in $0\leq u< u_*$ and $0\leq \ub <\ub_*$.
\end{theorem}

\begin{remark}
In the following, we will only prove a priori estimates for spacetimes arising from these initial data (see Theorem \ref{aprioriestimates} below). The existence of a spacetime and the propagation of regularity follow from standard arguments. (For an example of this argument in low regularity, see Sections 4 and 5 in \cite{LR}. See also Chapter 16 in \cite{Chr}.)
\end{remark}

\begin{remark}
In order to simplify notations, we will omit the subscripts $1$ and $2$ in the weight functions $f_1$ and $f_2$. They can be inferred from whether $f$ is a function of $u$ or $\ub$.
\end{remark}

\begin{remark}
In Section \ref{sec.data}, we will construct a class of characteristic initial data which satisfies the assumptions of Theorem \ref{extthm}.
\end{remark}

While the weight $f$ in the spacetime norms allows the spacetime to be singular, the spacetime metric can be extended beyond the singular hypersurfaces $H_{u_*}$ and $\Hb_{\ub_*}$ continuously.
\begin{theorem}\label{C0extthm}
Under the assumptions of Theorem \ref{extthm}, the spacetime $(\mathcal M, g)$ can be extended continuously up to and beyond the singular boundaries $\Hb_{\ub_*}:=\{\ub=\ub_*\}$, $H_{u_*}:=\{u=u_*\}$. Moreover, the induced metric and null second fundamental form on the interior of the limiting hypersurfaces $\Hb_{\ub_*}$ and $H_{u_*}$ are regular. More precisely, {for any coordinate chart $U_i$ on $S_{0,0}$, the metric components $\gamma$, $b$, $\Omega$ satisfy the following estimates in the coordinate chart given by $U_i(u,\ub):=\underline{\Phi}_{\ub}\circ\Phi_u(U_i)$, where $\Phi_u$ and $\underline{\Phi}_{\ub}$ are the diffeomorphisms generated by $\Lb$ and $L$ respectively\footnote{See definition of $L$ and $\Lb$ in Section \ref{secdnf}.}:}
$${\sum_{i_1+i_2\leq 4}} \sup_{0\leq u\leq {u_*}}{\|(\pr{\th^1})^{i_1}(\pr{\th^2})^{i_2}} (\gamma,b,\Omega){\|}_{L^2(U_i(u,\ub_*))} \leq C{.} $$
{Moreover, for any fixed $U<u_*$, we have the following bounds for} the Ricci coefficients $\chibh,\trchb,\omb,\eta,\etab${:}
$${\sum_{j\leq 1}\sum_{i\leq 3-j}}\sup_{0\leq u\leq U}\|\nab_3^j\nab^i(\chibh,\trchb,\omb,\eta,\etab) \|_{L^2(S_{u,\ub_*})}\leq C_{U}.$$
Similar regularity statements hold on $H_{u_*}$.
\end{theorem}
\begin{remark}
If we assume in addition that the higher angular derivatives of $\chi$ are bounded in $L^1_{\ub}L^\infty(S)$, then the metric and the second fundamental form also inherit higher regularity in the interior of $\Hb_{\ub_*}$. In particular, if all angular derivatives of $\chi$ are bounded in $L^1_{\ub}L^\infty(S)$, then the metric restricted to $\Hb_{\ub_*}\cap\{0\leq u\leq U\}$ is smooth along the directions tangential to $\Hb_{\ub_*}$. Similar statements hold on $H_{u_*}$. We will omit the details.
\end{remark}
Moreover, we show that if initially the data are indeed singular, then $H_{u_*}$ and $\Hb_{\ub_*}$ are terminal singularities of the spacetime in the following sense:
\begin{theorem}\label{blowupthm}
If, in addition to the assumptions of Theorem \ref{extthm}, we also have the following for the initial data
$$\int_0^{\ub_*} |\chih(0,\ub)|^2 d\ub =\infty$$
along Lebesgue-almost every null generator on $H_0$, 
then the Christoffel symbols in the coordinate system $(u,\ub,\th^1,\th^2)$ do not belong to $L^2$ in a neighborhood of any point on $\Hb_{\ub_*}$.

Similarly if the initial data satisfy
$$\int_0^{u_*} |\chibh(u,0)|^2 du =\infty$$
along Lebesgue-almost every null generator on $\Hb_0$, 
then the Christoffel symbols in the coordinate system $(u,\ub,\th^1,\th^2)$ do not belong to $L^2$ in a neighborhood of any point on $H_{u_*}$.
\end{theorem}

\begin{remark}
Theorem \ref{blowupthm} guarantees that if we extend the spacetime metric continuously in the obvious differentiable structure given by the coordinate system $(u,\ub,\th^1,\th^2)$, {then} the Christoffel symbols are non-square-integrable in the extension. However, it is an open problem whether the spacetime admits any continuous extensions with square integrable Christoffel symbols.
\end{remark}

\subsection{Main Ideas of the Proof}

All the known proofs of regularity for the Einstein equations without symmetry assumptions rely on $L^2$ estimates on the metric and its derivatives or the Riemann curvature tensor and its derivatives. Let us denote schematically by $\Gamma$ a general Ricci coefficient and by $\Psi$ a general curvature component decomposed with respect to a null frame adapted to the double null foliation. In the double null foliation gauge (see for example, \cite{KN, Chr}), the standard approach to obtain a priori bounds is to couple the $L^2$ estimates for the curvature components
$$\int_H \Psi^2 +\int_{\Hb} \Psi^2 \leq \mbox{ Data }+\iint \Gamma\Psi\Psi$$
with the estimates for the Ricci coefficients obtained using the transport equations
$$\nab_3\Gamma=\Psi+\Gamma\Gamma,$$ 
$$\nab_4\Gamma=\Psi+\Gamma\Gamma.$$
However, in the setting of two weak null singularities, \emph{none} of the spacetime curvature components $\alpha, \beta, \rho, \sigma, \betab, \alphab$ are in $L^2$! 

Nevertheless, while these curvature components are singular, the nature of their singularity is specific. More precisely, while {the spacetime curvature components} $\rho$ and $\sigma$ are not in $L^2$, {they can be written as a sum of some regular intrinsic curvature components $K$ and $\sigmac$ (see further discussion in Section \ref{sec.ree}) which belong to $L^2$ and terms which are} quadratic in $\Gamma$. {We therefore prove $L^2$ estimates for $K$ and $\sigmac$, which we will call the ``renormalized curvature components'' (see \cite{LR, LR2}).} Moreover, {by considering $(K,\sigmac)$ instead of $(\rho,\sigma)$, we} remove all appearances of $\alpha$ and $\alphab$ in the estimates and so that we do not have to deal with the singularities of $\alpha$ and $\alphab$! It still remains to control the singular curvature components $\beta$ and $\betab$. Here, we make use of the fact that $\beta$ and $\betab$ are singular in a specific manner towards the singular boundary $\Hb_{\ub_*}$ and $H_{u_*}$ respectively. We therefore introduce degenerate $L^2$ norms that incorporate these singularities. We will explain the renormalization and the degenerate estimates in more detail below.

\subsubsection{Renormalized Energy Estimates}\label{sec.ree}
As described above, a main ingredient of the proof of the main theorem is the renormalized energy estimates introduced in \cite{LR, LR2} in the study of impulsive gravitational waves. This can be seen as follows. For the class of weak null singularities that we consider, while the $\Ls_L$ derivative of the spacetime metric blows up, the metric restricted to the $2$-sphere remains regular in the angular directions. Since the Gauss curvature $K$ is intrinsic to the $2$-spheres, it remains bounded. {On the other hand, by} the Gauss equation:
$$K=-{\rho}+\frac 12 \chih\cdot\chibh-\frac 14\trch\trchb$$
{and the fact that} $\trch$ and $\chih$ blow up {at $\ub=\ub_*$}, {$\rho$} also blows up {at $\ub=\ub_*$}. In view of this, we estimate the Gauss curvature {$K$} instead of the spacetime curvature component $\rho$.

Indeed, we {see} that the Gauss curvature $K$ satisfies equations such that the right hand side contains terms that are less singular than the terms in the corresponding equation for $\rho$. More precisely, for the curvature component $\rho$, we have {(up to lower order terms)} the Bianchi equation
$$\nab_4\rho+\frac 32\trch\rho= \div\beta -\frac 12 \chibh\cdot\alpha +... ,$$
{which contains the non-integrable curvature component $\alpha$. On the other hand, the Gauss curvature obeys the equation (see \eqref{eq:null.Bianchi2})}
$$\nab_4 K+\trch K= -\div\beta +...,$$
where {there are no terms containing $\alpha$ or are quadratic in $\trch$, $\chih$ and $\om$, i.e. every term on t}he right hand side of the equation {is} integrable in the $\ub$ direction\footnote{The can be compared with the renormalization introduced in \cite{LR} and \cite{LR2}, where we estimated $\rhoc=\rho-\frac 12 \chih\cdot\chibh$ instead of $\rho$. Whereas the renormalization using $\rhoc$ allows one to eliminate $\alpha$ in the estimates, it nonetheless introduces a term $\frac 14\trchb|\chih|^2$, which is not integrable in the $\ub$ direction in the setting of the present paper. Instead, by studying the equation for $K$, we see none of these terms which are quadratic in $\trch$, $\chih$ or $\om$! This fact can also be derived directly by considering the equations for $\nab_4 K$ using the intrinsic definition of the Gauss curvature.}. 

In a similar fashion, {by considering} the renormalized curvature component\footnote{{This is in fact related to the intrinsic curvature of the normal bundle to $S_{u,\ub}$.}}
$$\sigmac:=\sigma+\frac 12\chibh\wedge\chih$$
instead of $\sigma$, we see that it satisfies an equation such that all the terms on the right hand side are integrable in the $\ub$ direction.

{One} consequence of the renormalization is that we have completely removed the appearances of the curvature component $\alpha$ in the equations. In fact, as in \cite{LR, LR2}, this allows us to derive a set of estimates for the renormalized curvature component without requiring any information on the curvature component $\alpha$. 

Moreover, when considering the equations for $\nab_3 K$ and $\nab_3\sigmac$ for the renormalized curvature components, one sees that $\alphab$ does not appear and all the terms are integrable in the $u$ direction. Therefore, although $\alpha$ or $\alphab$ can be very singular near one of the singular boundaries, we do not need to derive any estimates for them!

\subsubsection{Degenerate $L^2$ Estimates}
Since the renormalization above deals with the singularity in the $\rho$ and $\sigma$ components and avoids any information on $\alpha$ and $\alphab$, it remains to derive appropriate $L^2$ estimates for $\beta$ and $\betab$. 

The main observation is that while $\beta$ and $\betab$ are both singular and fail to be in $L^2$, their {singularities} can be captured quantitatively. Consider the curvature component $\beta$. Since {the blow up rate of} $\trch$ and $\chih$ {can be bounded above by} $f(\ub)^{-2}$, in view of the Codazzi equations {in \eqref{null.str3}}, $\beta$ {is also bounded above by} $f(\ub)^{-2}${. In particular, while $\beta$} is only in $L^1_{\ub}$ but not in $L^p_{\ub}$ for any $p>1${,} the assumptions on the initial data allow us to control $f(\ub)\beta$ in $L^2_{\ub}$. We will thus incorporate this blow up in the norms and will be able to still use an $L^2$ based estimate.

The energy estimates will be obtained directly from two sets of Bianchi equations instead of using the Bel-Robinson tensor. Notice that since the energy estimates for $K,\sigmac$ are obtained either together with that for $\beta$ or that for $\betab$, even though $K$ and $\sigmac$ are regular, their energy estimates degenerate. Therefore, at the highest level of derivatives, we have to contend with the weaker $L^2$ estimates for these curvature components.

A potentially more serious challenge is that the introduction of the degenerate weights in $u$ and $\ub$ would create terms that cannot be estimated by the energy estimates themselves. Nevertheless, {since the weights are chosen to be decreasing towards the future, these uncontrollable terms in fact possess} a good sign{.}

\subsubsection{Estimates for the Ricci Coefficients}

As indicated above, the Ricci coefficients enter as error terms in the energy estimates. Thus, to close all the estimates, we need to control the Ricci coefficients $\Gamma$ by using the transport equations which in turn have the curvature components in the source terms. Since the various Ricci coefficients have different singular behavior, we separate them according to the bounds that they obey. More precisely, denote by $\psi_H$ the components that behave like $f(\ub)^{-2}$ as $\ub\to \ub_*$; by $\psi_{\Hb}$ the components that behave like $f(u)^{-2}$ as $u\to u_*$; and by $\psi$ the components that are bounded.

For the singular Ricci coefficients $\psi_H$, we have the following schematic transport equations:
$$\nab_3\psi_H= K+\nab\psi+\psi\psi+\psi_{\Hb}\psi_H.$$
The first three terms on the right hand side of this equation are bounded while the last term is singular. Nevertheless, the singularity of $\psi_{\Hb}$ still allows it to be controlled in $L^1$ along the $e_3$ direction. Thus, this equation can be integrated to show that the initial (singular) bounds for $\psi_H$ can be propagated. It is important that the terms of the form $\psi_H\psi_H$ and $\psi_{\Hb}\psi_{\Hb}$ do not appear in the equations. A similar structure can also be seen in the equation for the other singular Ricci coefficients $\psi_{\Hb}$, which takes the form
$$\nab_4\psi_{\Hb}= K+\nab\psi+\psi\psi+\psi_{\Hb}\psi_H.$$

For the regular Ricci coefficients $\psi$, we have transport equations of the form 
$$\nab_4\psi=\beta+\psi\psi_H,\quad\mbox{or}\quad\nab_3\psi=\betab+\psi\psi_{\Hb}.$$
The bounds that we prove show that the right hand side is integrable and therefore $\psi$ remains bounded. For example, in the $\nab_4$ equation, it is important that we do not have terms of the form $\psi_H\psi_H$, $\psi\psi_{\Hb}$, $\psi_H\psi_{\Hb}$ and $\psi_{\Hb}\psi_{\Hb}$, which are not uniformly bounded after integrating along the $e_4$ direction. 

\subsubsection{Null Structure in the Energy Estimates}

A priori, the degenerate $L^2$ estimates that we introduce may not be sufficient to control the error terms. Nevertheless, the vacuum Einstein equations possess a remarkable null structure which allows one to close the estimates using only the degenerate $L^2$ estimates.

For example, in the energy estimates for the singular component $\beta$, we have
$$||f(\ub)\beta||_{L^2(H)}^2\leq \mbox{ Data }+ ||f^2(\ub)(\beta\psi_{\Hb}\beta+\beta\psi_{H}\betab+\beta\psi K)||_{L^1_u L^1_{\ub} L^1(S)}.$$
To estimate the first term, it suffices to note that $\psi_{\Hb}$, while singular, can be shown to be small after integrating along the $u$ direction. Thus the first term can be controlled using Gronwall's inequality. For the second term, since the singularity for $\beta$ has the same strength as that for $\psi_H$ (and similarly the singularity for $\betab$ has the same strength as that for $\psi_{\Hb}$), the singularity in this term is similar to that in the first term and can also be bounded. The final term is less singular since $\psi$ and $K$ are both uniformly bounded.\footnote{Although as pointed out before, the highest derivative estimates for $K$ in the energy norm suffer a loss as one approaches the singular boundaries, this term can nevertheless be controlled.} Notice that if other combinations of curvature terms and Ricci coefficients such as $\beta\psi_H\beta$, $\beta\psi_{\Hb}\betab$ or $\beta\psi_H K$ appear in the error terms, the degenerate energy will not be strong enough to close the bounds!

In order to close all the estimates, we need to commute also with higher derivatives. As in \cite{LR, LR2}, we will only commute with angular covariant derivatives. These commutations will not introduce terms that are more singular. Moreover, the null structure of the estimates indicated above is also preserved under these commutations.

Similar to \cite{LR, LR2}, the renormalization introduces error terms in the energy estimates such that the Ricci coefficients have one more derivative compared to the curvature components. These terms cannot be estimated via transport equations alone but are controlled using also elliptic estimates on the spheres. A form of null structure similar to that described above also makes an appearance in these elliptic estimates, allowing all the bounds to be closed.

\subsection{Outline of the Paper}

We end the introduction with an outline of the remainder of the paper. In Section \ref{secsetup}, we introduce the basic setup of the paper, including the double null foliation, the coordinate system and the Einstein vacuum equations recast in terms of the geometric quantities associated to the double null foliation. In Section \ref{secnorm}, we introduce the norms used in the paper and state a theorem on a priori estimates (Theorem \ref{aprioriestimates}) which imply our main existence theorem (Theorem \ref{extthm}). In Section \ref{sec.data}, we construct a class of characteristic initial data satisfying the assumptions of Theorem \ref{extthm}. In Sections \ref{secbasic}-\ref{seccurv}, we prove Theorem \ref{aprioriestimates}. In Section \ref{secbasic}, we obtain the estimates for the metric components and derive functional inequalities useful in our setting. Then in Sections \ref{secRicci} and \ref{secRicci32}, we prove bounds for the Ricci coefficients assuming control of the curvature components. In Section \ref{seccurv}, we close all the estimates by obtaining bounds for the curvature components. Finally, in Section \ref{secnature}, we discuss the nature of the singular boundary and prove Theorems \ref{C0extthm} and \ref{blowupthm}.

\section{Basic Setup}\label{secsetup}

\subsection{Double Null Foliation}\label{secdnf}
For a {smooth}\footnote{{The spacetimes considered in this paper are not smooth at $u=u_*$ or $\ub=\ub_*$. However, since we first construct the spacetime in the region $\{u<u_*\}\cap\{\ub<\ub_*\}$ in which the spacetime is smooth (see Theorem \ref{extthm}), it suffices to define the double null foliation for smooth spacetimes.}} spacetime in a neighborhood of $S_{0,0}$, we define a double null foliation as follows: Let $u$ and $\ub$ be solutions to the eikonal equation
$$g^{\mu\nu}\partial_\mu u\partial_\nu u=0,\quad g^{\mu\nu}\partial_\mu\ub\partial_\nu \ub=0,$$
{such that} $u=0$ on $H_0$ and $\ub=0$ on $\Hb_0$.
Let
$$L'^\mu=-2g^{\mu\nu}\partial_\nu u,\quad \Lb'^\mu=-2g^{\mu\nu}\partial_\nu \ub.$$ 
These are null and geodesic vector fields. Let
$$2\Omega^{-2}=-g(L',\Lb').$$
Define
$$e_3=\Omega\Lb'\mbox{, }e_4=\Omega L'$$
to be the normalized null pair such that 
$$g(e_3,e_4)=-2$$
and
$$\Lb=\Omega^2\Lb'\mbox{, }L=\Omega^2 L'$$
to be the so-called equivariant vector fields.

In this paper, we will consider spacetime solutions to the vacuum Einstein equations \eqref{E.Eqn} in the gauge such that 
$$\Omega=1,\quad\mbox{on $H_0$ and $\Hb_0$}.$$

The level sets of $u$ (resp. $\ub$) are denoted by $H_u$ (resp. $\Hb_{\ub}$). The eikonal equations imply {that} $H_u$ and $\Hb_{\ub}$ are null hypersurface. The intersections of the hypersurfaces $H_u$ and $\Hb_{\ub}$ are topologically 2-spheres, which we denote by $S_{u,\ub}$. Note that the integral flows of $L$ and $\Lb$ respect the foliation $S_{u,\ub}$.

\subsection{The Coordinate System}\label{coordinates}
We define a coordinate system $(u,\ub,\th^1,\th^2)$ in a neighborhood of $S_{0,0}$ as follows:
On the sphere $S_{0,0}$, we have an atlas such that in the local coordinate system $(\th^1,\th^2)$ in each coordinate chart, the metric $\gamma$ is smooth, bounded and positive definite. Recall that in a neighborhood of $S_{0,0}$, $u$ and $\ub$ are solutions to the eikonal equations:
$$g^{\mu\nu}\partial_\mu u\partial_\nu u=0,\quad g^{\mu\nu}\partial_\mu\ub\partial_\nu \ub=0.$$
We then require the coordinates {to satisfy 
$$\Ls_{\Lb} \th^A=0$$ 
on the initial hypersurface $\Hb_0$ and 
$$\Ls_L \th^A=0$$
in the spacetime region. Here, $\Ls_L$ and $\Ls_{\Lb}$ denote} the restriction of the Lie derivative to $TS_{u,\ub}$ (See \cite{Chr}, Chapter 1) and $L$ and $\Lb$ are defined as in the Section \ref{secdnf}.
Relative to the coordinate system $(u,\ub,\th^1,\th^2)$, the null pair $e_3$ and $e_4$ can be expressed as
$$e_3=\Omega^{-1}\left(\frac{\partial}{\partial u}+b^A\frac{\partial}{\partial \th^A}\right),\quad e_4=\Omega^{-1}\frac{\partial}{\partial \ub},$$
for some $b^A$ such that $b^A=0$ on $\Hb_0$, while the metric $g$ takes the form
$$g=-2\Omega^2(du\otimes d\ub+d\ub\otimes du)+\gamma_{AB}(d\th^A-b^Adu)\otimes (d\th^B-b^Bdu).$$ 

\subsection{Equations}\label{seceqn}
We will recast the Einstein equations as a system for Ricci coefficients and curvature components associated to a null frame $e_3$, $e_4$ defined above and an orthonormal frame\footnote{{Of course the orthonormal frame is only defined locally. Alternatively, the capital Latin indices can be understood as abstract indices.}} {$\{e_A\}_{A=1,2}$} tangent to the 2-spheres $S_{u,\ub}$. {W}e define the Ricci coefficients relative to the null fame:
 \begin{equation}
\begin{split}
&\chi_{AB}=g(D_A e_4,e_B),\, \,\, \quad \chib_{AB}=g(D_A e_3,e_B),\\
&\eta_A=-\frac 12 g(D_3 e_A,e_4),\quad \etab_A=-\frac 12 g(D_4 e_A,e_3){,}\\
&\omega=-\frac 14 g(D_4 e_3,e_4),\quad\,\,\, \omegab=-\frac 14 g(D_3 e_4,e_3),\\
&\zeta_A=\frac 1 2 g(D_A e_4,e_3),
\end{split}
\end{equation}
where $D_A=D_{e_{(A)}}$. We also introduce the  null curvature components,
 \begin{equation}
\begin{split}
\a_{AB}&=R(e_A, e_4, e_B, e_4),\quad \, \,\,   \ab_{AB}=R(e_A, e_3, e_B, e_3),\\
\b_A&= \frac 1 2 R(e_A,  e_4, e_3, e_4) ,\quad \bb_A =\frac 1 2 R(e_A,  e_3,  e_3, e_4),\\
\rho&=\frac 1 4 R(e_4,e_3, e_4,  e_3),\quad \sigma=\frac 1 4  \,^*R(e_4,e_3, e_4,  e_3).
\end{split}
\end{equation}
Here $\, ^*R$ denotes the Hodge dual of $R$.  We denote by $\nab$ the 
induced covariant derivative operator on $S_{u,\ub}$ and by $\nab_3$, $\nab_4$
the projections to $S_{u,\ub}$ of the covariant derivatives $D_3$, $D_4$ (see
precise definitions in Chapter 3.1 of \cite{KN}). 

Observe that,
\begin{equation}\label{RC.relation}
\begin{split}
&\omega=-\frac 12 \nab_4 (\log\Omega),\qquad \omegab=-\frac 12 \nab_3 (\log\Omega),\\
&\eta_A=\zeta_A +\nab_A (\log\Omega),\quad \etab_A=-\zeta_A+\nab_A (\log\Omega).
\end{split}
\end{equation}

Define the following contractions of the tensor product $\phi^{(1)}$ and $\phi^{(2)}$ with respect to the metric $\gamma$:
$$\phi^{(1)}\cdot\phi^{(2)}:=(\gamma^{-1})^{AC}(\gamma^{-1})^{BD}\phi^{(1)}_{AB}\phi^{(2)}_{CD} \quad\mbox{for symmetric $2$-tensors $\phi^{(1)}_{AB}$, $\phi^{(2)}_{AB}$,}$$
$$\phi^{(1)}\cdot\phi^{(2)}:=(\gamma^{-1})^{AB}\phi^{(1)}_{A}\phi^{(2)}_{B} \quad\mbox{for $1$-forms $\phi^{(1)}_{A}$, $\phi^{(2)}_{A}$,}$$
$$(\phi^{(1)}\cdot\phi^{(2)})_A:=(\gamma^{-1})^{BC}\phi^{(1)}_{AB}\phi^{(2)}_{C} \quad\mbox{for a symmetric $2$-tensor $\phi^{(1)}_{AB}$ and a $1$-form $\phi^{(2)}_{A}$,}$$
$$(\phi^{(1)}\hot\phi^{(2)})_{AB}:=\phi^{(1)}_A\phi^{(2)}_B+\phi^{(1)}_B\phi^{(2)}_A-{\gamma_{AB}}(\phi^{(1)}\cdot\phi^{(2)}) \quad\mbox{for one forms $\phi^{(1)}_A$, $\phi^{(2)}_A$,}$$
$$\phi^{(1)}\wedge\phi^{(2)}:=\eps^{AB}(\gamma^{-1})^{CD}\phi^{(1)}_{AC}\phi^{(2)}_{BD}\quad\mbox{for symmetric two tensors $\phi^{(1)}_{AB}$, $\phi^{(2)}_{AB}$},$$
where $\eps$ is the volume form associated to the metric $\gamma$. We also define by $^*$ for $1$-forms and symmetric $2$-tensors respectively as follows (note that on $1$-forms this is the Hodge dual on $S_{u,\ub}$):
\begin{align*}
^*\phi_A := & \gamma_{AC} \eps^{CB} \phi_B, \\
^*\phi_{AB} := & \gamma_{BD} \eps^{DC} \phi_{AC}.
\end{align*}
Define the operator $\nab\widehat{\otimes}$ on a $1$-form $\phi_{A}$ by
$$(\nab\widehat{\otimes}\phi)_{AB} :=  \nab_A \phi_B + \nab_B \phi_A - \gamma_{AB} \div \phi.$$
For totally symmetric tensors, define the $\div$ and $\curl$ operators as follows
$$(\div\phi)_{A_1...A_r}:=\nabla^B\phi_{BA_1...A_r},$$
$$(\curl\phi)_{A_1...A_r}:=\eps^{BC}\nabla_B\phi_{CA_1...A_r}.$$
Define also the trace of totally symmetric tensors to be
$$(\mbox{tr}\phi)_{A_1...A_{r-1}}:=(\gamma^{-1})^{BC}\phi_{BCA_1...A_{r-1}}.$$

We separate the trace and traceless part of $\chi$ and $\chib$. Let $\chih$ and $\chibh$ be the traceless parts of $\chi$ and $\chib$ respectively. Then $\chi$ and $\chib$ satisfy the following null structure equations:
\begin{equation}
\label{null.str1}
\begin{split}
\nab_4 \trch+\frac 12 (\trch)^2&=-|\chih|^2-2\omega \trch{,}\\
\nab_4\chih+\trch \chih&=-2 \omega \chih-\alpha{,}\\
\nab_3 \trchb+\frac 12 (\trchb)^2&=-2\omegab \trchb-|\chibh|^2{,}\\
\nab_3\chibh + \trchb\,  \chibh&= -2\omegab \chibh -\alphab{,}\\
\nab_4 \trchb+\frac1 2 \trch \trchb &=2\omega \trchb +2\rho- \chih\cdot\chibh +2\div \etab +2|\etab|^2{,}\\
\nab_4\chibh +\frac 1 2 \trch \chibh&=\nab\widehat{\otimes} \etab+2\omega \chibh-\frac 12 \trchb \chih +\etab\widehat{\otimes} \etab{,}\\
\nab_3 \trch+\frac1 2 \trchb \trch &=2\omegab \trch+2\rho- \chih\cdot\chibh+2\div \eta+2|\eta|^2{,}\\
\nab_3\chih+\frac 1 2 \trchb \chih&=\nab\widehat{\otimes} \eta+2\omegab \chih-\frac 12 \trch \chibh +\eta\widehat{\otimes} \eta{.}
\end{split}
\end{equation}
The other Ricci coefficients satisfy the following null structure equations:
\begin{equation}
\label{null.str2}
\begin{split}
\nabla_4\eta&=-\chi\cdot(\eta-\etab)-\b{,}\\
\nabla_3\etab &=-\chib\cdot (\etab-\eta)+\bb{,}\\
\nabla_4\omegab&=2\omega\omegab-\eta\cdot\etab+\f12|\eta|^2+\frac 12 \rho,\\
\nabla_3\omega&=2\omega\omegab-\eta\cdot\etab+\f12|\etab|^2+\frac 12 \rho.
\end{split}
\end{equation}
The Ricci coefficients also satisfy the following constraint equations
\begin{equation}
\label{null.str3}
\begin{split}
\div\chih&=\frac 12 \nabla \trch - \frac 12 (\eta-\etab)\cdot (\chih -\frac 1 2 \trch) -\beta,\\
\div\chibh&=\frac 12 \nabla \trchb + \frac 12 (\eta-\etab)\cdot (\chibh-\frac 1 2   \trchb) +\betab{,}\\
\curl\eta &=-\curl\etab=\sigma +\frac 1 2\chibh \wedge\chih{,}\\
K&=-\rho+\frac 1 2 \chih\cdot\chibh-\frac 1 4 \trch \trchb{,}
\end{split}
\end{equation}
with $K$ the Gauss curvature of the spheres $S_{u,\ub}$.
The null curvature components satisfy the following null Bianchi equations:
\begin{equation}
\label{eq:null.Bianchi}
\begin{split}
&\nab_3\alpha+\frac 12 \trchb \alpha={\nabla}\hot \beta+ 4\omegab\alpha-3(\chih\rho+^*\chih\sigma)+
(\zeta+4\eta)\hot\beta,\\
&\nab_4\beta+2\trch\beta = \div\alpha - 2\omega\beta +  (2\zeta+\etab)\cdot \alpha,\\
&\nab_3\beta+\trchb\beta={\nabla}\rho + 2\omegab \beta +^*{\nabla}\sigma +2\chih\cdot\betab+3(\eta\rho+^*\eta\sigma),\\
&\nab_4\sigma+\frac 32\trch\sigma=-\div^*\beta+\frac 12\chibh\wedge\alpha-\zeta\wedge\beta-2\etab\wedge\beta,\\
&\nab_3\sigma+\frac 32\trchb\sigma=-\div ^*\betab-\frac 12\chih\wedge\alphab+\zeta\wedge\betab-2\eta\wedge\betab,\\
&\nab_4\rho+\frac 32\trch\rho=\div\beta-\frac 12\chibh\cdot\alpha+\zeta\cdot\beta+2\etab\cdot\beta,\\
&\nab_3\rho+\frac 32\trchb\rho=-\div\betab- \frac 12\chih\cdot\alphab+\zeta\cdot\betab-2\eta\cdot\betab,\\
&\nab_4\betab+\trch\betab=-{\nabla}\rho +^*{\nabla}\sigma+ 2\omega\betab +2\chibh\cdot\beta-3(\etab\rho-^*\etab\sigma),\\
&\nab_3\betab+2\trchb\betab=-\div\alphab-2\omegab\betab-(-2\zeta+\eta) \cdot\alphab,\\
&\nab_4\alphab+\frac 12 \trch\alphab=-{\nabla}\hot \betab+ 4\omega\alphab-3(\chibh\rho-^*\chibh\sigma)+
(\zeta-4\etab)\hot \betab.
\end{split}
\end{equation}
where $^*$ denotes the Hodge dual on $S_{u,\ub}$.

We now rewrite the Bianchi equations in terms of the Gauss curvature $K$ of the spheres $S_{u,\ub}$ and the renormalized curvature component $\sigmac$ defined by
$$\sigmac=\sigma+\frac 12 \chibh\wedge\chih.$$
The Bianchi equations take the following form
\begin{equation}
\label{eq:null.Bianchi2}
\begin{split}
\nab_3\beta+\trchb\beta=&-\nabla K  +^*{\nabla}\sigmac + 2\omegab \beta+2\chih\cdot\betab-3(\eta K-^*\eta\sigmac)+\frac 1 2({\nabla}(\chih\cdot\chibh)+^*{\nabla}(\chih\wedge\chibh))\\
&+\f 32(\eta\chih\cdot\chibh+^*\eta\chih\wedge\chibh)-\frac 14 (\nab\trch \trchb+\trch\nab\trchb)-\frac 34 \eta\trch\trchb,\\
\nab_4\sigmac+\frac 32\trch\sigmac=&-\div^*\beta-\zeta\wedge\beta-2\etab\wedge
\beta-\frac 12 \chih\wedge(\nab\widehat{\otimes}\etab)-\frac 12 \chih\wedge(\etab\widehat{\otimes}\etab),\\
\nab_4 K+\trch K=&-\div\beta-\zeta\cdot\beta-2\etab\cdot\beta+\frac 12 \chih\cdot\nab\widehat{\otimes}\etab+\frac 12 \chih\cdot(\etab\widehat{\otimes}\etab)-\frac 12 \trch\div\etab-\frac 12\trch |\etab|^2,\\
\nab_3\sigmac+\frac 32\trchb\sigmac=&-\div ^*\betab+\zeta\wedge\betab-2\eta\wedge
\betab+\frac 12 \chibh\wedge(\nab\widehat{\otimes}\eta)+\frac 12 \chibh\wedge(\eta\widehat{\otimes}\eta),\\
\nab_3 K+\trchb K=&\div\betab-\zeta\cdot\betab+2\eta\cdot\betab+\frac 12 \chibh\cdot\nab\widehat{\otimes}\eta+\frac 12 \chibh\cdot(\eta\widehat{\otimes}\eta)-\frac 12 \trchb\div\eta-\frac 12 \trchb |\eta|^2,\\
\nab_4\betab+\trch\betab=&{\nabla} K +^*{\nabla}\sigmac+ 2\omega\betab +2\chibh\cdot\beta+3(\etab K+^*\etab\sigmac)-\frac 1 2({\nabla}(\chih\cdot\chibh)-^*{\nabla}(\chih\wedge\chibh))\\
&+\frac 14 (\nab\trch \trchb+\trch\nab\trchb)-\f 32(\etab\chih\cdot\chibh-^*\etab\chih\wedge\chibh)+\frac 34 \etab\trch\trchb.
\end{split}
\end{equation}
Notice that we have obtained a system for the renormalized curvature components in which the curvature components $\alpha$ and $\alphab$ do not appear.\footnote{Moreover, compared to the renormalization in \cite{LR}, this system do not contain the terms $\trch|\chibh|^2$ and $\trchb|\chih|^2$ which would be uncontrollable in the context of this paper.}

From now on, we will use capital Latin letters $A\in \{1,2\}$ for indices on the spheres $S_{u,\ub}$ and Greek letters $\mu\in\{1,2,3,4\}$ for indices in the whole spacetime.

\subsection{Schematic Notation}\label{schnot}

We define a schematic notation for the Ricci coefficients according to the estimates that they obey. Introduce the following conventions:\footnote{Notice that this definition is different form that in \cite{LR} since in the context of the present paper, $\trch$ and $\trchb$ verify different bounds compared to \cite{LR}.}
$$\psi\in\{\eta,\etab\},\quad \psi_H\in \{\trch,\chih,\om\},\quad \psi_{\Hb}\in\{\trchb,\chibh,\omegab\}. $$

We will use this schematic notation only in the situations where the exact constant in front of the term {is} irrelevant to the argument. We will denote by $\psi\psi$ (or $\psi\psi_H$, etc) an arbitrary contraction with respect to the metric $\gamma$ and by $\nab\psi$ an arbitrary angular covariant derivative. $\nab^i\psi^j$ will be used to denote the sum of all terms which are products of $j$ factors, such that each factor takes the form $\nab^{i_k}\psi$ and that the sum of all $i_k$'s is $i$, i.e., 
$$\nab^i\psi^j=\displaystyle\sum_{i_1+i_2+...+i_j{=i}}\underbrace{\nab^{i_1}\psi\nab^{i_2}\psi...\nab^{i_j}\psi}_\text{j factors}.$$

We will use brackets to denote terms with one of the components in the brackets. For instance, the notation $\psi(\psi,\psi_H)$ denotes the sum of all terms of the form $\psi\psi$ or $\psi\psi_H$.

In this schematic notation, the Ricci coefficients $\psi_H$ satisfy
$$\nab_3 \psi_H = K+ \nab\psi + \psi\psi+\psi_H \psi_{\Hb}.$$
The Ricci coefficients $\psi_{\Hb}$ similarly obey
$$\nab_4 \psi_{\Hb} = K+ \nab\psi + \psi\psi+\psi_H \psi_{\Hb}.$$
The Ricci coefficients $\psi$ obey either one of the following equations:
$$\nab_3 \psi = \betab + \psi\psi_{\Hb}$$
or
$$\nab_4 \psi = \beta + \psi\psi_H.$$

We also rewrite the Bianchi equations in the schematic notation:
\begin{equation}
\label{eq:null.Bianchi3}
\begin{split}
\nab_3\beta{+}\nabla K  -^*\nabla\sigmac =& \sum_{i_1+i_2=1}\psi_{\Hb}\psi^{i_1} \nab^{i_2}\psi_H+\psi K+\sum_{i_1+i_2=1} \psi^{i_1}\psi\nab^{i_2}\psi\\
\nab_4\sigmac+\div^*\beta=&\psi_H\sigmac+\psi\sum_{i_1+i_2+i_3=1}\psi^{i_1}\nab^{i_2}\psi\nab^{i_3}\psi_H,\\
\nab_4 K{+}\div\beta=&\psi_H K+\psi\sum_{i_1+i_2+i_3=1}\psi^{i_1}\nab^{i_2}\psi\nab^{i_3}\psi_H,\\
\nab_3\sigmac+\div ^*\betab =& \psi_{\Hb}\sigmac+\psi\sum_{i_1+i_2+i_3=1}\psi^{i_1}\nab^{i_2}\psi\nab^{i_3}\psi_{\Hb},\\
\nab_3 K{-}\div\betab=&\psi_{\Hb}K+\psi\sum_{i_1+i_2+i_3=1}\psi^{i_1}\nab^{i_2}\psi\nab^{i_3}\psi_{\Hb},\\
\nab_4\betab{-}\nabla K -^*\nabla\sigmac=&\sum_{i_1+i_2=1}\psi_{H}\psi^{i_1} \nab^{i_2}\psi_{\Hb}+\psi K+\sum_{i_1+i_2=1} \psi^{i_1}\psi\nab^{i_2}\psi.
\end{split}
\end{equation}

\section{Norms}\label{secnorm}
In this section, we define the norms that we will use to control the geometric quantities. We will in particular use the schematic notation defined in Section \ref{schnot}. Our norms will be of the form $L^p_u L^q_{\ub} L^r(S)$, where $L^p_u$ and $L^q_{\ub}$ are defined with respect to the measures $du$ and $d\ub$ respectively and $L^r(S)$ is defined for any tensors $\phi$ on $S_{u,\ub}$ by
$$\|\phi\|_{L^r(S_{u,\ub})}:=\left(\int_{S_{u,\ub}} (\phi_{A_1 A_2 ... A_n} \phi^{A_1 A_2 ...A_n})^{\f r2}\right)^{\f 1r},$$
where the integral is with respect to the volume form induced by $\gamma$.

We define the following norms for the Ricci coefficients $\psi$ {for $p\in [1,\infty]$, $i\in \mathbb N$}:
\begin{equation}\label{O.psi.def}
\mathcal O_{i,p}[\psi]:=||\nab^i\psi||_{L^\infty_uL^\infty_{\ub}L^p(S)}.
\end{equation}
Define the following norms for the Ricci coefficients $\psi_H$ {for $p\in [1,\infty]$, $i\in \mathbb N$}:
\begin{equation}\label{O.psiH.def}
\mathcal O_{i,p}[\psi_H]:=||f(\ub)\nab^i\psi_H||_{L^2_{\ub}L^\infty_uL^p(S)}.
\end{equation}
Similarly, we define the following norms for the Ricci coefficients $\psi_{\Hb}$ {for $p\in [1,\infty]$, $i\in \mathbb N$}:
\begin{equation}\label{O.psiHb.def}
\mathcal O_{i,p}[\psi_{\Hb}]:=||f(u)\nab^i\psi_{\Hb}||_{L^2_{u}L^\infty_{\ub}L^p(S)}.
\end{equation}
{As a shorthand, we define the following norm combining all of the norms above:
$$\mathcal O_{i,p}:=\sum_{\psi\in \{\eta,\etab\}}\mathcal O_{i,p}[\psi]+\sum_{\psi_H\in  \{tr\chi,\chih,\om\}}\mathcal O_{i,p}[\psi_H]+\sum_{\psi_{\Hb}\in \{tr \chib,\chibh,\omb\}}\mathcal O_{i,p}[\psi_{\Hb}].$$}

We make two remarks concerning these norms:
\begin{remark}
While the norms for $\psi_H$ and $\psi_{\Hb}$ are based on $L^2$ in $\ub$ and $u$ respectively, by virtue of the weights $f(\ub)$ and $f(u)$, they actually control the $L^1$ norms. More precisely, since $\int_0^{\ub_*} \frac{1}{f^2(\ub')} d\ub' <\ep^{{2}}$ and $\int_0^{u_*} \frac{1}{f^2(u')} du' <\ep^{{2}}$, by {the} Cauchy-Schwarz {inequality}, we have
$$||\nab^i\psi_H||_{L^1_{\ub}L^\infty_u L^p(S)}\leq C\ep \mathcal O_{i,p}[\psi_H],$$
and 
$$||\nab^i\psi_{\Hb}||_{L^1_{u}L^\infty_{\ub} L^p(S)}\leq C\ep \mathcal O_{i,p}[\psi_{\Hb}].$$
\end{remark}

\begin{remark}\label{order}
The norm $\mathcal O_{i,p}[\psi_H]$ (resp. $\mathcal O_{i,p}[\psi_{\Hb}]$) allows us to first take $L^\infty$ along {the} $u$ direction (resp. $\ub$ direction) before the $L^2$ norm in $\ub$ (resp. $u$) is taken. This is stronger than the norms such that the order is reversed, i.e., we have
$$||f(\ub)\nab^i\psi_H||_{L^\infty_uL^2_{\ub} L^p(S)}\leq C \mathcal O_{i,p}[\psi_H],$$
and 
$$||f(u)\nab^i\psi_{\Hb}||_{L^\infty_{\ub}L^2_{\ub} L^p(S)}\leq C \mathcal O_{i,p}[\psi_{\Hb}].$$
\end{remark}

In addition to the above norms, we need to define norms for the highest derivatives for the Ricci coefficients. Let
\begin{equation} \label{O.42.def}
\begin{split}
\tilde{\mathcal O}_{4,2}:=&||f(\ub)^2\nab^4\trch||_{L^\infty_u L^\infty_{\ub}L^2(S)}+||f(u)^2\nab^4\trchb||_{L^\infty_u L^\infty_{\ub}L^2(S)}\\
&+||f(\ub)\nab^4(\chih,\om)||_{L^\infty_u L^2_{\ub}L^2(S)}+||f(u)\nab^4(\eta,\etab)||_{L^\infty_u L^2_{\ub}L^2(S)}\\
&+||f(u)\nab^4(\chibh,\omb)||_{L^\infty_{\ub} L^2_uL^2(S)}+||f(\ub)\nab^4(\eta,\etab)||_{L^\infty_{\ub} L^2_{u}L^2(S)}.
\end{split}
\end{equation}
\begin{remark}
Here, note that for the norms for $\chih$, $\om$, $\eta$, $\etab$, $\chibh$ and $\om$, $L^\infty$ in $\ub$ (or $u$) is taken after $L^2$ in $u$ (or $\ub$). According to Remark \ref{order}, this is weaker than the $\mathcal O_{i,2}$ norms defined above.
\end{remark}
\begin{remark}
{Notice that the norms for the fourth derivatives of $\eta$ and $\etab$ come with a weight $f(u)$ or $f(\ub)$. This is in contrast to the lower order derivatives for $\eta$ and $\etab$, which can be estimated in $L^\infty_uL^\infty_{\ub}$ without any degeneration. The degeneration here arises from the fact that these higher order derivatives are recovered from the energy estimates for $\nab^3 K$. These energy estimates for $\nab^3 K$, which are derived simultaneously with the estimates for the singular components $\nab^3\beta$ or $\nab^3\betab$, have a degeneration either in $\ub$ or $u$.}
\end{remark}
We also define the curvature norms for the curvature components. {For $i\in \mathbb N$, l}et
\begin{equation}\label{R.def}
\begin{split}
\mathcal R_i:=&||f(\ub)\nab^i\beta||_{L^\infty_u L^2_{\ub}L^2(S)}+||f(u)\nab^i(K,\sigmac)||_{L^\infty_u L^2_{\ub}L^2(S)}\\
&+||f(\ub)\nab^i(K,\sigmac)||_{L^\infty_{\ub} L^2_{u}L^2(S)}+||f(u)\nab^i\betab||_{L^\infty_{\ub} L^2_{u}L^2(S)}.
\end{split}
\end{equation}
As a shorthand, we also let
$$\mathcal R:=\sum_{i\leq 3} \mathcal R_i.$$

Finally, let $\mathcal O_{ini}$ and $\mathcal R_{ini}$ denote the corresponding norms for the initial data, i.e.
\begin{equation*}
\begin{split}
\mathcal O_{ini}:=&\sum_{i\leq 3}\big(||\nab^i\psi||_{L^\infty_{\ub}L^{2}(S_{0,\ub})}+||\nab^i\psi||_{L^\infty_{u}L^{2}(S_{u,0})}\\
&\quad+||f(\ub)\nab^i\psi_H||_{L^2_{\ub}L^{{2}}(S_{0,\ub})}+||f(u)\nab^i\psi_{\Hb}||_{L^2_{u}L^{{2}}(S_{u,0})}\big)\\
&+||f(\ub)^2\nab^4\trch||_{L^\infty_{\ub}L^2(S_{0,\ub})}+||f(u)^2\nab^4\trchb||_{L^\infty_u L^2(S_{u,0})}\\
&+||\nab^4\trchb||_{L^\infty_{\ub}L^2(S_{0,\ub})}+||\nab^4\trch||_{L^\infty_u L^2(S_{u,0})}\\
&+||f(\ub)\nab^4(\chih,\om)||_{L^2_{\ub}L^2(S_{0,\ub})}+||\nab^4(\eta,\etab)||_{L^2_{\ub}L^2(S_{0,\ub})}\\
&+||f(u)\nab^4(\chibh,\omb)||_{L^2_uL^2(S_{u,0})}+||\nab^4(\eta,\etab)||_{L^2_{u}L^2(S_{u,0})}
\end{split}
\end{equation*}
and
\begin{equation*}
\begin{split}
\mathcal R_{ini}:=&{\sum_{i\leq 3}}\left(||f(\ub)\nab^i\beta||_{L^2_{\ub}L^2(S_{0,\ub})}+||\nab^i(K,\sigmac)||_{L^2_{\ub}L^2(S_{0,\ub})}\right.\\
&\left.+||\nab^i(K,\sigmac)||_{L^2_{u}L^2(S_{u,0})}+||f(u)\nab^i\betab||_{L^2_{u}L^2(S_{u,0})}\right).
\end{split}
\end{equation*}

In order to prove Theorem \ref{extthm}, we will establish a priori estimates for the geometric quantities in the above norms: 
\begin{theorem}\label{aprioriestimates}
Assume that the initial data for the characteristic initial value problem satisfy the assumptions of Theorem \ref{extthm} with $\ep$ sufficiently small. Then there exists $B$ depending only on $D$ and $d$ such that
$$\sum_{i\leq 3}\mathcal O_{i,2}+\tilde{\mathcal O}_{4,2}+\mathcal R\leq B.$$
\end{theorem}
In the remainder of the paper, we will focus on the proof of Theorem \ref{aprioriestimates} (after constructing initial data sets in the next section). Standard methods show that Theorem \ref{aprioriestimates} implies Theorem \ref{extthm}. We will omit the details and refer the readers to \cite{Chr, LR} for a proof that the a priori estimates imply the existence theorem.

\begin{remark}\label{rmk.extthm}
The assumptions of {Theorem \ref{extthm}} imply the boundedness of {the following} weighted $L^2$ norms of the curvature components:
$$\sum_{i\leq 3}||f(\ub)\nab^i\beta||_{L^2_{\ub}L^2(S_{0,\ub})}+\sum_{i\leq 3}||{\nab^i(K,\sigmac)}||_{L^2_{\ub}L^2(S_{0,\ub})}\leq \tilde{D},$$
and
$$\sum_{i\leq 3}||f({u})\nab^i\betab||_{L^2_{u}L^2(S_{u,0})}+\sum_{i\leq 3}||{\nab^i(K,\sigmac)}||_{L^2_{{u}}L^2(S_{u,0})}\leq \tilde{D},$$
for some $\tilde{D}$ depending only on $D$ and $d${.} These estimates for $\beta$, $\sigmac$ and $\betab$ follow immediately from the constraint equations on the 2-spheres {(see \eqref{null.str3})}. The bound for $K$ follows after integrating the null Bianchi equations for $K$ on each of the initial null hypersurfaces {(see \eqref{eq:null.Bianchi2})}.\footnote{Notice that it is precisely for the initial bound for $K$ that we require an extra derivative for $\chi$ on $H_0$ (and $\chib$ on $\Hb_0$) in the assumptions of the theorem. This is related to the intrinsic loss of derivatives for the characteristic initial value problem for second order hyperbolic systems (see \cite{M}).} {In particular, the assumptions of Theorem \ref{extthm} imply that
$$\mathcal O_{ini}+\mathcal R_{ini}\leq \tilde{D}.$$
}
\end{remark}

\section{Construction of Initial Data Set}\label{sec.data}

In this section, we construct initial data sets satisfying the assumptions of Theorems \ref{extthm} and \ref{blowupthm}. In particular, we show that the constraint equations can be solved for $|\chih(0,\ub)|\sim (f(\ub))^{-2}$ and $|\chibh(u,0)| \sim (f(u))^{-2}$. Our approach in this section follows closely that of Christodoulou in Chapter 2 of \cite{Chr}.

Assume for simplicity that $S_{0,0}$ is a standard sphere of radius $1$.
Introduce\footnote{While we only write down one coordinate chart, it is implicit that we have two stereographic charts - the north pole chart and the south pole chart. In the following, when we derive the estimates for the geometric quantities, we only prove the bounds in a sufficiently large ball $B_\rho$ in each of these charts.} the standard stereographic coordinates $(\theta^1, \theta^2)$ such that the standard metric $\stackrel{\circ}{\gamma}$ on the sphere takes the form
$$\stackrel{\circ}{\gamma}_{AB}=\frac{\delta_{AB}}{(1+\frac{1}{4}|\theta|^2)^2}.$$

Clearly, it suffices to construct initial data on $H_0$ {(with $0\leq \ub< \ub_*$ for $\ub_*\leq \ep$)}. The construction on $\Hb_0$ is similar.
On $H_0$, {we set $\Om=1$ and therefore $e_4=\frac{\partial}{\partial \ub}$}.
We will construct a metric on $H_0$ in {the $(\ub,\th^1,\th^2)$} coordinates taking the form
\bea
{\gamma}_{AB}=\Phi^2\hat{\gamma}_{AB},{\quad \mbox{where }\hat{\gamma}_{AB}=\frac{m_{AB}}{(1+\frac{1}{4}|\theta|^2)^2}}\label{gamma.data}
\eea
{and} $\det m_{AB} =1 $ and $\Phi \restriction_{S_{0,0}}=1$.
In order to ensure that $m$ satisfies $\det m=1$, we write
$$m=\exp \Psi,$$ 
with $\Psi\in \hat{S}$, where $\hat{S}$ denotes the set of all matrices taking the form
\[ \left( \begin{array}{cc}
a & b  \\
b & -a \end{array} \right).\] 
We will impose upper and lower bounds on $\Psi$. Since there are no smooth globally non-vanishing $\Psi\in \hat{S}$ on the $2$-sphere, we use the convention that $\ls$ denotes that the quantity is bounded above by a uniform constant, while $\sim$ denotes that the quantity is bounded above by a uniform constant, and is bounded below at every $(\th^1,\th^2)$ by a constant depending on $(\th^1,\th^2)$ (where the constant is moreover allowed to vanish at finitely many isolated points). We require $\Psi\in \hat{S}$ {to satisfy}\footnote{{Here and in the rest of this section, we use the notation that $J=(j_1,j_2)\in (\mathbb N\cup\{0\})\times(\mathbb N\cup\{0\})$ is a multi-index and $(\frac{\partial}{\partial\theta})^J=(\frac{\partial}{\partial\theta^1})^{j_1}(\frac{\partial}{\partial\theta^2})^{j_2}$. We moreover denote $|J|=j_1+j_2$.}}
\bea
{\sum_{|J|\leq N} |\big(\frac{\partial}{\partial\theta}\big)^J\Psi | \lesssim 1,}\quad \sum_{|J|\leq N} |\big(\frac{\partial}{\partial\theta}\big)^J\frac{\partial}{\partial\ub} \Psi | \lesssim  f(\ub)^{-2}{,\quad |\frac{\partial}{\partial\ub} \Psi | \sim  f(\ub)^{-2}} \label{duPsibound}
\eea
{for some sufficiently large integer $N$.}
Following \cite{Chr}, we have
\bea
\chih_{AB}=\frac 12 \Phi^2 \frac{\partial}{\partial\ub}\hat{\gamma}_{AB},\quad \trch=\frac{2}{\Phi}\frac{\partial\Phi}{\partial\ub}.\label{chihdef}
\eea
We can also derive that
$$|\chih|_\gamma^2=\frac 14 ({\hat\gamma}^{-1})^{AC}({\hat\gamma}^{-1})^{BD}\frac{\partial}{\partial \ub}\hat\gamma_{AB}\frac{\partial}{\partial \ub}\hat\gamma_{CD}.$$
Thus by \eqref{duPsibound}, we have
\bea
|\chih|_\gamma^2 \sim f(\ub)^{-4}.\label{chihbound}
\eea
In particular, this implies that the requirement in Theorem \ref{blowupthm} is satisfied if $\int_0^{\ub_*} f(\ub)^{{-}4} d\ub = \infty$.
By the equation 
$$\Ls_{\frac{\partial}{\partial\ub}}\trch=-\frac 12 (\trch)^2 -|\chih|^2,$$
$\Phi$ can be solved from the ODE
\bea
\frac{\partial^2\Phi}{\partial\ub^2}+\frac 1 8 (({\hat\gamma}^{-1})^{AC}({\hat\gamma}^{-1})^{BD}\frac{\partial}{\partial \ub}\hat\gamma_{AB}\frac{\partial}{\partial \ub}\hat\gamma_{CD})\Phi=0.\label{PhiODE}
\eea
We prescribe $\trch$ on $S_{0,0}$ to obey the initial conditions
\bea
{\Phi\restriction_{S_{0,0}}=1,\quad\frac{\partial\Phi}{\partial\ub}\restriction_{S_{0,0}}=\frac 12 \trch\restriction_{S_{0,0}}\ls 1.\label{trch.IC}}
\eea
Finally, we prescribe $\zeta$ on $S_{0,0}$ such that
\bea
\sum_{|J|\leq N-1}|\big(\frac{\partial}{\partial \th}\big)^J\zeta|_\gamma^2 \ls 1.\label{zetabound}
\eea
We check that th{ese} initial data obey all the estimates required by Theorem \ref{extthm}:

\noindent{\bf Estimates for $\nab^i\chi$ and the metric}

To satisfy the upper bounds in Theorem \ref{extthm}, we need to show that
\bea
\sum_{i\leq N}|\nab^i\chi|_{\gamma}(0,\ub)\lesssim f(\ub)^{-2}\label{psiH.desired}
\eea
We will show the estimates separately for $\trch$ and $\chih$. By \eqref{chihbound}, \eqref{psiH.desired} {holds for $\chih$ when $i=0$}. To derive this bound for $\trch$, notice that by the ODE \eqref{PhiODE} for $\Phi$, the initial conditions {\eqref{trch.IC}}, and the bound {\eqref{chihbound}} for $|\chih|^2$, we have
\bea
\frac 12\leq \Phi\leq 1\label{Phi.upper.lower}
\eea
and 
$$|\frac{\partial\Phi}{\partial\ub}|\lesssim 1+\int_0^{\ub} f(\ub')^{-4} d\ub'\leq  1+f(\ub)^{-2}\int_0^{\ub_*} f(\ub')^{-2} d\ub'\leq 1+\ep^2 f(\ub)^{-2}$$
{for $\ep$ sufficiently small. In the above estimate, we have used $\int_0^{\ub_*} f(\ub')^{-2}\, d\ub'\leq \ep^2$.} By {\eqref{chihdef}}, we thus have
$$|\trch|\ls f(\ub)^{-2}.$$
We now move on to control the angular derivatives of $\chi$. By \eqref{duPsibound},
$$\sum_{|J|\leq N}| \big(\frac{\partial}{\partial\theta}\big)^J\frac{\partial}{\partial\ub} m_{AB} | \ls f(\ub)^{-2}.$$
Using this bound and commuting the ODE \eqref{PhiODE} with $\frac{\partial}{\partial \th}$, we also have that for up to $N$ coordinate angular derivatives $\frac{\partial}{\partial \th}$,
\bea
\sum_{|J|\leq N}|\big(\frac{\partial}{\partial \th}\big)^J\Phi| \lesssim 1.\label{phiderbound}
\eea
This implies {via \eqref{gamma.data} and \eqref{duPsibound}} that the metric $\gamma$ obeys the bounds
\bea\label{metricbound}
\displaystyle\sum_{|J|\leq N}|\big(\frac{\partial}{\partial \th}\big)^J\gamma_{AB}| \lesssim 1,\quad \sum_{|J|\leq N}|\big(\frac{\partial}{\partial \th}\big)^J(\gamma^{-1})^{AB}| \lesssim 1.
\eea
Together with \eqref{duPsibound} and \eqref{chihdef}, \eqref{phiderbound} implies
\bea
\sum_{|J|\leq N}|\big(\frac{\partial}{\partial \th}\big)^J\chih| \lesssim f(\ub)^{-2}.\label{chihderbound}
\eea
By {\eqref{chihdef}}, we also have
\bea
\sum_{|J|\leq N}|\big(\frac{\partial}{\partial \th}\big)^J\trch|\ls  f(\ub)^{-2}.\label{trchderbound}
\eea
Finally, we notice that by \eqref{metricbound}, the angular \emph{covariant} derivatives of $\trch$ and $\chih$ {can be controlled by} the angular coordinate derivatives of $\trch$ and $\chih$. Therefore, \eqref{psiH.desired} follows from \eqref{chihderbound} and \eqref{trchderbound}.

{\noindent{\bf Estimates for $\nab^i K$}}
{
To control $\nab^i K$, we simply notice that by \eqref{metricbound}, we have
$$\sum_{i\leq N-2}|\nab^i K|_{\gamma}\lesssim 1.$$
}

\noindent{\bf Estimates for $\nab^i\zeta$}

On $H_0$, since $\Omega=1$, $\eta=\zeta$. Thus combining the transport equation for $\zeta$ in \eqref{null.str2} and the Codazzi equation for $\beta$ in \eqref{null.str3} and rewriting in $\Ls$ (instead of $\nab_4$), we have
$$\Ls_{\frac{\partial}{\partial\ub}} \zeta+\trch\zeta=\div\chi-\nab\trch.$$
Recall from \eqref{zetabound} that the initial data for $\zeta$ and its angular derivatives are bounded. Therefore, by the estimates for $\trch$ and $\chih$ (and their angular derivatives) above, we have
$$\sum_{|J|\leq N-1}|\big(\frac{\partial}{\partial \th}\big)^J\zeta|\ls 1.$$
The bounds for the metric and Christoffel symbols on the sphere imply
$$\sum_{j\leq N-1}||\nab^j\zeta||_{L^\infty_{\ub}L^\infty(S_{0,\ub})}\ls 1$$
as desired.

\noindent{\bf Estimates for $\nab^i \trchb$}

Similar to $\zeta$, $\trchb$ obeys a transport equations along the null generators of $H_0$. More precisely, \eqref{null.str2} and the Gauss equation in \eqref{null.str3} imply that
$$\Ls_{\frac{\partial}{\partial\ub}} \trchb+\trch\trchb=-2K-2\div\zeta+2|\zeta|^2.$$
Thus, the previous estimates imply
$$\sum_{j\leq N-2}||\nab^j\trchb||_{L^\infty_{\ub}L^\infty(S_{0,\ub})}\ls 1$$

{Now, combining all the estimates that we have obtained so far, requiring $f$ to satisfy
$$\int_0^{\ub_*} f(\ub)^{-4} d\ub=\infty$$
and taking $N$ to sufficiently large, we have thus constructed initial data set on $H_0$ that obeys the assumptions of Theorems \ref{extthm} and \ref{blowupthm} on $H_0$. As mentioned above, it is easy to construct initial data set analogously on $\Hb_0$ so that the full set of assumptions of Theorems \ref{extthm} and \ref{blowupthm} are satisfied.}

\section{The Preliminary Estimates}\label{secbasic}

We now turn to the proof Theorem \ref{aprioriestimates}, which will form the content of Sections \ref{secbasic}-\ref{seccurv}. In this section, we derive the necessary preliminary estimates. In Section \ref{secRicci} (see Proposition \ref{Ricci}), we will prove the bound
$$\sum_{i\leq 3}\mathcal O_{i,2}\leq C(\mathcal O_{ini});$$
in Section \ref{secRicci32} (see Proposition \ref{Ricci32}), we will prove 
$$\tilde{\mathcal O}_{4,2}\leq C(\mathcal O_{ini})(1+\mathcal R);$$
and in Section \ref{seccurv} (see Proposition \ref{R.final}), we will derive the estimate
$$\mathcal R\leq C(\mathcal O_{ini},\mathcal R_{ini}).$$
Combining these estimates then imply the conclusion of Theorem \ref{aprioriestimates}.

{We now begin with the preliminary estimates.} All estimates in this section will be proved under the following bootstrap assumption:

\begin{equation}\tag{A1}\label{BA1}
\sum_{i\leq 1}\mathcal O_{i,\infty}+\sum_{i\leq 2}\mathcal O_{i,4}+\sum_{i\leq 3}\mathcal O_{i,2}\leq \Delta_1,
\end{equation}
where $\Delta_1$ is a constant that will be chosen later.

\subsection{Estimates for Metric Components}\label{metric}
We first show that we can control $\Omega$ under the bootstrap assumption (\ref{BA1}):
\begin{proposition}\label{Omega}
There exists $\epsilon_0=\epsilon_0(\Delta_1)$ such that for every $\epsilon\leq\epsilon_0$,
$$\frac 12\leq \Omega\leq 2.$$
Moreover, $\Omega$ is continuous up to $u=u_*$ and $\ub=\ub_*$.
\end{proposition}
\begin{proof}
Consider the equation
\begin{equation}\label{Omegatransport}
 \omega=-\frac{1}{2}\nabla_4\log\Omega=\frac{1}{2}\Omega\nabla_4\Omega^{-1}=\frac{1}{2}\frac{\partial}{\partial \ub}\Omega^{-1}.
\end{equation}
Fix $\ub$. Notice that both $\omega$ and $\Omega$ are scalars and therefore the $L^\infty$ norm is independent of the metric. We can integrate equation (\ref{Omegatransport}) using the fact that $\Omega^{-1}=1$ on $\Hb_0$ to obtain
$$||\Omega^{-1}-1||_{L^\infty(S_{u,\ub})}\leq C\int_0^{\ub}||\omega||_{L^\infty(S_{u,\ub'})}d\ub'\leq C||f(\ub)^{-1}||_{L^2_{\ub}}||f(\ub)\omega||_{L^\infty_uL^2_{\ub}L^\infty(S)}\leq C\Delta_1\epsilon.$$
This implies both the upper and lower bounds for $\Omega$ for sufficiently small $\epsilon$. To show continuity, take a sequence of points $(u_n, \ub_n, \th^1_n, \th^2_n)$ such $u_n \to u_\infty$, $\ub_n\to \ub_\infty$, $\th_n^1\to \th_\infty^1$ and $\th_n^2\to \th_\infty^2$. Then
\begin{equation*}
\begin{split}
&|\Omega^{-1}(u_n,\ub_n,\th_n^1,\th_n^2)-\Omega^{-1}(u_m,\ub_m,\th_m^1,\th_m^2)|\\
\leq & |\Omega^{-1}(u_n,\ub_n,\th_n^1,\th_n^2)-\Omega^{-1}(u_n,\ub_n,\th_m^1,\th_m^2)|+|\Omega^{-1}(u_n,\ub_n,\th_m^1,\th_m^2)-\Omega^{-1}(u_n,\ub_m,\th_m^1,\th_m^2)|\\
&+|\Omega^{-1}(u_n,\ub_m,\th_m^1,\th_m^2)-\Omega^{-1}(u_m,\ub_m,\th_m^1,\th_m^2)|\\
\leq &C||\nab\log\Omega||_{L^\infty(S_{u_n,\ub_n})}\mbox{dist}_{S_{u_n,\ub_n}}(\th_n,\th_m)+2|\int_{\ub_n}^{\ub_m} ||\om||_{L^\infty(S_{u_n,\ub'})} d\ub'|\\
&+2|\int_{u_n}^{u_m} ||\omb||_{L^\infty(S_{u',\ub_m})} du'|.
\end{split}
\end{equation*}
Since by the bootstrap assumption \eqref{BA1}, $\nab\log\Omega=\frac 12(\eta+\etab)$ is uniformly bounded, $||\om||_{L^\infty(S_{u,\ub})}$ is uniformly integrable in $\ub$ for all $u$ and $||\omb||_{L^\infty(S_{u,\ub})}$ is uniformly integrable in $u$ for all $\ub$, the right hand side can be made arbitrarily small by taking $n,m \geq N$ for $N$ sufficiently large. The conclusion thus follows.
\end{proof}

We then show that we can control $\gamma$ under the bootstrap assumption (\ref{BA1}):
\begin{proposition}\label{gamma}
There exists $\ep_0=\ep_0(\Delta_1)$ such that for $\ep\leq \ep_0$, in the $(u,\ub,\th^1,\th^2)$ coordinate system, we have
$$c\leq \det\gamma\leq C,\quad |\gamma_{AB}|,|(\gamma^{-1})^{AB}|\leq C{,}$$
{where the constants depend only on $d$ and $D$.}
Moreover, $\gamma$ remains continuous up to $u=u_*$ and $\ub=\ub_*$.
\end{proposition}
\begin{proof}
{We first prove the bound for $\gamma$ on the initial hypersurface $\Hb_0$. Using
$$\Ls_{\Lb}\gamma=2\Omega\chib,$$
we get\footnote{Note that $b^A=0$ on $\Hb_0$.} 
$$\frac{\partial}{\partial u}\gamma_{AB}=2\Omega\chib_{AB},\quad \frac{\partial}{\partial u}\log(\det\gamma)=\Omega\trchb$$
on $\Hb_0$.
We therefore have
\begin{equation}\label{detgaper}
|\frac{\det\gamma(u,0)}{\det\gamma(0,0)}|\leq C\exp(\int_0^{u}|\trchb|(u',0)\,du')\leq C(D).
\end{equation}
This implies that the $\det \gamma$ is bounded above and below. Let $\Lambda$ be the larger eigenvalue of $\gamma$. Clearly,
\begin{equation}\label{La}
\Lambda\leq C\sup_{A,B=1,2}|\gamma_{AB}|,\quad \sum_{A,B=1,2}|\chib_{AB}|^2\leq C\Lambda^2 ||\chib||_{L^\infty(S_{u,\ub})}^2.
\end{equation}
Then
\begin{equation}\label{gamma.diff}
|\gamma_{AB}(u,0)-(\gamma)_{AB}(0,0)|\leq C\ep(\sup_{u'\leq u}\Lambda)(\int_0^{u}f(u')^2||\chib||_{L^\infty(S_{u',0})}^2du')^{\frac 12}\leq C(D)(\sup_{u'\leq u}\Lambda)\ep.
\end{equation}
Using the first upper bound in (\ref{La}), we thus obtain the upper bound for $|\gamma_{AB}|$ after choosing $\ep$ to be sufficiently small. The upper bound for $|(\gamma^{-1})^{AB}|$ follows from the upper bound for $|\gamma_{AB}|$ and the lower bound for $\det\gamma$.

Now, in order to obtain the bounds for $\gamma_{AB}$ in the spacetime, we argue similarly but using the propagation equation in the $\ub$ direction and compare $\gamma(u,\ub)$ with $\gamma(u,0)$. Here, we use bootstrap assumption \eqref{BA1} instead of the assumptions on the initial data. More precisely, we have
\begin{equation}\label{1st.var}
\frac{\partial}{\partial \ub}\gamma_{AB}=2\Omega\chi_{AB},\quad \frac{\partial}{\partial \ub}\log(\det\gamma)=\Omega\trch.
\end{equation}
We then derive as above that
\begin{equation*}
|\frac{\det\gamma(u,\ub)}{\det\gamma(u,0)}|\leq C\exp(C\Delta_1\epsilon),\quad |\gamma_{AB}(u,\ub)-\gamma_{AB}(u,0)|\leq  C(\sup_{\substack{u'\leq u\\ \ub'\leq \ub}}\Lambda)\Delta_1\epsilon,
\end{equation*}
where $\Lambda$ is the larger eigenvalue for $\gamma_{AB}$. As before, we thus obtain the upper bounds for $|\gamma_{AB}|$ and $|(\gamma^{-1})^{AB}|$. Finally, the continuity of $\gamma$ up to the boundary follows as in the proof of continuity for $\Omega$ in Proposition \ref{Omega}.
} 
\end{proof}

With the estimates on $\gamma$, it follows that the $L^p$ norms defined with respect to the metric and the $L^p$ norms defined with respect to the coordinate system are equivalent.
\begin{proposition}\label{eqnorm}
Given a covariant tensor $\phi_{A_1...A_r}$ on $S_{u,\ub}$, we have
$$\int_{S_{u,\ub}} \langle\phi,\phi\rangle_{\gamma}^{p/2} \sim \sum_{A_i=1,2}\iint |\phi_{A_1...A_r}|^p \sqrt{\det\gamma} d\th^1 d\th^2.$$
\end{proposition}
We can also bound $b$ under the bootstrap assumption, thus controlling the full spacetime metric: 
\begin{proposition}\label{b}
In the {coordinate} system $(u,\ub,\th^1,\th^2)$,
$$|b^A|\leq C\Delta_1\epsilon.$$
Moreover, $b^A$ is continuous up to $u=u_*$ and $\ub=\ub_*$.
\end{proposition}
\begin{proof}
$b^A$ satisfies the equation
\begin{equation}\label{btrans}
\frac{\partial b^A}{\partial \ub}=-4\Omega^2\zeta^A.
\end{equation}
This can be derived from 
$$[L,\Lb]=\frac{\partial b^A}{\partial \ub}\frac{\partial}{\partial \th^A}.$$
Now, integrating \eqref{btrans} and using Proposition \ref{eqnorm} gives the bound on $b$. Continuity of $b$ up to the boundary follows as in the proof of Proposition \ref{Omega}.
\end{proof}

\subsection{Estimates for Transport Equations}\label{transportsec}
In this subsection, we prove general propositions for obtaining bounds from the covariant null transport equations. Such estimates require the integrability of $\trch$ and $\trchb$, which is consistent with our bootstrap assumption (\ref{BA1}). This will be used in the following sections to derive some estimates for the Ricci coefficients and the null curvature components from the null structure equations and the null Bianchi equations respectively. Below, we state two propositions which provide $L^p$ estimates for general quantities satisfying transport equations either in the $e_3$ or $e_4$ direction.
\begin{proposition}\label{transport}
There exists $\epsilon_0=\epsilon_0(\Delta_1)$ such that for all $\epsilon \leq \epsilon_0$ and for every $2\leq p<\infty$, we have
\[
 ||\phi||_{L^p(S_{u,\ub})}\leq C(||\phi||_{L^p(S_{u,\ub'})}+\int_{\ub'}^{\ub} ||\nabla_4\phi||_{L^p(S_{u,\ub''})}d{\ub''}),
\]
\[
 ||\phi||_{L^p(S_{u,\ub})}\leq C(||\phi||_{L^p(S_{u',\ub})}+\int_{u'}^{u} ||\nabla_3\phi||_{L^p(S_{u'',\ub})}d{u''})
\]
for any {tensor $\phi$ tangential to the spheres $S_{u,\ub}$}.
\end{proposition}

\begin{proof}

The following identity\footnote{Here, $\f{\rd}{\rd\ub}$ on the left hand side is to be understood as the coordinate vector field in the $(u,\ub)$-plane. Similarly for $\f{\rd}{\rd u}$ below.} holds for any scalar $f$:
\[
 \frac{\rd}{\rd\ub}\int_{\mathcal S_{u,\ub}} f=\int_{\mathcal S_{u,\ub}} \Omega\left(e_4(f)+ \trch f\right).
\]
Similarly, we have
\[
 \frac{\rd}{\rd u}\int_{\mathcal S_{u,\ub}} f=\int_{\mathcal S_{u,\ub}} \Omega\left(e_3(f)+ \trchb f\right).
\]
Hence, taking $f=|\phi|_{\gamma}^p$, we have
\begin{equation}\label{Lptransport}
\begin{split}
 ||\phi||^p_{L^p(S_{u,\ub})}=&||\phi||^p_{L^p(S_{u,\ub'})}+\int_{\ub'}^{\ub}\int_{S_{u,\ub''}} p|\phi|^{p-2}\Omega\left(\langle\phi,\nabla_4\phi\rangle_\gamma+ \frac{1}{p}\trch |\phi|^2_{\gamma}\right)d{\ub''}{,}\\
 ||\phi||^p_{L^p(S_{u,\ub})}=&||\phi||^p_{L^p(S_{u',\ub})}+\int_{u'}^{u}\int_{S_{u'',\ub}} p|\phi|^{p-2}\Omega\left(\langle\phi,\nabla_3\phi\rangle_\gamma+ \frac{1}{p}\trchb |\phi|^2_{\gamma}\right)d{u''}{.}
\end{split}
\end{equation}
The bootstrap assumption (\ref{BA1}) implies that $\trch$ and $\trchb$ are integrable (and in fact it also implies that $||\trch||_{L^1_{\ub}L^\infty_u L^\infty(S)}$ and $||\trchb||_{L^1_{u}L^\infty_{\ub} L^\infty(S)}$ are small after choosing $\ep$ to be small depending on $\Delta_1$). Thus the proposition can be proved by using H\"older's inequality and Gronwall's inequality, together with the bound for $\Omega$ given in Proposition \ref{Omega}.
\end{proof}
We also have the following bounds for the $p=\infty$ case by integrating along the integral curves of $e_3$ and $e_4$:
\begin{proposition}\label{transportinfty}
There exists $\epsilon_0=\epsilon_0(\Delta_1)$ such that for all $\epsilon \leq \epsilon_0$, we have
\[
 ||\phi||_{L^\infty(S_{u,\ub})}\leq C\left(||\phi||_{L^\infty(S_{u,\ub'})}+\int_{\ub'}^{\ub} ||\nabla_4\phi||_{L^\infty(S_{u,\ub''})}d{\ub''}\right)
\]
\[
 ||\phi||_{L^\infty(S_{u,\ub})}\leq C\left(||\phi||_{L^\infty(S_{u',\ub})}+\int_{u'}^{u} ||\nabla_3\phi||_{L^\infty(S_{u'',\ub})}d{u''}\right)
\]
for any {tensor $\phi$ tangential to the spheres $S_{u,\ub}$}.
\end{proposition}
\begin{proof}
This follows simply from integrating along the integral curves of $L$ and $\Lb$, and the estimate on $\Omega$ in Proposition \ref{Omega}.
\end{proof}

\subsection{Sobolev Embedding}\label{Embedding}
Using the {estimates for the metric $\gamma$ in Proposition \ref{gamma}, we have the following Sobolev embedding theorem:}
\begin{proposition}\label{Sobolev}
There exists $\epsilon_0=\epsilon_0(\Delta_1)$ such that as long as $\epsilon\leq \epsilon_0$, we have
$$||\phi||_{L^4(S_{u,\ub})}\leq C\sum_{i=0}^1||\nabla^i\phi||_{L^2(S_{u,\ub})}$$
and
$$||\phi||_{L^\infty(S_{u,\ub})}\leq C\left(||\phi||_{L^2(S_{u,\ub})}+||\nabla\phi||_{L^4(S_{u,\ub})}\right)$$
for any {tensor $\phi$ tangential to the spheres $S_{u,\ub}$}. Combining the above estimates, we also have
$$||\phi||_{L^\infty(S_{u,\ub})}\leq C\sum_{i=0}^2||\nabla^i\phi||_{L^2(S_{u,\ub})}. $$
\end{proposition}
\begin{proof}
{By \eqref{gamma.diff} in the proof of Proposition \ref{gamma}, $|\gamma_{AB}(u,\ub)-\gamma_{AB}(0,0)|$ can be made arbitrarily small by choosing $\ep$ to be small. Therefore, the isoperimetric constant 
$$I(S_{u,\ub})=\sup_U\f{\min\{\mbox{Area}(U),\mbox{Area}(U^c)\}}{\mbox{Perimeter}(\partial U)}$$ 
on every sphere $S_{u,\ub}$ is controlled\footnote{This argument is standard. We refer the readers for instance to Lemma 5.4 in \cite{Chr}.} up to a constant factor by the corresponding isoperimetric constant on $S_{0,0}$. Once the isoperimetric constants are uniformly controlled, the Sobolev embedding theorem follows from Lemmas 5.1 and 5.2 in \cite{Chr} and the fact that the volume of $S_{u,\ub}$ is bounded uniformly above and below.
}
\end{proof}

\subsection{Commutation Formulae}\label{commutation}
We have the following formula from Lemma 7.3.3 in \cite{CK}:
\begin{proposition}
The commutator $[\nabla_4,\nabla]$ acting on a {rank $r$} {tensor $\phi$ tangential to the spheres $S_{u,\ub}$} is given by
\begin{equation*}
 \begin{split}
[\nabla_4,\nabla_B]\phi_{A_1...A_r}=&(\nabla_B\log\Omega)\nabla_4\phi_{A_1...A_r}-(\gamma^{-1})^{CD}\chi_{BD}\nabla_C\phi_{A_1...A_r} \\
&+\sum_{i=1}^r ((\gamma^{-1})^{CD}\chi_{A_iB}\etab_{D}-(\gamma^{-1})^{CD}\chi_{BD}\etab_{A_i}+\eps_{A_i}{ }^C{ }^*\beta_B)\phi_{A_1...\hat{A_i}C...A_r}.
 \end{split}
\end{equation*}
Similarly, the commutator $[\nab_3,\nab]$ is given by
\begin{equation*}
 \begin{split}
[\nabla_3,\nabla_B]\phi_{A_1...A_r}=&(\nabla_B\log\Omega)\nabla_3\phi_{A_1...A_r}-(\gamma^{-1})^{CD}\chib_{BD}\nabla_C\phi_{A_1...A_r} \\
&+\sum_{i=1}^r ((\gamma^{-1})^{CD}\chib_{A_iB}\eta_{D}-(\gamma^{-1})^{CD}\chib_{BD}\eta_{A_i}-\eps_{A_i}{ }^C{ }^*\betab_B)\phi_{A_1...\hat{A_i}C...A_r}.
 \end{split}
\end{equation*}
\end{proposition}

Recall the schematic notation 
$$\psi\in\{\eta,\etab\},\quad \psi_H\in \{\trch,\chih,\om\},\quad \psi_{\Hb}\in\{\trchb,\chibh,\omegab\}. $$
By induction and the schematic Codazzi equations
$$\beta=\nabla\chi+\psi\chi=\nab\psi_H+\psi\psi_H,\quad {\betab}=\nabla\chib+\psi\chib=\nab\psi_{\Hb}+\psi\psi_{\Hb},$$
we get the following schematic formula for repeated commutations (see \cite{LR}):
\begin{proposition}\label{repeated.com}
Suppose $\nabla_4\phi=F_0$ for some tensors $\phi$ and $F_0$. Let $F_i$ be the tensor defined by $\nabla_4\nabla^i\phi=F_i$.
Then
\begin{equation*}
\begin{split}
F_i\sim &\sum_{i_1+i_2+i_3=i}\nabla^{i_1}\psi^{i_2}\nabla^{i_3} F_0+\sum_{i_1+i_2+i_3+i_4=i}\nabla^{i_1}\psi^{i_2}\nabla^{i_3}\psi_H\nabla^{i_4} \phi.
\end{split}
\end{equation*}
Similarly, suppose $\nabla_3\phi=G_{0}$ for some tensors $\phi$ and $G_0$. Let $G_i$ be the tensor defined by $\nabla_3\nabla^i\phi=G_{i}$.
Then
\begin{equation*}
\begin{split}
G_{i}\sim &\sum_{i_1+i_2+i_3=i}\nabla^{i_1}\psi^{i_2}\nabla^{i_3} G_{0}+\sum_{i_1+i_2+i_3+i_4=i}\nabla^{i_1}\psi^{i_2}\nabla^{i_3}\psi_{\Hb}\nabla^{i_4} \phi.
\end{split}
\end{equation*}
\end{proposition}

\subsection{General Elliptic Estimates for Hodge Systems}\label{elliptic}
We recall the definition of the divergence and curl of a symmetric covariant tensor of rank $r+1$:
$$(\div\phi)_{A_1...A_r}=\nabla^B\phi_{BA_1...A_r},$$
$$(\curl\phi)_{A_1...A_r}=\eps^{BC}\nabla_B\phi_{CA_1...A_r},$$
where $\eps$ is the volume form associated to the metric $\gamma$.
Recall also that the trace is defined to be
$$(tr\phi)_{A_1...A_{r-1}}=(\gamma^{-1})^{BC}\phi_{BCA_1...A_{r-1}}.$$
The following elliptic estimate is standard (see for example Lemmas 2.2.2, 2.2.3 in \cite{CK} or Lemmas 7.1, 7.2, 7.3 in \cite{Chr}):
\begin{proposition}\label{ellipticthm}
Let $\phi$ be a symmetric $r$ covariant tensor on a 2-sphere $(\mathbb S^2,\gamma)$ satisfying
$$\div\phi=f,\quad \curl\phi=g,\quad tr\phi=h.$$
Suppose also that
$$\sum_{i\leq 2}||\nabla^i K||_{L^2(S)}< \infty.$$
Then for $i\leq 4$, there exists a constant $C_E$ depending only on $\displaystyle\sum_{i\leq 2}||\nabla^i K||_{L^2(S)}$ such that 
$$||\nabla^{i}\phi||_{L^2(S)}\leq C_E(\sum_{j=0}^{i-1}(||\nabla^{j}f||_{L^2(S)}+||\nabla^{j}g||_{L^2(S)}+||\nabla^{j}h||_{L^2(S)}+||{\nab^j}\phi||_{L^2(S)})).$$
\end{proposition}
For the special case that $\phi$ is a symmetric traceless 2-tensor, we only need to know its divergence:
\begin{proposition}\label{elliptictraceless}
Suppose $\phi$ is a symmetric traceless 2-tensor satisfying
$$\div\phi=f.$$
Suppose moreover that $$\sum_{i\leq 2}||\nabla^i K||_{L^2(S)}< \infty.$$
Then for $i\leq 4$, there exists a constant $C_E$ depending only on $\displaystyle\sum_{i\leq 2}||\nabla^i K||_{L^2(S)}$ such that 
$$||\nabla^{i}\phi||_{L^2(S)}\leq C_E(\sum_{j=0}^{i-1}(||\nabla^{j}f||_{L^2(S)}+||{\nab^j}\phi||_{L^2(S)})).$$
\end{proposition}
\begin{proof}
This follows from Proposition \ref{ellipticthm} and the fact that
$$\curl\phi=^*f.$$
\end{proof}

\section{Estimates for the Ricci Coefficients via Transport Equations}\label{secRicci}

In this section, we prove $L^2$ estimates for the Ricci coefficients and their first, second and third derivatives. We will assume bounds for $\mathcal R$ and $\tilde{\mathcal O}_{4,2}$ and show that for $\epsilon_0$ chosen to be sufficiently small, $\displaystyle\sum_{i\leq 3}\mathcal O_{i,2}$ is likewise bounded. In order to achieve this, we continue to work under the bootstrap assumption (\ref{BA1}) and will show that the constant in \eqref{BA1} can in fact be improved (see Proposition \ref{Ricci}).

Recall that we will use the following notation: $\psi\in\{\eta,\etab\}$, $\psi_{\Hb}\in\{\trchb,\chibh,\omegab\}$ and $\psi_H\in\{\trch,\chih,\omega\}$.

We first show bounds for $\psi$. 

\begin{proposition}\label{psi}
Assume
$$\mathcal R <\infty.$$
Then there exists $\epsilon_0=\epsilon_0(\Delta_1,\mathcal R)$ such that whenever $\epsilon\leq \epsilon_0$, 
\[
 \sum_{i\leq 3}\mathcal O_{i,2}[\psi]\leq C(\mathcal O_{ini}),
\]
i.e., the bounds depends only on the initial data norm $\mathcal O_{ini}$. In particular, $C(\mathcal O_{ini})$ is independent of $\Delta_1$.
\end{proposition}
\begin{proof}
We first estimate $\eta$, the estimates for $\etab$ {are} similar after we replace $u$ with $\ub$ and $3$ with $4$. Using the null structure equations, we have a schematic equation of the type
$$\nabla_{4}\eta=\beta+\psi_H\psi.$$
We also commute the null structure equations with angular derivatives to get
\begin{equation}\label{etaeqn}
\nabla_{4}\nabla^i\eta=\sum_{i_1+i_2+i_3=i}\nabla^{i_1}\psi^{i_2}\nabla^{i_3}\beta+\sum_{i_1+i_2+i_3+i_4=i}\nabla^{i_1}\psi^{i_2}\nabla^{i_3}\psi\nabla^{i_4}\psi_H.
\end{equation}

By Proposition \ref{transport}, in order to estimate $||\nabla^i\eta||_{L^\infty_uL^\infty_{\ub}L^2(S)}$, 
it suffices to estimate the initial data and 
the $||\cdot||_{L^\infty_uL^1_{\ub}L^2(S)}$ norm of the right hand side \eqref{etaeqn}. {Using the} bootstrap assumption{, we will show that the} right hand side is bounded in a weighted $L^2_{\ub}$ norm{. This in turns imply via an application of the Cauchy-Schwarz inequality that the $L^1_{\ub}$ norm is also bounded. We now turn to the details.}

We first estimate the curvature term
$$\sum_{i_1+i_2+i_3\leq 3}\nab^{i_1}\psi^{i_2}\nabla^{i_3}\beta.$$
For the terms such that at most $1$ derivative falling on $\psi$, the {bootstrap} assumption {\eqref{BA1}} allows us to control $\displaystyle\sum_{i\leq 1}\|\nab^i\psi\|_{L^\infty_uL^\infty_{\ub}L^\infty(S)}$ by $\Delta_1$. We then need to control $\displaystyle\sum_{i\leq 3}\nab^i\beta$ in $L^\infty_uL^1_{\ub}L^2(S)$. By {the Cauchy-Schwarz inequality}, since {the $L^2_{\ub}$ norm of} $f(\ub)^{-1}$ is smaller than $\ep$, we can bound this by $\displaystyle\sum_{i\leq 3}\nab^i\beta$ in the weighted {$L^2$} norms. More precisely, we have
\begin{equation}\label{psi.1}
\begin{split}
&||\sum_{i_1\leq 1,i_2\leq 3,i_3\leq 3}\nab^{i_1}\psi^{i_2}\nabla^{i_3}\beta||_{L^\infty_{u}L^1_{\ub}L^2(S)} \\
\leq& C(\sum_{i_1\leq 1,i_2\leq 3}||\nab^{i_1}\psi||^{i_2}_{L^\infty_{u}L^\infty_{\ub}L^\infty(S)})(\sum_{i_3\leq 3}||f(\ub)\nabla^{i_3}\beta||_{L^\infty_{u}L^1_{\ub}L^2(S)})\\
{\leq }&C(\sum_{i_1\leq 1,i_2\leq 3}||\nab^{i_1}\psi||^{i_2}_{L^\infty_{u}L^\infty_{\ub}L^\infty(S)})(\sum_{i_3\leq 3}||f(\ub)\nabla^{i_3}\beta||_{L^\infty_{u}L^2_{\ub}L^2(S)})||f(\ub)^{-1}||_{L^\infty_{u}L^2_{\ub}L^\infty(S)}\\
\leq &C\ep(1+\Delta_1)^3\mathcal R{.}
\end{split}
\end{equation}
For the term where exactly $2$ derivatives fall on $\psi$ (notice that this is the highest number of derivatives that can fall on $\psi$), we control {$\nab^2\psi$} in $L^\infty_uL^\infty_{\ub}L^2(S)$ by $\Delta_1$ (using \eqref{BA1}). Thus we are left with $\beta$ in $L^\infty_uL^1_{\ub}L^\infty(S)$. By Sobolev embedding (Proposition {\ref{Sobolev}}), this can be bounded by $\displaystyle\sum_{i\leq 3}\|\nab^i\beta \|_{L^\infty_u L^1_{\ub}L^2(S)}$, which in turn can be controlled by $\mathcal R$ after applying the Cauchy-Schwarz inequality as in \eqref{psi.1}. More precisely,
\begin{equation}\label{psi.2}
\begin{split}
&||\nab^2\psi\beta||_{L^\infty_{u}L^1_{\ub}L^2(S)} \\
\leq& C||\nab^2\psi||_{L^\infty_{u}L^\infty_{\ub}L^2(S)}||\beta||_{L^\infty_{u}L^1_{\ub}L^\infty(S)}\\
\leq& C||\nab^2\psi||_{L^\infty_{u}L^\infty_{\ub}L^2(S)}(\sum_{i\leq 2}||\nab^i\beta||_{L^\infty_{u}L^1_{\ub}L^2(S)})\\
\leq& C||\nab^2\psi||_{L^\infty_{u}L^\infty_{\ub}L^2(S)}(\sum_{i\leq 2}||f(\ub)\nab^i\beta||_{L^\infty_{u}L^2_{\ub}L^\infty(S)})||f(\ub)^{-1}||_{L^\infty_{u}L^2_{\ub}L^\infty(S)}\\
\leq &C\ep\Delta_1\mathcal R.
\end{split}
\end{equation}
Combining \eqref{psi.1} and \eqref{psi.2}, we have
$$||\sum_{i_1+i_2+i_3\leq 3}\nab^{i_1}\psi^{i_2}\nabla^{i_3}\beta||_{L^\infty_{u}L^1_{\ub}L^2(S)}\leq C\ep(1+\Delta_1)^3\mathcal R.$$
We then estimate the {second term in \eqref{etaeqn}}. We separate the terms where more derivatives fall on $\psi_H$ and those where more derivatives fall on $\psi$:
\begin{equation}\label{etanonlinear}
\begin{split}
&||\sum_{i_1+i_2+i_3+i_4\leq 3}\nab^{i_1}\psi^{i_2}\nabla^{i_3}\psi\nabla^{i_4}\psi_{H}||_{L^\infty_{u}L^1_{\ub}L^2(S)} \\
\leq& C(\sum_{i_1\leq 1,{1\leq }i_2\leq 4}||\nab^{i_1}\psi||_{L^\infty_{u}L^\infty_{\ub}L^\infty(S)}^{i_2}) (\sum_{i_3\leq 3}||\nabla^{i_3}\psi_H||_{L^\infty_{u}L^1_{\ub}L^2(S)}) \\
&+C(1+||\psi||_{L^\infty_u L^\infty_{\ub}L^\infty(S)})(\sum_{i_1\leq 3}||\nabla^{i_1}\psi||_{L^\infty_{u}L^\infty_{\ub}L^2(S)})(\sum_{i_2\leq 1}||\nab^{i_2}\psi_H||_{L^\infty_{u}L^1_{\ub}L^\infty(S)})\\
\leq &C\Delta_1(1+\Delta_1)^3(\sum_{{i}\leq 3}||\nabla^{{i}}\psi_H||_{L^\infty_{u}L^1_{\ub}L^2(S)}+\sum_{{i}\leq 1}||\nab^{{i}}\psi_H||_{L^\infty_{u}L^1_{\ub}L^\infty(S)})\\
\leq &C\Delta_1(1+\Delta_1)^3||f(\ub)^{-1}||_{L^\infty_uL^2_{\ub}L^\infty(S)}\\
&\times(\sum_{{i}\leq 3}||f(\ub)\nabla^{{i}}\psi_H||_{L^\infty_{u}L^2_{\ub}L^2(S)}+\sum_{{i}\leq 1}||f(\ub)\nab^{{i}}\psi_H||_{L^\infty_{u}L^2_{\ub}L^\infty(S)})\\
\leq &C\Delta_1^2(1+\Delta_1)^3\epsilon.
\end{split}
\end{equation}
Hence, by Proposition \ref{transport}, we have
$$\sum_{i\leq 3}||\nab^i\eta||_{L^\infty_uL^\infty_{\ub}L^2(S)}\leq C(\mathcal O_{ini})+C\ep (\Delta_1^2(1+\Delta_1)^3+\mathcal R(1+\Delta_1)^3)\leq C(\mathcal O_{ini}),$$
after choosing $\epsilon$ to be sufficiently small.
Similarly, we consider the equation for $\nab_3\nab^i\etab$ to get
$$\sum_{i\leq 3}||\nab^i\etab||_{L^\infty_uL^\infty_{\ub}L^2(S)}\leq C(\mathcal O_{ini}).$$
\end{proof}

We now move to the terms that we denote by $\psi_{\Hb}$, i.e., $\trchb$, $\chibh$ and $\omegab$. All of them obey a $\nabla_4$ equation. Unlike the previous estimates for {$\psi$}, the initial data for the quantities $\psi_{\Hb}$ are not in $L^\infty_u$. We will therefore prove only a bound for $\psi_{\Hb}$ in the weighted norm $||f(u)\cdot||_{L^2_u L^\infty_{\ub}L^\infty(S)}$.

\begin{proposition}\label{psiHb}
Assume
$$\mathcal R <\infty,\quad\tilde{\mathcal O}_{4,2}<\infty.$$
Then there exists $\epsilon_0=\epsilon_0(\Delta_1,\mathcal R,\tilde{\mathcal O}_{4,2})$ such that whenever $\epsilon\leq \epsilon_0$, 
$$\sum_{i\leq 3}\mathcal O_{i,2}[\psi_{\Hb}]\leq C(\mathcal O_{ini}) .$$
In particular, as before, this estimate is independent of $\Delta_1$.
\end{proposition}
\begin{proof}
According to the definition of the $\mathcal O_{i,2}$ norm, we need to control the weighted $L^2_{u}L^\infty_{\ub}L^2(S)$ norm of $\psi_{\Hb}$. Using the null structure equations, for each $\psi_{\Hb}\in\{\trchb,\chibh,\omegab\}$, we have an equation of the type
$$\nabla_{4}\psi_{\Hb}=K+\nabla\etab+\psi\psi+\psi_H\psi_{\Hb}.$$
We also use the null structure equations commuted with angular derivatives:
\begin{equation*}
\begin{split}
\nabla_{4}\nabla^i\psi_{\Hb}
=&\sum_{i_1+i_2+i_3=i}\nabla^{i_1}\psi^{i_2}\nabla^{i_3}(K+\nabla\etab)+\sum_{i_1+i_2+i_3+i_4=i}\nabla^{i_1}\psi^{i_2}\nabla^{i_3}\psi\nabla^{i_4}\psi\\
&+\sum_{i_1+i_2+i_3+i_4=i}\nabla^{i_1}\psi^{i_2}\nabla^{i_3}\psi_H\nabla^{i_4}\psi_{\Hb}.
\end{split}
\end{equation*}
We estimate the curvature term using the curvature norm. Recall that the curvature norm for $K$ along the $H_u$ is weighted with $f(u)$. Using the Sobolev embedding theorem in Proposition \ref{Sobolev}, we have
\begin{equation}\label{psihb1}
\begin{split}
&||\sum_{i_1+i_2+i_3\leq 3}\nab^{i_1}\psi^{i_2}\nabla^{i_3}K||_{L^1_{\ub}L^2(S)} \\
\leq &C f(u)^{-1}(\sum_{i_1\leq 1,i_2\leq 3}||\nab^{i_1}\psi||^{i_2}_{L^2_{\ub}L^\infty(S)})(\sum_{i_3\leq 3}||f(u)\nab^{i_3} K||_{L^2_{\ub} L^2(S)})\\
&+C f(u)^{-1}||\nab^2\psi||_{L^2_{\ub}L^4(S)}||f(u) K||_{L^2_{\ub} L^4(S)}\\
\leq &C {\ep^{\f12}}f(u)^{-1}(\sum_{i_1\leq 3,i_2\leq 3}||\nab^{i_1}\psi||^{i_2}_{L^{{\infty}}_{\ub}L^2(S)})(\sum_{i_3\leq 3}||f(u)\nab^{i_3} K||_{L^2_{\ub} L^2(S)})\\
\leq &C\ep^{\frac 12} f(u)^{-1}(1+\Delta_1)^3\mathcal R.
\end{split}
\end{equation}
The term linear in $\nab^{4}\eta$ can be estimated analogously but using the $\tilde{\mathcal O}_{4,2}$ norms instead of the $\mathcal R$ norms:
\begin{equation}\label{psihb2}
\begin{split}
&||\sum_{i_1+i_2+i_3\leq 3}\nab^{i_1}\psi^{i_2}\nabla^{i_3+1}\eta||_{L^1_{\ub}L^2(S)} \\
\leq &C\ep (\sum_{i_1\leq 3}\sum_{i_2\leq 4}||\nab^{i_1}\psi||^{i_2}_{L^{{\infty}}_{\ub}L^2(S)})+C\ep^{\frac 12}f(u)^{-1}||f(u)\nab^4\eta||_{L^2_{\ub}L^2(S)}\\
\leq &C\ep(1+\Delta_1)^4+C\ep^{\frac 12}f(u)^{-1}\tilde{\mathcal O}_{4,2}.
\end{split}
\end{equation}
We now move to control the terms that are nonlinear in the Ricci coefficients. First, we estimate the terms without $\psi_H$ or $\psi_{\Hb}$:
\begin{equation}\label{psihb3}
\begin{split}
&||\sum_{i_1+i_2+i_3+i_4\leq 3}\nab^{i_1}\psi^{i_2}\nabla^{i_3}\psi\nabla^{i_4}\psi||_{L^1_{\ub}L^2(S)} \\
\leq& C\ep(\sum_{i_1\leq 1,i_2\leq 4}||\nab^{i_1}\psi||_{L^\infty_{\ub} L^\infty(S)}^{i_2})(\sum_{i_3\leq 3}||\nabla^{i_3}\psi||_{ L^\infty_{\ub} L^2(S)})\\
\leq &C\ep(1+\Delta_1)^5.
\end{split}
\end{equation}
We then control the term with both $\psi_H$ and $\psi_{\Hb}$:
\begin{equation}\label{psihb4}
\begin{split}
&||\sum_{i_1+i_2+i_3+i_4\leq 3}\nabla^{i_1}\psi^{i_2}\nabla^{i_3}\psi_H\nabla^{i_4}\psi_{\Hb}||_{L^1_{\ub}L^2(S)} \\
\leq& C(\sum_{i_1\leq 1,i_2\leq 3}||\nab^{i_1}\psi||_{L^\infty_{\ub}L^\infty(S)}^{i_2})(\sum_{i_3\leq 1}||\nab^{i_3}\psi_{\Hb}||_{L^\infty_{\ub}L^\infty(S)}{)}(\sum_{i_4\leq 3}||\nabla^{i_4}\psi_{H}||_{L^1_{\ub}L^2(S)}) \\
&+C(\sum_{i_1\leq 1,i_2\leq 3}||\nab^{i_1}\psi||_{L^\infty_{\ub}L^\infty(S)}^{i_2})(\sum_{i_3\leq 3}||\nab^{i_3}\psi_{\Hb}||_{L^\infty_{\ub}L^2(S)}{)}(\sum_{i_4\leq 1}||\nabla^{i_4}\psi_{H}||_{L^1_{\ub}L^\infty(S)}) \\
&+C||\nab^2\psi||_{L^\infty_{\ub}L^2(S)}||\psi_{\Hb}||_{L^\infty_{\ub}L^\infty(S)}||\psi_{H}||_{L^1_{\ub}L^\infty(S)} \\
\leq& C(\sum_{i_1\leq 1,i_2\leq 3}||\nab^{i_1}\psi||_{L^\infty_{\ub}L^\infty(S)}^{i_2})(\sum_{i_3\leq 1}||\nab^{i_3}\psi_{\Hb}||_{L^\infty_{\ub}L^\infty(S)}{)}(\sum_{i_4\leq 3}||f(\ub)\nabla^{i_4}\psi_{H}||_{L^2_{\ub}L^2(S)})||f(\ub)^{-1}||_{L^2_{\ub}} \\
&+C(\sum_{i_1\leq 1,i_2\leq 3}||\nab^{i_1}\psi||_{L^\infty_{\ub}L^\infty(S)}^{i_2})(\sum_{i_3\leq 3}||\nab^{i_3}\psi_{\Hb}||_{L^\infty_{\ub}L^2(S)}{)}(\sum_{i_4\leq 1}||f(\ub)\nabla^{i_4}\psi_{H}||_{L^2_{\ub}L^\infty(S)})||f(\ub)^{-1}||_{L^2_{\ub}} \\
&+C||\nab^2\psi||_{L^\infty_{\ub}L^2(S)}||\psi_{\Hb}||_{L^\infty_{\ub}L^\infty(S)}||f(\ub)\psi_{H}||_{L^2_{\ub}L^\infty(S)}||f(\ub)^{-1}||_{L^2_{\ub}} \\
\leq &C\ep(1+\Delta_1)^3(\sum_{i_1\leq 3}||f(\ub)\nab^{i_1}\psi_H||_{L^2_{\ub}L^2(S)})(\sum_{i_2\leq 3}||\nab^{i_2}\psi_{\Hb}||_{L^\infty_{\ub}L^2(S)})\\
\leq &C\epsilon\Delta_1(1+\Delta_1)^3(\sum_{i\leq 3}||\nab^i\psi_{\Hb}||_{L^\infty_{\ub}L^2(S)}).
\end{split}
\end{equation}
Therefore, by {the bounds} \eqref{psihb1}, \eqref{psihb2}, \eqref{psihb3} and \eqref{psihb4}, we have that for every fixed $u$,
\begin{equation*}
\begin{split}
\sum_{i\leq 3}||\nab^i\psi_{\Hb}||_{L^\infty_{\ub}L^{{2}}(S)}
\leq &{C(\sum_{i\leq 3}||\nab^i\psi_{\Hb}||_{L^2(S_{u,0})})}+C\ep^{\frac 12}f(u)^{-1}(\mathcal R+\tilde{\mathcal O}_{4,2})+C\ep(1+\Delta_1)^5\\
&+C\epsilon\Delta_1(1+\Delta_1)^3(\sum_{i\leq 3}||\nab^i\psi_{\Hb}||_{L^\infty_{\ub}L^2(S)}).
\end{split}
\end{equation*}
We now multiply this inequality by $f(u)$ and take the $L^2$ norm in $u$ to get
\begin{equation*}
\begin{split}
&\sum_{i\leq 3}||f(u)\nabla^i\psi_{\Hb}||_{L^2_{u}L^\infty_{\ub}L^2(S)} \\
\leq &C(\mathcal O_{ini})+C\ep^{\frac 12}||f(u)(f(u)^{-1})||_{L^2_u}(\mathcal R+\tilde{\mathcal O}_{4,2})+C\ep(1+\Delta_1)^5\\
&+C\epsilon\Delta_1(1+\Delta_1)^3(\sum_{i\leq 3}||f(u)\nab^i\psi_{\Hb}||_{L^2_uL^\infty_{\ub}L^2(S)})\\
\leq &C(\mathcal O_{ini}),
\end{split}
\end{equation*}
for $\epsilon$ sufficiently small. 
\end{proof}

Using instead the equation for $\nab_3\psi_H$, we obtain the following estimates in a completely analogous manner:
\begin{proposition}\label{psiH}
Assume
$$\mathcal R <\infty,\quad\tilde{\mathcal O}_{4,2}<\infty.$$
Then there exists $\epsilon_0=\epsilon_0(\Delta_1, \mathcal R,\tilde{\mathcal O}_{4,2})$ such that whenever $\epsilon\leq \epsilon_0$, 
$$\sum_{i\leq 3}\mathcal O_{i,2}[\psi_{H}]\leq C(\mathcal O_{ini}) .$$
In particular, this estimate is independent of $\Delta_1$.
\end{proposition}

By the Sobolev embedding theorems given by Proposition {\ref{Sobolev}}, we have thus closed our bootstrap assumption (\ref{BA1}) after choosing $\Delta_1$ to be sufficiently large depending on the initial data norm $\mathcal O_{ini}$. We have therefore proved the desired estimates for the Ricci coefficients and their first {three} angular {covariant} derivatives. We summarize this in the following proposition.
\begin{proposition}\label{Ricci}
Assume
$$\mathcal R <\infty,\quad\tilde{\mathcal O}_{4,2}<\infty.$$
Then there exists $\epsilon_0=\epsilon_0(\mathcal O_{ini},\mathcal R,\tilde{\mathcal O}_{4,2})$ such that whenever $\epsilon\leq \epsilon_0$, 
$$\sum_{i\leq 3}\mathcal O_{i,2}[\psi,\psi_{\Hb},\psi_H]\leq C(\mathcal O_{ini}).$$
\end{proposition}

\section{Elliptic Estimates for Fourth Derivatives of the Ricci Coefficients}\label{secRicci32}

We now estimate the 4th derivative of the Ricci coefficients. 
We introduce the following bootstrap assumption:
\begin{equation}\tag{A2}\label{BA2}
\tilde{\mathcal O}_{4,2}\leq \Delta_2,
\end{equation}
where $\Delta_2$ is a constant to be chosen later.

The estimates for the 4th derivative of the Ricci coefficients cannot be achieved only by the transport equations since there would be a loss in derivatives. We can however use the transport equation - Hodge system type estimates as in \cite{KN, Chr, KlRo}. We will first derive estimates for some chosen combination of $\nab^4(\psi,\psi_H,\psi_{\Hb})+{\nab^3}(\beta,K,\sigmac,\betab)$ by using transport equations. We will then show that the estimates for all the 4th derivatives of the Ricci coefficients can be proved via elliptic estimates.

In order to apply the elliptic estimates in Section~\ref{elliptic}, we need to first control the Gauss curvature and its first and second derivatives in $L^2(S)$.

\begin{proposition}\label{Kest}
Assume
$$\mathcal R<\infty.$$
Then there exists $\epsilon_0=\epsilon_0(\Delta_2,\mathcal R)$ such that whenever $\epsilon\leq \epsilon_0$,
$$\sum_{i\leq 2}||\nab^i K||_{L^\infty_uL^\infty_{\ub}L^2(S)}\leq C(\mathcal O_{ini}{)}.$$
\end{proposition}

\begin{proof}
$K$ obeys the following Bianchi equation:
$$\nab_4 K=\nab\beta+\psi_{H}K+\sum_{i_1+i_2+i_3\leq 1}\psi^{i_1}\nab^{i_2}\psi\nab^{i_3}\psi_H.$$
Commuting with angular derivatives, we have, for $i\leq 2$,
\begin{equation*}
\begin{split}
&\nab_4\nab^i K\\
=&\sum_{i_1+i_2+i_3\leq 2}\nab^{i_1}\psi^{i_2}\nab^{i_3+1}\beta+\sum_{i_1+i_2+i_3+i_4\leq 2}\nab^{i_1}\psi^{i_2}\nab^{i_3}\psi_H\nab^{i_4}K+\sum_{i_1+i_2+i_3+i_4\leq 3}\nab^{i_1}\psi^{i_2}\nab^{i_3}\psi\nab^{i_4}\psi_H.
\end{split}
\end{equation*}
By Proposition \ref{transport}, in order to control $\nab^iK$ in $L^\infty_uL^\infty_{\ub}L^2(S)$, we need to bound the right hand side in $L^\infty_uL^1_{\ub}L^2(S)$. We first control the term containing $\beta$:
\begin{equation*}
\begin{split}
&||\sum_{i_1+i_2+i_3\leq 2}\nab^{i_1}\psi^{i_2}\nab^{i_3+1}\beta||_{L^\infty_uL^1_{\ub}L^2(S)}\\
\leq &C(\sum_{i_1\leq 1,i_2\leq 2}||\nab^{i_1}\psi||_{L^\infty_uL^\infty_{\ub}L^\infty(S)}^{i_2})||f(\ub)^{-1}||_{L^2_{\ub}}(\sum_{i_3\leq 2}||f(\ub)\nab^{i_2+1}\beta||_{L^\infty_u L^2_{\ub} L^2(S)})\\
\leq &C(\mathcal O_{ini})\epsilon\mathcal R,
\end{split}
\end{equation*}
where we have used the estimates for $\psi$ given by Proposition \ref{Ricci}.
The term containing $K$ can be controlled by
\begin{equation*}
\begin{split}
&||\sum_{i_1+i_2+i_3+i_4\leq 2}\nab^{i_1}\psi^{i_2}\nab^{i_3}\psi_H\nab^{i_4}K||_{L^\infty_uL^1_{\ub}L^2(S)}\\
\leq &C(\sum_{i_1\leq 1,i_2\leq 2}||\nab^{i_1}\psi||^{i_2}_{L^\infty_uL^\infty_{\ub}L^\infty(S)})\int_0^{\ub}(\sum_{i_3+i_4\leq 2}||\nab^{i_3}\psi_H\nab^{i_4}K||_{L^\infty_u L^2(S_{u,\ub'})})d\ub'\\
\leq &C(\mathcal O_{ini})\int_0^{\ub}(\sum_{i_1+i_2\leq 2}||\nab^{i_1}\psi_H\nab^{i_2}K||_{L^\infty_u L^2(S_{u,\ub'})})d\ub'\\
\leq &C(\mathcal O_{ini})\int_0^{\ub}(\sum_{i_1\leq 2}||\nab^{i_1}\psi_H||_{L^\infty_u L^2(S_{u,\ub'})})(\sum_{i_2\leq 2}||\nab^{i_2}K||_{L^\infty_u L^2(S_{u,\ub'})})d\ub'.
\end{split}
\end{equation*}
The remaining term has been bounded in the previous section. By \eqref{etanonlinear} and Proposition \ref{Ricci},
\begin{equation*}
||\sum_{i_1+i_2+i_3\leq 2}\psi^{i_1}\nab^{i_2}\psi\nab^{i_3}\psi_H||_{L^\infty_uL^1_{\ub}L^2(S)}\leq C(\mathcal O_{ini})\epsilon.
\end{equation*}
Therefore, by Proposition \ref{transport},
\begin{equation*}
\begin{split}
&\sum_{i\leq 2}||\nab^i K||_{L^\infty_u L^2(S_{u,\ub})}\\
\leq &C(\mathcal O_{ini})(1+\ep+\ep\mathcal R+\int_0^{\ub}(\sum_{i_1\leq 2}||\nab^{i_1}\psi_H||_{L^\infty_u L^2(S_{u,\ub'})})(\sum_{i_2\leq 2}||\nab^{i_2}K||_{L^\infty_u L^2(S_{u,\ub'})})d\ub').
\end{split}
\end{equation*}
Gronwall's inequality implies
\begin{equation*}
\begin{split}
\sum_{i\leq 2}||\nab^i K||_{L^\infty_u L^2(S_{u,\ub})}
\leq C(\mathcal O_{ini})\exp(\sum_{i\leq 2}||\nab^i\psi_H||_{L^1_{\ub} L^\infty_u L^2(S)})
\leq C(\mathcal O_{ini})
\end{split}
\end{equation*}
since by Proposition \ref{Ricci}, $\displaystyle\sum_{i\leq 1}||\nab^i\psi_H||_{L^1_{\ub} L^\infty_u L^2(S)}\leq C(\mathcal O_{ini})$ for $\ep$ sufficiently small.
\end{proof}

It is easy to see that since $\sigmac$ satisfies a similar schematic Bianchi equation as $K$, we also have the following estimates for $\sigmac$ and its derivative:

\begin{proposition}\label{sigmacest}
Assume 
$$\mathcal R<\infty.$$
Then there exists $\epsilon_0=\epsilon_0(\Delta_2,\mathcal R)$ such that whenever $\epsilon\leq \epsilon_0$,
$$\sum_{i\leq 2}||\nab^i\sigmac||_{L^\infty_uL^\infty_{\ub}L^2(S)}\leq C(\mathcal O_{ini}).$$
\end{proposition}
Using Proposition \ref{Kest}, we now control the fourth derivatives of the Ricci coefficients. We first bound $\nab^4\trch$ using the transport equation.
\begin{proposition}\label{trch3}
There exists $\epsilon_0=\epsilon_0(\mathcal O_{ini},\Delta_2)$ such that whenever $\epsilon\leq \epsilon_0$,
$$||f(\ub)\nab^4\trch||_{L^2_{\ub}L^\infty_uL^2(S)}\leq C(\mathcal O_{ini}).$$
\end{proposition}
\begin{proof}
Consider the following equation:
$$\nabla_4 \trch=\psi_H\psi_H,$$
After commuting with angular derivatives, we have
$$\nabla_{4}\nabla^{4}\trch=\sum_{i_1+i_2+i_3+i_4=4}\nabla^{i_1}\psi^{i_2}\nabla^{i_3}\psi_H\nabla^{i_4}\psi_H.$$
By Proposition \ref{transport}, in order to control $\nab^4\trch$ in $ L^2(S_{u,\ub})$ , we need to bound the right hand side in $L^1_{\ub} L^2(S)$. Using the fact that $f(\ub)$ is decreasing, this can be achieved using Sobolev embedding (Proposition {\ref{Sobolev}}) by
\begin{equation*}
\begin{split}
&{\sum_{i_1+i_2+i_3+i_4\leq 4}\int_0^{\ub}} ||\nabla^{i_1}\psi^{i_2}\nabla^{i_3}\psi_H\nabla^{i_4}\psi_H||_{L^2(S_{u,\ub'})} {d\ub'}\\
\leq& Cf(\ub)^{-2}(\sum_{i_1\leq 3}\sum_{i_2\leq 4}||\nab^{i_1}\psi||^{i_2}_{L^\infty_{\ub}L^2(S)})(\sum_{i_3\leq 2}||f(\ub)\nabla^{i_3}\psi_H||_{L^2_{\ub}L^2(S)})(\sum_{i_4\leq 4}||f(\ub)\nabla^{i_4}\psi_H||_{L^2_{\ub}L^2(S)})\\
\leq &Cf(\ub)^{-2}\Delta_2.
\end{split}
\end{equation*}
By Proposition \ref{transport}, we have
\begin{equation}\label{trch0}
||\nabla^4\trch||_{L^2(S_{u,\ub})}\leq C(\mathcal O_{ini})+C(\mathcal O_{ini}) f(\ub)^{-2}\Delta_2.
\end{equation}
Multiplying \eqref{trch0} by $f(\ub)$ and taking first the $L^\infty$ norm in $u$ and then the $L^2$ norm in $\ub$, we have
\begin{equation*}
\begin{split}
||f(\ub)\nabla^{4}\trch||_{L^2_{\ub}L^\infty_u L^2(S)}\leq &C(\mathcal O_{ini})+C(\mathcal O_{ini})||f(\ub)^{-1}||_{L^2_{\ub}}\Delta_2\leq C(\mathcal O_{ini})+C\ep\Delta_2,
\end{split}
\end{equation*}
where we have used
$$||f(\ub)^{-1}||_{L^2_{\ub}}\leq C\ep.$$
Thus, the conclusion follows by choosing $\ep$ to be sufficiently small depending on $\Delta_2$.
\end{proof}

Once we have the estimates for $\nab^4\trch$, we can control $\nab^4\chih$ using elliptic estimates:
\begin{proposition}\label{chih3}
Assume 
$$\mathcal R<\infty.$$
Then there exists $\epsilon_0=\epsilon_0(\mathcal O_{ini},\Delta_2,\mathcal R)$ such that whenever $\epsilon\leq \epsilon_0$,
$$||f(\ub)\nab^4\chih||_{L^\infty_uL^2_{\ub}L^2(S)}\leq C(\mathcal O_{ini})+C\mathcal R.$$
\end{proposition}
\begin{proof}
We now use the Codazzi equation
$$\div\chih=\frac 12\nabla\trch-\beta+\psi\psi_H$$
and apply elliptic estimates from Proposition \ref{elliptictraceless} to get
\begin{equation}\label{chihelliptic}
\begin{split}
&||\nabla^4\chih||_{L^2(S)}\\
\leq &C(\sum_{i\leq 4}||\nabla^i\trch||_{L^2(S)}+\sum_{i\leq 3}||\nabla^i\beta||_{L^2(S)}+\sum_{i_1+i_2\leq 3}||\nabla^{i_1}\psi\nabla^{i_2}\psi_H||_{L^2(S)} +\sum_{i\leq 3}||\nab^i\chih||_{L^2(S)}).
\end{split}
\end{equation}
Notice that we can apply elliptic estimates using Proposition \ref{elliptictraceless} since we have {the} estimates for the Gauss curvature from Proposition \ref{Kest}. Multiply (\ref{chihelliptic}) by $f(\ub)$ and take {the} $L^\infty_uL^2_{\ub}$ norm to get
\begin{equation*}
\begin{split}
||f(\ub)\nabla^4\chih||_{L^\infty_uL^2_{\ub}L^2(S)}
\leq &C(\sum_{i\leq 4}||f(\ub)\nabla^i\trch||_{L^\infty_uL^2_{\ub}L^2({S})}+\sum_{i\leq 3}||f(\ub)\nabla^i\beta||_{L^\infty_uL^2_{\ub}L^2(S)}\\
&\quad+\sum_{i_1+i_2\leq 3}||f(\ub)\nabla^{i_1}\psi\nabla^{i_2}\psi_H||_{L^\infty_uL^2_{\ub}L^2(S)} +\sum_{i\leq 3}||f(\ub)\nab^i\chih||_{L^\infty_uL^2_{\ub}L^2(S)})\\
\leq &C(\mathcal O_{ini})+C\mathcal R+C\sum_{i_1+i_2\leq 3}||f(\ub)\nabla^{i_1}\psi\nabla^{i_2}\psi_H||_{L^\infty_uL^2_{\ub}L^2(S)}.
\end{split}
\end{equation*}
By Proposition \ref{Ricci} and Sobolev embedding theorem in Proposition {\ref{Sobolev}}, we have
\begin{equation*}
\begin{split}
&\sum_{i_1+i_2\leq 3}||f(\ub)\nabla^{i_1}\psi\nabla^{i_2}\psi_H||_{L^\infty_uL^2_{\ub}L^2(S)}\\
\leq &C(\sum_{i_1\leq 3}||\nab^{i_1}\psi||_{L^\infty_uL^\infty_{\ub}L^2(S)})(\sum_{i_2\leq 3}||f(\ub)\nab^{i_2}\psi_H||_{L^\infty_uL^2_{\ub}L^2(S)})\leq C(\mathcal O_{ini}).
\end{split}
\end{equation*}
Therefore,
$$||f(\ub)\nabla^4\chih||_{L^\infty_uL^2_{\ub}L^2(S)}\leq C(\mathcal O_{ini})+C\mathcal R.$$
\end{proof}

The $\tilde{\mathcal O}_{4,2}$ estimates for $\nab^4\trchb$ and $\nab^4\chibh$ follow identically as that for $\nab^4\trch$ and $\nab^4\chih$:

\begin{proposition}\label{chib3}
Assume
$$\mathcal R<\infty.$$
Then there exists $\epsilon_0=\epsilon_0(\mathcal O_{ini},\Delta_2,\mathcal R)$ such that whenever $\epsilon\leq \epsilon_0$,
$$||f(u)\nab^4\trchb||_{L^\infty_{\ub}L^2_{u}L^2(S)}\leq C(\mathcal O_{ini}),$$
and
$$||f(u)\nab^4\chibh||_{L^\infty_{\ub}L^2_{u}L^2(S)}\leq C(\mathcal O_{ini})+C\mathcal R.$$
\end{proposition}

We then prove estimates for $\nab^4\eta$. To do so, we first prove estimates for third derivatives of $\mu=-\div\eta+K$ and recover the control for $\nab^4\eta$ via elliptic estimates. 
\begin{proposition}\label{mu}
Assume
$$\mathcal R<\infty.$$
Then there exists $\epsilon_0=\epsilon_0(\mathcal O_{ini},\Delta_2,\mathcal R)$ such that whenever $\epsilon\leq \epsilon_0$,
$$||f(u)\nab^4\eta||_{L^\infty_uL^2_{\ub}L^2(S)}\leq C(\mathcal O_{ini})(\epsilon^{\frac 12}+\mathcal R),$$
and
$$||f(\ub)\nab^4\eta||_{L^\infty_{\ub}L^2_{u}L^2(S)}\leq C(\mathcal O_{ini})(\epsilon^{\frac 12}+\mathcal R).$$
\end{proposition}
\begin{proof}
Recall that
$$\mu=-\div\eta+K.$$
Then $\mu$ satisfies the following equation\footnote{It is important to note that the potentially harmful term $\psi_{\Hb}\psi_H\psi_H$ is absent in this equation. This required structure is the reason that we perform this renormalization instead of using $\mu=-\div\eta-\rho+\frac 12\chih\cdot\chibh$ as in \cite{LR, LR2}.}:
$$\nabla_4 \mu=\psi_H(K,\sigmac)+\sum_{i_1+i_2+i_3=1}\psi^{i_1}\nab^{i_2}\psi\nabla^{i_3}\psi_H.$$

After commuting with angular derivatives, we get
\begin{equation*}
\begin{split}
\nabla_{4}\nabla^3\mu=&\sum_{i_1+i_2+i_3+i_4=3}\nabla^{i_1}\psi^{i_2}\nabla^{i_3}\psi_H\nabla^{i_4}(K,\sigmac)+\sum_{i_1+i_2+i_3+i_4=4}\nabla^{i_1}\psi^{i_2}\nabla^{i_3}\psi\nabla^{i_4}\psi_H.
\end{split}
\end{equation*}
We now control each of the terms on the right hand side in $L^1_{\ub}L^2(S)$. The first term, which contains curvature components, can be estimated by
\begin{equation*}
\begin{split}
&{\sum_{i_1+i_2+i_3+i_4=3}\int_0^{\ub}}||\nabla^{i_1}\psi^{i_2}\nabla^{i_3}\psi_H\nabla^{i_4}(K,\sigmac)||_{L^2(S_{u,\ub'})}{d\ub'}\\
\leq &Cf(u)^{-1}f(\ub)^{-1}(\sum_{i_1\leq 3}\sum_{i_2\leq 3}||\nab^{i_1}\psi||^{i_2}_{L^\infty_{\ub}L^2(S)})(\sum_{i_3\leq 3}||f(\ub)\nab^{i_3}\psi_H||_{L^2_{\ub}L^2(S)})\\
&\quad\quad\times(\sum_{i_4\leq 3}||f(u)\nab^{i_4}(K,\sigmac)||_{L^2_{\ub}L^2(S)})\\
\leq &C(\mathcal O_{ini})f(u)^{-1}f(\ub)^{-1}\mathcal R,
\end{split}
\end{equation*}
using the bounds obtained in Proposition \ref{Ricci}.
The second term can be controlled using Sobolev embedding in Proposition {\ref{Sobolev}} by
\begin{equation*}
\begin{split}
&{\sum_{i_1+i_2+i_3+i_4=4}\int_0^{\ub}}||\nabla^{i_1}\psi^{i_2}\nabla^{i_3}\psi\nabla^{i_4}\psi_H||_{L^2(S_{u,\ub'})}{d\ub'}\\
\leq &Cf(u)^{-1}f(\ub)^{-1}(\sum_{i_1\leq 4}\sum_{i_2\leq 5}||\nab^{i_1}\psi||^{i_2}_{L^\infty_{\ub}L^2(S)})(\sum_{i_3\leq 4}||f(\ub)\nab^{i_3}\psi_H||_{L^2_{\ub}L^2(S)})(\sum_{i_4\leq 4}||f(u)\nab^{i_4}\psi||_{L^2_{\ub}L^2(S)})\\
\leq &C(\mathcal O_{ini})f(u)^{-1}f(\ub)^{-1}(1+\Delta_2)^2
\end{split}
\end{equation*}
using the estimates in Proposition \ref{Ricci}.
Therefore, by Proposition \ref{transport}, we have
\begin{equation}\label{mumain}
\begin{split}
||\nab^3\mu||_{L^2(S_{u,\ub})}
\leq &C(\mathcal O_{ini})(1+ f(u)^{-1}f(\ub)^{-1}(\mathcal R+(1+\Delta_2)^2)).
\end{split}
\end{equation}
Recall that the $L^2_{\ub}$ norm of $f(\ub)^{-1}$ is bounded by $\ep$.
Thus, multiplying \eqref{mumain} by $f(u)$ and taking the $L^2$ norm in $\ub$, we get
\begin{equation*}
\begin{split}
||f(u)\nab^3\mu||_{L^2_{\ub}L^2(S)}\leq C(\mathcal O_{ini})(\ep^{\frac 12}+\epsilon(\mathcal R+(1+\Delta_2)^2))\leq C(\mathcal O_{ini})\ep^{\frac 12},
\end{split}
\end{equation*}
for $\ep$ sufficiently small. Similarly, multiplying \eqref{mumain} by $f(\ub)$ and taking the $L^2$ norm in $u$, we get
$$||f(\ub)\nab^3\mu||_{L^2_{u}L^2(S)}\leq C(\mathcal O_{ini})\ep^{\frac 12}.$$
We can obtain bounds for $\nab^4\eta$ from the control of $\nab^3\mu$ using elliptic estimates as follows.
By the div-curl systems
$$\div\eta=-\mu+K,\quad \curl\eta=\sigmac,$$
and the elliptic estimates given by Propositions \ref{ellipticthm} and \ref{Kest}, we have
\begin{equation*}
||\nabla^4\eta||_{L^2(S)}\leq C(\sum_{i\leq 3}||\nabla^i\mu||_{L^2(S)}
+\sum_{i\leq 3}||\nabla^i(K,\sigmac)||_{L^2(S)}+\sum_{i\leq 3}||\nab^i\eta||_{L^2(S)}).
\end{equation*}
Therefore,
\begin{equation*}
\begin{split}
&||f(u)\nabla^4\eta||_{L^2_{\ub}L^2(S)}\\
\leq &C(\sum_{i\leq 3}||f(u)\nab^i\mu||_{L^2_{\ub}L^2(S)}+\sum_{i\leq 3}||f(u)\nabla^i(K,\sigmac)||_{L^2_{\ub}L^2(S)}+\sum_{i\leq 3}||f(u)\nab^i\eta||_{L^2_{\ub}L^2(S)})\\
\leq &C(\mathcal O_{ini})(\epsilon^{\frac 12}+\mathcal R).
\end{split}
\end{equation*}
Similarly,
$$||f(\ub)\nabla^4\eta||_{L^2_{u}L^2(S)}\leq C(\mathcal O_{ini})(\epsilon^{\frac 12}+\mathcal R).$$
\end{proof}

A similar proof shows that the conclusion of Proposition \ref{mu} holds also for $\nab^3\etab$:
\begin{proposition}\label{mub}
Assume
$$\mathcal R<\infty.$$
Then there exists $\epsilon_0=\epsilon_0(\mathcal O_{ini},\Delta_2,\mathcal R)$ such that whenever $\epsilon\leq \epsilon_0$,
$$||f(u)\nab^4\etab||_{L^\infty_uL^2_{\ub}L^2(S)}\leq C(\mathcal O_{ini})(\epsilon^{\frac 12}+\mathcal R),$$
and
$$||f(\ub)\nab^4\etab||_{L^\infty_{\ub}L^2_{u}L^2(S)}\leq C(\mathcal O_{ini})(\epsilon^{\frac 12}+\mathcal R).$$
\end{proposition}

We now move to the estimates for $\nab^4\omegab$:

\begin{proposition}\label{omegab3}
Assume
$$\mathcal R<\infty.$$
Then there exists $\epsilon_0=\epsilon_0(\mathcal O_{ini},\Delta_2,\mathcal R)$ such that whenever $\epsilon\leq \epsilon_0$,
$$||f(u)\nab^4\omegab||_{L^\infty_{\ub}L^2_uL^2(S)}\leq C(\mathcal O_{ini})(1+\mathcal R).$$
\end{proposition}
\begin{proof}
Let $\omegab^{\dagger}$ be defined as the solution to
$$\nab_4\omegab^{\dagger}=\frac 12 \sigmac$$
with zero data {on $\Hb_0$} and 
$$\kappab:=-\nab\omegab+^*\nab\omegab^{\dagger}-\frac 12\betab.$$
By the definition of $\omegab^{\dagger}$, it is easy to see that using Proposition \ref{transport},
$$\sum_{i\leq 3}||\nab^i\omegab^{\dagger}||_{L^2_uL^\infty_{\ub}L^2(S)}\leq C\epsilon\mathcal R\leq C(\mathcal O_{ini}).$$
In other words, ${\nab^i}\omegab^{\dagger}$ satisfies {much better} estimates\footnote{{We recall that for $\psi_{\Hb}$ we only have the degenerate estimate
$$\sum_{i\leq 3}||f(u)\nab^i\psi_{\Hb}||_{L^2_uL^\infty_{\ub}L^2(S)}\leq C(\mathcal O_{ini}).$$}} {than} ${\nab^i}\psi_{\Hb}$ {for $i\leq 3$}. {With this in mind, i}n the proof of this proposition, we will also use $\psi_{\Hb}$ to denote $\omegab^{\dagger}$ (in addition to $\trchb$, $\chibh$ and $\omegab$).

{With this convention,} $\kappab$ then obeys the following {schematic} equation:
$$\nabla_4 \kappab=\psi(K,\sigmac)+\sum_{i_1+i_2+i_3=1}\psi^{i_1}\nab^{i_2}\psi\nab^{i_3}\psi+\sum_{i_1+i_2+i_3=1}\psi^{i_1}\nab^{i_2}\psi_H\nab^{i_3}\psi_{\Hb}.$$
After commuting with angular derivatives, we get
\begin{equation*}
\begin{split}
\nabla_{4}\nabla^3\kappab=&\sum_{i_1+i_2+i_3+i_4=3}\nabla^{i_1}\psi^{i_2}\nabla^{i_3}\psi\nabla^{i_4}(K,\sigmac)
+\sum_{i_1+i_2+i_3+i_4=4}\nabla^{i_1}\psi^{i_2}\nabla^{i_3}\psi\nabla^{i_4}\psi\\
&+\sum_{i_1+i_2+i_3+i_4=4}\nabla^{i_1}\psi^{i_2}\nabla^{i_3}\psi_H\nabla^{i_4}\psi_{\Hb}.\\
\end{split}
\end{equation*}
Therefore,
\begin{equation*}
\begin{split}
||\nab^3\kappab||_{L^2(S_{u,\ub})}
\leq &C||\nab^3\kappab||_{L^2(S_{u,0})}+C||\sum_{i_1+i_2+i_3+i_4=3}\nabla^{i_1}\psi^{i_2}\nabla^{i_3}\psi\nabla^{i_4}(K,\sigmac)||_{L^1_{\ub}L^2(S)}\\
&+C||\sum_{i_1+i_2+i_3+i_4=4}\nabla^{i_1}\psi^{i_2}\nabla^{i_3}\psi\nabla^{i_4}\psi||_{L^1_{\ub}L^2(S)}\\
&+C||\sum_{i_1+i_2+i_3+i_4=4}\nabla^{i_1}\psi^{i_2}\nabla^{i_3}\psi_H\nabla^{i_4}\psi_{\Hb}||_{L^1_{\ub}L^2(S)}.
\end{split}
\end{equation*}
Multiplying by $f(u)$ and taking the $L^2$ norm in $u$, we get
\begin{equation*}
\begin{split}
&||f(u)\nab^3\kappab||_{L^2_uL^2(S)}\\
\leq &C||f(u)\nab^3\kappab||_{L^2_uL^2(S_{u,0})}+C||f(u)\sum_{i_1+i_2+i_3+i_4=3}\nabla^{i_1}\psi^{i_2}\nabla^{i_3}\psi\nabla^{i_4}(K,\sigmac)||_{{L^2_u}L^1_{\ub}L^2(S)}\\
&+C||f(u)\sum_{i_1+i_2+i_3+i_4=4}\nabla^{i_1}\psi^{i_2}\nabla^{i_3}\psi\nabla^{i_4}\psi||_{L^2_uL^1_{\ub}L^2(S)}\\
&+C||f(u)\sum_{i_1+i_2+i_3+i_4=4}\nabla^{i_1}\psi^{i_2}\nabla^{i_3}\psi_H\nabla^{i_4}\psi_{\Hb}||_{L^2_uL^1_{\ub}L^2(S)}.
\end{split}
\end{equation*}
The first term is an initial data term and it is bounded by a constant depending only on $\mathcal O_{ini}$.
We estimate each of the nonlinear terms. The second term can be controlled by
\begin{equation*}
\begin{split}
&||f(u)\sum_{i_1+i_2+i_3+i_4=3}\nabla^{i_1}\psi^{i_2}\nabla^{i_3}\psi\nabla^{i_4}(K,\sigmac)||_{L^2_uL^1_{\ub}L^2(S)}\\
\leq &C\ep(\sum_{i_1\leq 3}\sum_{i_2\leq 4}||\nab^{i_1}\psi||^{i_2}_{L^\infty_{\ub}L^\infty_{\ub}L^2(S)})(\sum_{i_3\leq 3}||f(u)\nab^{i_3}(K,\sigmac)||_{L^\infty_uL^2_{\ub}L^2(S)})\\
\leq &C(\mathcal O_{ini})\ep\mathcal R.
\end{split}
\end{equation*}
The third term can be bounded by
\begin{equation*}
\begin{split}
&||f(u)\sum_{i_1+i_2+i_3+i_4=4}\nabla^{i_1}\psi^{i_2}\nabla^{i_3}\psi\nabla^{i_4}\psi||_{L^2_uL^1_{\ub}L^2(S)}\\
\leq &C\ep(\sum_{i_1\leq 3}\sum_{i_2\leq 4}||\nab^{i_1}\psi||^{i_2}_{L^\infty_{\ub}L^\infty_uL^2(S)})(\sum_{i_3\leq 3}||\nab^{i_3}\psi||_{L^\infty_{\ub}L^\infty_uL^2(S)})(\sum_{i_4\leq 4}||f(u)\nab^{i_4}\psi||_{L^\infty_{u}L^2_{\ub}L^2(S)})\\
\leq &C(\mathcal O_{ini})\ep\Delta_2.
\end{split}
\end{equation*}
The fourth term can be estimated by
\begin{equation*}
\begin{split}
&||f(u)\sum_{i_1+i_2+i_3+i_4=4}\nabla^{i_1}\psi^{i_2}\nabla^{i_3}\psi_H\nabla^{i_4}\psi_{\Hb}||_{L^2_uL^1_{\ub}L^2(S)}\\
\leq &C\ep^{\frac 12}(\sum_{i_1\leq 3}\sum_{i_2\leq 4}||\nab^{i_1}\psi||^{i_2}_{L^\infty_{\ub}L^\infty_uL^2(S)})(\sum_{i_3\leq 3}||\nab^{i_3}\psi_H||_{L^1_{\ub}L^\infty_uL^2(S)})(\sum_{i_4\leq 4}||f(u)\nab^{i_4}\psi_{\Hb}||_{L^\infty_{\ub}L^2_uL^2(S)})\\
&+C||f(\ub)^{-1}||_{L^2_{\ub}}||f(\ub)\nab^4\psi_H||_{L^\infty_uL^2_{\ub}L^2(S)}||f(u)\psi_{\Hb}||_{L^2_uL^\infty_{\ub} L^\infty(S)}\\
\leq &C(\mathcal O_{ini})\ep(1+\Delta_2).
\end{split}
\end{equation*}
Therefore,
\begin{equation}\label{kappab}
\begin{split}
||{f(u)}\nabla^3\kappab||_{L^2_uL^\infty_{\ub} L^2(S)}\leq &C(\mathcal O_{ini})(1+\ep(1+\Delta_2+\mathcal R))\leq C(\mathcal O_{ini}),
\end{split}
\end{equation}
after choosing $\ep$ to be sufficiently small.
Finally, we retrieve the estimates for $\nab^4\omb$ and $\nab^4\omb^{\dagger}$ from the bounds for $\nab^3\kappab$. To this end, consider the div-curl system
$$\div\nabla\omegab=-\div\kappab-\frac 12\div\betab,$$
$$\curl\nabla\omegab=0,$$
$$\div\nabla\omegab^{\dagger}=-\curl\kappab-\frac 12\curl\betab,$$
$$\curl\nabla\omegab^{\dagger}=0.$$
By elliptic estimates given by Propositions \ref{ellipticthm} and \ref{Kest}, we have
\begin{equation*}
\begin{split}
||\nabla^4(\omegab,\omegab^{\dagger})||_{L^2(S_{u,\ub})}
\leq &C(\sum_{i\leq 3}||\nab^i\kappab||_{L^2(S_{u,\ub})}+\sum_{i\leq 3}||\nab^i\betab||_{L^2(S_{u,\ub})}+\sum_{i\leq 3}||\nab^i(\omegab,\omegab^{\dagger})||_{L^2(S_{u,\ub})}).
\end{split}
\end{equation*}
Therefore, using Proposition \ref{psi}, \eqref{kappab} and the curvature norm,
$$||{f(u)}\nabla^4(\omegab,\omegab^{\dagger})||_{L^\infty_{\ub}L^2_uL^2(S)}\leq C(\mathcal O_{ini})(1+\mathcal R).$$
\end{proof}

By switching $\omegab$ and $\omega$ as well as $e_3$ and $e_4$, we also have the following estimates for $\nab^4\omega$:
\begin{proposition}\label{omega3}
Assume
$$\mathcal R<\infty.$$
Then there exists $\epsilon_0=\epsilon_0(\mathcal O_{ini},\Delta_2,\mathcal R)$ such that whenever $\epsilon\leq \epsilon_0$,
$$||f(\ub)\nab^4\omega||_{L^\infty_{u}L^2_{\ub}L^2(S)}\leq C(\mathcal O_{ini})(1+\mathcal R).$$
\end{proposition}

We have thus controlled the fourth angular derivatives of all Ricci coefficients and have closed the bootstrap assumption (\ref{BA2}) after choosing $\Delta_2$ to be sufficiently large depending on $\mathcal O_{ini}$ and $\mathcal R$. We summarize {this} in the following proposition:
\begin{proposition}\label{Ricci32}
Assume
$$\mathcal R<\infty.$$
There exists $\epsilon_0=\epsilon_0(\mathcal O_{ini},\mathcal R)$ such that whenever $\epsilon\leq \epsilon_0$,
$$\tilde{\mathcal O}_{4,2}\leq C(\mathcal O_{ini})(1+\mathcal R).$$
\end{proposition}

\section{Estimates for Curvature}\label{seccurv}

In this section, we derive and prove the energy estimates. To this end, we introduce the following bootstrap assumptions:
\begin{equation}\tag{A3}\label{BA3}
\mathcal R\leq \Delta_3,
\end{equation}
where $\Delta_3$ is a constant to be chosen later.

In order to derive the energy estimates, we need the following integration by parts formula, which can be proved by direct computations:
\begin{proposition}\label{intbyparts34}
Let $D_{u,\ub}$ be defined as the spacetime region whose coordinates $(u',\ub')$ satisfy $0\leq u'\leq u$ and $0\leq \ub'\leq \ub$. Suppose $\phi_1$ and $\phi_2$ are tensors of rank $r$, then
$$\int_{D_{u,\ub}} \phi_1 \nabla_4\phi_2+\int_{D_{u,\ub}}\phi_2\nabla_4\phi_1= \int_{\Hb_{\ub}(0,u)} \phi_1\phi_2-\int_{\Hb_0(0,u)} \phi_1\phi_2+\int_{D_{u,\ub}}(2\omega-\trch)\phi_1\phi_2,$$
$$\int_{D_{u,\ub}} \phi_1 \nabla_3\phi_2+\int_{D_{u,\ub}}\phi_2\nabla_3\phi_1= \int_{H_{u}(0,\ub)} \phi_1\phi_2-\int_{H_0(0,\ub)} \phi_1\phi_2+\int_{D_{u,\ub}}(2\omegab-\trchb)\phi_1\phi_2.$$
\end{proposition}
\begin{proposition}\label{intbypartssph}
Suppose we have a tensor $^{(1)}\phi$ of rank $r$ and a tensor $^{(2)}\phi$ of rank $r-1$. Then
$$\int_{D_{u,\ub}}{ }^{(1)}\phi^{A_1A_2...A_r}\nabla_{A_r}{ }^{(2)}\phi_{A_1...A_{r-1}}+\int_{D_{u,\ub}}\nabla^{A_r}{ }^{(1)}\phi_{A_1A_2...A_r}{ }^{(2)}\phi^{A_1...A_{r-1}}= -\int_{D_{u,\ub}}(\eta+\etab){ }^{(1)}\phi{ }^{(2)}\phi.$$
\end{proposition}
With these we are now ready to derive energy estimates for $\nab^i(K,\sigmac)$ in $L^2(H_u)$ and for $\nab^i\betab$ in $L^2(\Hb_{\ub})$. The most important observation is that the two uncontrollable terms have favorable signs. {This in turn is due to the choice of $f(u)$ which is decreasing towards the future.}
\begin{proposition}\label{ee1}
The following $L^2$ estimates for the curvature hold:
\begin{equation*}
\begin{split}
&\sum_{i\leq 3}(||f(u)\nab^i(K,\sigmac)||^2_{L^\infty_u L^2_{\ub}L^2(S)}+||f(u)\nab^i\betab||^2_{L^\infty_{\ub} L^2_{u}L^2(S)})  \\
\leq &\sum_{i\leq 3}(||f(u)\nab^i(K,\sigmac)||^2_{L^2_{\ub}L^2(S_{0,\ub})}+||f(u)\nab^i\betab||^2_{L^2_{u}L^2(S_{u,0})}) \\
&+||f(u)^2(\sum_{i\leq 3}\nabla^i(K,\sigmac))(\sum_{i_1+i_2+i_3+i_4\leq 4}\nabla^{i_1}\psi^{i_2}\nabla^{i_3}\psi\nabla^{i_4}\psi_{\Hb})||_{L^1_uL^1_{\ub}L^1(S)}\\
&+||f(u)^2(\sum_{i\leq 3}\nabla^i(K,\sigmac))(\sum_{i_1+i_2+i_3+i_4\leq 3}\nab^{i_1}\psi^{i_2}\nabla^{i_3}\psi_{\Hb}\nabla^{i_4}(K,\sigmac))||_{L^1_uL^1_{\ub}L^1(S)}\\
&+||f(u)^2(\sum_{i\leq 3}\nabla^i\betab)(\sum_{i_1+i_2+i_3+i_4\leq 3}\nab^{i_1}\psi^{i_2}\nabla^{i_3}\psi\nabla^{i_4}(K,\sigmac))||_{L^1_uL^1_{\ub}L^1(S)}\\
&+||f(u)^2(\sum_{i\leq 3}\nabla^i\betab)(\sum_{i_1+i_2+i_3+i_4\leq 2}\nab^{i_1}\psi^{i_2}\nabla^{i_3}K\nabla^{i_4}(K,\sigmac))||_{L^1_uL^1_{\ub}L^1(S)}\\
&+||f(u)^2(\sum_{i\leq 3}\nabla^i\betab)(\sum_{i_1+i_2+i_3+i_4\leq 4}\nabla^{i_1}\psi^{i_2}\nabla^{i_3}\psi_{\Hb}\nabla^{i_4}\psi_H)||_{L^1_uL^1_{\ub}L^1(S)}{.}
\end{split}
\end{equation*}
\end{proposition}
\begin{proof}
Consider the following schematic Bianchi equations:
\begin{equation*}
\begin{split}
\nab_3\sigmac+\div ^*\betab=&\psi_{\Hb}\sigmac+\sum_{i_1+i_2+i_3= 1}\psi^{i_1}\nab^{i_2}\psi\nab^{i_3}\psi_{\Hb},\\
\nab_3 K-\div\betab=&\psi_{\Hb}K+\sum_{i_1+i_2+i_3=1}\psi^{i_1}\nab^{i_2}\psi\nab^{i_3}\psi_{\Hb},\\
\nab_4\betab-\nabla K -^*\nabla\sigmac=& \psi(K,\sigmac)+\sum_{i_1+i_2+i_3=1}\psi^{i_1}\nab^{i_2}\psi_H\nab^{i_3}\psi_{\Hb}.
\end{split}
\end{equation*}
Commute the first equation with $i$ angular derivatives for $i\leq 3$, we get the equation for $\nab_3\nab^i\sigmac$,
\begin{equation}\label{ee1.1}
\begin{split}
&\nab_3\nab^i\sigmac+\div ^*\nab^i\betab\\
=&\sum_{i_1+i_2+i_3+i_4=i}\nab^{i_1}\psi^{i_2}\nab^{i_3}\psi_{\Hb}\nab^{i_4}{(K,\sigmac)}+\sum_{i_1+i_2+i_3+i_4=i+1}\nab^{i_1}\psi^{i_2}\nab^{i_3}\psi\nab^{i_4}\psi_{\Hb}.
\end{split}
\end{equation}
Notice that in the above equation, there are terms arising from the commutator $[\nab^i,\div]\betab$, which can be expressed in terms of the Gauss curvature. After substituting also the Codazzi equations for $\betab$, we get that these terms have the form of the first term in the above expression. The equation for $\nab_3\nab^i K$ has a similar structure:
\begin{equation}\label{ee1.2}
\begin{split}
&\nab_3\nab^i K-\div\nab^i\betab\\
=&\sum_{i_1+i_2+i_3+i_4=i}\nab^{i_1}\psi^{i_2}\nab^{i_3}\psi_{\Hb}\nab^{i_4}{(K,\sigmac)}+\sum_{i_1+i_2+i_3+i_4=i+1}\nab^{i_1}\psi^{i_2}\nab^{i_3}\psi\nab^{i_4}\psi_{\Hb}.
\end{split}
\end{equation}
Finally, we have the following structure for $\nab_4\nab^i\betab$:
\begin{equation}\label{ee1.3}
\begin{split}
&\nab_4\nab^i\betab-\nabla\nab^iK -^*\nabla\nab^i\sigmac\\
=& \sum_{i_1+i_2+i_3+i_4=i}\nab^{i_1}\psi^{i_2}\nab^{i_3}\psi\nab^{i_4}(K,\sigmac)+\sum_{i_1+i_2+i_3=i-1}\psi^{i_1}\nab^{i_2}K\nab^{i_3}(K,\sigmac)\\
&+\sum_{i_1+i_2+i_3+i_4=i+1}\nab^{i_1}\psi^{i_2}\nab^{i_3}\psi_H\nab^{i_4}\psi_{\Hb}.
\end{split}
\end{equation}
As a shorthand, we denote by $F_{i,1}$ the terms of the form 
\begin{equation*}
\begin{split}
F_{i,1}:=&\sum_{i_1+i_2+i_3+i_4=i}\nab^{i_1}\psi^{i_2}\nab^{i_3}\psi_{\Hb}\nab^{i_4}(K,\sigmac)+\sum_{i_1+i_2+i_3+i_4= i+1}\nab^{i_1}\psi^{i_2}\nab^{i_3}\psi\nab^{i_4}\psi_{\Hb},
\end{split}
\end{equation*}
and by $F_{i,2}$ the terms of the form
\begin{equation*}
\begin{split}
F_{i,2}:=&\sum_{i_1+i_2+i_3+i_4=i}\nab^{i_1}\psi^{i_2}\nab^{i_3}\psi\nab^{i_4}(K,\sigmac)+\sum_{i_1+i_2+i_3= i-1}\psi^{i_1}\nab^{i_2}K\nab^{i_3}(K,\sigmac)\\
&+\sum_{i_1+i_2+i_3+i_4= i+1}\nab^{i_1}\psi^{i_2}\nab^{i_3}\psi_H\nab^{i_4}\psi_{\Hb}.
\end{split}
\end{equation*}
{Contracting \eqref{ee1.3} with $\nab^i\betab$, integrating} in the region $D_{u,\ub}$, applying Proposition \ref{intbypartssph} and using equations \eqref{ee1.1} and \eqref{ee1.2} yield the following identity on the derivatives of the curvature:
\begin{equation}\label{ee1.4}
\begin{split}
&\int_{D_{u,\ub}} f(u)^2 \langle\nab^i\betab,\nabla_4\nab^i\betab\rangle_\gamma\\
 =&\int_{D_{u,\ub}} f(u)^2 \langle\betab,\nabla\nab^i K+^*\nabla\nab^i\sigma\rangle_\gamma+f(u)^2 \langle\nab^i\betab,F_{i,2}\rangle_\gamma \\
=&\int_{D_{u,\ub}} -f(u)^2 \langle\div\nab^i\betab,\nab^iK\rangle_\gamma + f(u)^2\langle\div ^*\nab^i\betab,\nab^i\sigmac\rangle_\gamma +f(u)^2\langle\nab^i\betab,F_{i,2}\rangle_\gamma\\
=&\int_{D_{u,\ub}} -f(u)^2 \langle\nabla_3\nab^iK,\nab^iK\rangle_\gamma-f(u)^2\langle\nabla_3\nab^i\sigmac,\nab^i\sigmac\rangle_\gamma\\ 
&+\int_{D_{u,\ub}} f(u)^2 \langle\nab^i\betab,F_{i,2}\rangle_\gamma+f(u)^2 \langle\nab^i(K,\sigmac),F_{i,1}\rangle_\gamma.
\end{split}
\end{equation}
Using Proposition \ref{intbyparts34}, since $\nab_4 f(u)=0$, we have
\begin{equation}\label{ee1.5}
\begin{split}
&\int f(u)^2 \langle\nab^i\betab,\nabla_4\nab^i\betab\rangle_\gamma\\
=&\frac 12 (\int_{\Hb_{\ub}}f(u)^2|\nab^i\betab|^2 - \int_{\Hb_0}f(u)^2|\nab^i\betab|^2)+\int_{D_{u,\ub}} f(u)^2(\omega-\frac 12 \trch)|\nab^i\betab|^2.
\end{split}
\end{equation}
For the terms with $\nabla_3\nab^iK$ and $\nabla_3\nab^i\sigmac$, we similarly apply Proposition \ref{intbyparts34}, but noting that there is an extra contribution coming from $\nab_3f(u)$:
\begin{equation}\label{ee1.6}
\begin{split}
&\int_{D_{u,\ub}} f(u)^2 \langle\nab^i K,\nabla_3\nab^i K\rangle_\gamma\\
=&-\int_{D_{u,\ub}} f(u)\nab_3f(u)|\nab^i K|^2+\frac 12 (\int_{H_{u}}f(u)^2|\nab^i K|^2 - \int_{H_0}f(u)^2|\nab^i K|^2)\\
 & +\int_{D_{u,\ub}} f(u)^2(\omegab-\frac 12 \trchb)|\nab^i K|^2.
\end{split}
\end{equation}
Similarly,
\begin{equation}\label{ee1.7}
\begin{split}
&\int_{D_{u,\ub}} f(u)^2 \langle\nab^i \sigmac,\nabla_3\nab^i \sigmac\rangle_\gamma\\
=&-\int_{D_{u,\ub}} f(u)\nab_3f(u)|\nab^i \sigmac|^2+\frac 12 (\int_{H_{u}}f(u)^2|\nab^i \sigmac|^2 - \int_{H_0}f(u)^2|\nab^i \sigmac|^2)\\
 & +\int_{D_{u,\ub}} f(u)^2(\omegab-\frac 12 \trchb)|\nab^i \sigmac|^2.
\end{split}
\end{equation}
Combining \eqref{ee1.4}-\eqref{ee1.7}, we thus have the identity
\begin{equation*}
\begin{split}
&\int_{\Hb_{\ub}} f(u)^2|\nabla^i\betab|^2+\int_{H_{u}} f(u)^2|\nabla^i K|^2+\int_{H_{u}} f(u)^2|\nabla^i \sigmac|^2  \\
&-2\int_{D_{u,\ub}} f(u)\nab_3f(u)|\nab^i K|^2-2\int_{D_{u,\ub}} f(u)\nab_3f(u)|\nab^i \sigmac|^2\\
= &\int_{\Hb_{\ub'}} f(u)^2|\nabla^i\beta|^2+\int_{H_{u'}} f(u)^2|\nabla^i K|^2+\int_{H_{u'}} f(u)^2|\nabla^i \sigmac|^2\\
&-2\int_{D_{u,\ub}} f(u)^2(\omega-\frac 12 \trch)|\nab^i\betab|^2-2\int_{D_{u,\ub}} f(u)^2(\omegab-\frac 12 \trchb)(|\nab^iK|^2+|\nab^i \sigmac|^2)\\
&+\int_{D_{u,\ub}}f(u)^2\langle \nab^i\betab,F_{i,2}\rangle_\gamma+\int_{D_{u,\ub}}f(u)^2\langle\nab^i(K,\sigmac),F_{i,1}\rangle_\gamma.
\end{split}
\end{equation*}
The terms
$$-2\int_{D_{u,\ub}} f(u)\nab_3f(u)|\nab^i K|^2-2\int_{D_{u,\ub}} f(u)\nab_3f(u)|\nab^i \sigmac|^2$$
on the left hand side, which cannot be controlled\footnote{In fact, if we do not drop this term, we can control the spacetime integral $$\int_{D_{u,\ub}} (-f(u)\nab_3f(u))|\nab^i (K,\sigmac)|^2$$
where the weight $(-f(u)\nab_3f(u))$ can be singular. For weights such as $f(u)=(u-u_*)^{\alpha}$ for $\alpha<\frac 12$ or $f(u)=(u-u_*)^{\frac 12}\log^\beta(\frac{1}{u-u_*})$ for $\beta>\frac 12$, this bound is ``logarithmically'' stronger than simply taking the bound for $\int_{H_u}f(u)^2|\nab^i (K,\sigmac)|^2$ and integrating in $u$.} by the curvature flux (i.e., the integrals of $\nab^i$ of the curvature components along $H_u$ or $\Hb_{\ub}$), have a favorable sign! This is because the weight function $f$ satisfies $f(u)\nab_3f(u) <0$. Therefore, we get an \emph{inequality} for every $i$:
\begin{equation*}
\begin{split}
&\int_{\Hb_{\ub}} f(u)^2|\nabla^i\betab|^2+\int_{H_{u}} f(u)^2|\nabla^i K|^2+\int_{H_{u}} f(u)^2|\nabla^i \sigmac|^2  \\
\leq &\int_{\Hb_{\ub'}} f(u)^2|\nabla^i\beta|^2+\int_{H_{u'}} f(u)^2|\nabla^i K|^2+\int_{H_{u'}} f(u)^2|\nabla^i \sigmac|^2\\
&+C|| f(u)^2(\omega-\frac 12 \trch)\nab^i\betab\nab^i\betab||_{L^1_uL^1_{\ub}L^1(S)}+C|| f(u)^2(\omegab-\frac 12 \trchb)\nab^i(K,\sigmac)\nab^i(K,\sigmac)||_{L^1_uL^1_{\ub}L^1(S)}\\
&+C||f(u)^2 \nab^i\betab F_{i,2}||_{L^1_uL^1_{\ub}L^1(S)}+C||f(u)^2\nab^i(K,\sigmac) F_{i,1}||_{L^1_uL^1_{\ub}L^1(S)}.
\end{split}
\end{equation*}
We now add the above inequalities for $i\leq 3$. One can easily check that the terms $$\sum_{i\leq 3}|| f(u)^2(\omegab-\frac 12 \trchb)\nab^i(K,\sigmac)\nab^i(K,\sigmac)||_{L^1_uL^1_{\ub}L^1(S)},$$ $$\sum_{i\leq 3}||f(u)^2 \nab^i\betab F_{i,2}||_{L^1_uL^1_{\ub}L^1(S)}$$ and $$\sum_{i\leq 3}||f(u)^2\nab^i(K,\sigmac) F_{i,1}||_{L^1_uL^1_{\ub}L^1(S)}$$ have the form of one of the terms in the statement of the proposition.
After applying the Codazzi equation 
$$\betab=\nab\psi_{\Hb}+\psi(\psi+\psi_{\Hb})$$
to one of the $\betab$'s, we note that the term
$$\sum_{i\leq 3}|| f(u)^2(\omega-\frac 12 \trch)\nab^i\betab\nab^i\betab||_{L^1_uL^1_{\ub}L^1(S)}$$
is also one of the terms in the statement of the proposition.
\end{proof}
To close the energy estimates, we also need to control $\nab^i\beta$ in $L^2(H)$ and $\nab^i(K,\sigmac)$ in $L^2(\Hb)$. It is not difficult to see, by virtue of the structure of the Einstein equations, that Proposition \ref{ee1} also holds when all the barred and unbarred quantities are exchanged. The proof is exactly analogous to that of Proposition \ref{ee1}.
\begin{proposition}\label{ee2}
The following $L^2$ estimates for the curvature components hold:
\begin{equation*}
\begin{split}
&\sum_{i\leq 3}(||f(\ub)\nab^i(K,\sigmac)||^2_{L^\infty_{\ub} L^2_{u}L^2(S)}+||f(\ub)\nab^i\beta||^2_{L^\infty_{u} L^2_{\ub}L^2(S)})  \\
\leq &\sum_{i\leq 3}(||f(\ub)\nab^i(K,\sigmac)||^2_{L^2_{u}L^2(S_{u,0})}+||f(\ub)\nab^i\beta||^2_{L^2_{\ub}L^2(S_{0,\ub})}) \\
&+||f(\ub)^2(\sum_{i\leq 3}\nabla^i(K,\sigmac))(\sum_{i_1+i_2+i_3+i_4\leq 4}\nabla^{i_1}\psi^{i_2}\nabla^{i_3}\psi\nabla^{i_4}\psi_H)||_{L^1_uL^1_{\ub}L^1(S)}\\
&+||f(\ub)^2(\sum_{i\leq 3}\nabla^i(K,\sigmac))(\sum_{i_1+i_2+i_3+i_4\leq 3}\nab^{i_1}\psi^{i_2}\nabla^{i_3}\psi_H\nabla^{i_4}(K,\sigmac))||_{L^1_uL^1_{\ub}L^1(S)}\\
&+||f(\ub)^2(\sum_{i\leq 3}\nabla^i\beta)(\sum_{i_1+i_2+i_3+i_4\leq 3}\nab^{i_1}\psi^{i_2}\nabla^{i_3}\psi\nabla^{i_4}(K,\sigmac))||_{L^1_uL^1_{\ub}L^1(S)}\\
&+||f(\ub)^2(\sum_{i\leq 3}\nabla^i\beta)(\sum_{i_1+i_2+i_3+i_4\leq 2}\nab^{i_1}\psi^{i_2}\nabla^{i_3}K\nabla^{i_4}(K,\sigmac))||_{L^1_uL^1_{\ub}L^1(S)}\\
&+||f(\ub)^2(\sum_{i\leq 3}\nabla^i\beta)(\sum_{i_1+i_2+i_3+i_4\leq 4}\nabla^{i_1}\psi^{i_2}\nabla^{i_3}\psi_{\Hb}\nabla^{i_4}\psi_H)||_{L^1_uL^1_{\ub}L^1(S)}{.}
\end{split}
\end{equation*}
\end{proposition}

We now show that we can control all the nonlinear error terms in the energy estimates. 
We show this for $K$ and $\sigmac$ in $L^2(H_u)$ and $\betab$ in $L^2(\Hb_{\ub})$. 
The other case can be dealt with in a similar fashion {(see Proposition \ref{R2})}.
\begin{proposition}\label{R1}
There exists $\epsilon_0=\ep_0(\mathcal O_{ini},\mathcal R_{ini},\Delta_3)$ sufficiently small such that whenever $\epsilon\leq\epsilon_0$,
$$\sum_{i\leq 3}(||f(u)\nab^i(K,\sigmac)||_{L^\infty_u L^2_{\ub}L^2(S)}+||f(u)\nab^i\betab||_{L^\infty_{\ub} L^2_{u}L^2(S)}) \leq C(\mathcal O_{ini},\mathcal R_{ini}).$$
\end{proposition}
\begin{proof}
To prove the curvature estimates, we use Proposition \ref{ee1}. By assumptions of Theorem \ref{extthm} (see {also} Remark \ref{rmk.extthm}), the two terms corresponding to the initial data are bounded by a constant $C(\mathcal R_{ini})$ depending only on initial data. Therefore, we need to control the remaining five error terms in Proposition \ref{ee1}. We first look at the term
$$||f(u)^2(\sum_{i\leq 3}\nabla^i(K,\sigmac))(\sum_{i_1+i_2+i_3+i_4\leq 4}\nab^{i_1}\psi^{i_2}\nabla^{i_3}\psi_{\Hb}\nabla^{i_4}\psi)||_{L^1_uL^1_{\ub}L^1(S)}.$$
Using Propositions \ref{Ricci} and \ref{Ricci32}, together with the bootstrap assumption (\ref{BA3}), we have
\begin{equation*}
\begin{split}
&||f(u)^2(\sum_{i\leq 3}\nabla^i(K,\sigmac))(\sum_{i_1+i_2+i_3+i_4\leq 4}\nab^{i_1}\psi^{i_2}\nabla^{i_3}\psi_{\Hb}\nabla^{i_4}\psi)||_{L^1_uL^1_{\ub}L^1(S)}\\
\leq &C (\sum_{i\leq 3}||f(u)\nabla^i(K,\sigmac)||_{L^\infty_uL^2_{\ub}L^2(S)})(\sum_{i_1\leq 3}\sum_{i_2\leq 4}||\nabla^{i_1}\psi||_{L^\infty_uL^\infty_{\ub}L^2(S)}^{i_2})(\sum_{i_3\leq 4}||f(u)\nabla^{i_3}\psi||_{L^\infty_uL^2_{\ub}L^2(S)})\\
&\quad\times(\sum_{i_4\leq 3}||f(u)\nabla^{i_4}\psi_{\Hb}||_{L^2_uL^\infty_{\ub}L^2(S)})||f(u)^{-1}||_{L^2_u}\\
&+C\ep (\sum_{i\leq 2}||f(u)\nabla^i(K,\sigmac)||_{L^\infty_uL^2_{\ub}L^2(S)})||\psi||_{L^\infty_u L^\infty_{\ub}L^\infty(S)}||f(u)\nabla^4\psi_{\Hb}||_{L^\infty_{\ub}L^2_uL^2(S)}\\
\leq &C(\mathcal O_{ini})\Delta_3(1+\Delta_3)\epsilon.
\end{split}
\end{equation*}
The term
$$||f(u)^2(\sum_{i\leq 3}\nabla^i(K,\sigmac))(\sum_{i_1+i_2+i_3+i_4\leq 3}\nab^{i_1}\psi^{i_2}\nabla^{i_3}\psi_{\Hb}\nabla^{i_4}(K,\sigmac))||_{L^1_uL^1_{\ub}L^1(S)}\leq C(\mathcal O_{ini})\Delta_3(1+\Delta_3)\epsilon$$
similarly as in the previous estimate since by Propositions \ref{Kest} and \ref{sigmacest},
 $\nab^i(K,\sigmac)$ satisfies exactly the same estimates as $\nab^{i+1}\psi$.
We then consider the third nonlinear term
$$||f(u)^2(\sum_{i\leq 3}\nabla^i\betab)(\sum_{i_1+i_2+i_3+i_4\leq 3}\nab^{i_1}\psi^{i_2}\nabla^{i_3}\psi\nabla^{i_4}(K,\sigmac))||_{L^1_uL^1_{\ub}L^1(S)}.$$
Using Propositions \ref{Ricci} and \ref{Ricci32} and the bootstrap assumptions (\ref{BA3}), we have
\begin{equation*}
\begin{split}
&||f(u)^2(\sum_{i\leq 3}\nabla^i\betab)(\sum_{i_1+i_2+i_3+i_4\leq 3}\nab^{i_1}\psi^{i_2}\nabla^{i_3}\psi\nabla^{i_4}(K,\sigmac))||_{L^1_uL^1_{\ub}L^1(S)}\\
\leq &C\ep (\sum_{i\leq 3}||f(u)\nabla^i\betab||_{L^\infty_{\ub}L^2_uL^2(S)})(\sum_{i_1\leq 3,1\leq i_2\leq 3}||\nabla^{i_1}\psi||_{L^\infty_uL^\infty_{\ub}L^2(S)}^{i_2})||f(u)\sum_{i\leq 3}\nabla^i(K,\sigmac)||_{L^\infty_uL^2_{\ub}L^2(S)}\\
\leq & C(\mathcal O_{ini})\Delta_3(1+\Delta_3)\ep.
\end{split}
\end{equation*}
The fourth nonlinear term can be estimated analogously as the third nonlinear term by
$$||f(u)^2(\sum_{i\leq 3}\nabla^i\betab)(\sum_{i_1+i_2+i_3+i_4\leq 2}\nab^{i_1}\psi^{i_2}\nabla^{i_3}K\nabla^{i_4}(K,\sigmac))||_{L^1_uL^1_{\ub}L^1(S)}\leq C(\mathcal O_{ini})\Delta_3(1+\Delta_3)\ep.$$
As before, this is because by Propositions \ref{Kest} and \ref{sigmacest},
 $\nab^i(K,\sigmac)$ satisfies exactly the same estimates as $\nab^{i+1}\psi$. 
Thus it remains to control
$$ ||f(u)^2(\sum_{i\leq 3}\nabla^{i}\betab)(\sum_{i_1+i_2+i_3+i_4\leq 4}\nabla^{i_1}\psi^{i_2}\nabla^{i_3}\psi_H\nabla^{i_4}\psi_{\Hb})||_{L^1_uL^1_{\ub}L^1(S)}.$$
This term can be bounded as follows:
\begin{equation*}
\begin{split}
&||f(u)^2(\sum_{i\leq 3}\nabla^{i}\betab)(\sum_{i_1+i_2+i_3+i_4\leq 4}\nabla^{i_1}\psi^{i_2}\nabla^{i_3}\psi_H\nabla^{i_4}\psi_{\Hb})||_{L^1_uL^1_{\ub}L^1(S)} \\
\leq &C(\sum_{i\leq 3}||f(u)\nabla^{i}\betab||_{L^\infty_{\ub} L^2_{u}L^2(S)})(\sum_{i_1\leq 3}\sum_{i_2\leq 3}||\nabla^{i_1}\psi||_{L^\infty_uL^\infty_{\ub}L^2(S)}^{i_2})\\
&\quad\times(\sum_{i_3\leq 3}||f(\ub)\nabla^{i_3}\psi_H||_{L^2_{\ub}L^\infty_{u}L^2(S)})(\sum_{i_4\leq 4}||f(u)\nabla^{i_4}\psi_{\Hb}||_{L^\infty_{\ub}L^2_{u}L^2(S)})||f(\ub)^{-1}||_{L^2_{\ub}}\\
&+C(\sum_{i\leq 2}||f(u)\nabla^{i}\betab||_{L^\infty_{\ub} L^2_{u}L^2(S)})||f(\ub)^{-1}||_{L^2_{\ub}}||f(\ub)\nabla^4\psi_H||_{L^\infty_u L^2_{\ub}L^\infty(S)}||f(u)\psi_{\Hb}||_{L^2_{u}L^\infty_{\ub}L^\infty(S)}\\
\leq &C(\mathcal O_{ini})\Delta_3(1+\Delta_3)\ep.
\end{split}
\end{equation*}
Therefore, gathering all the above estimates, we have
$$\sum_{i\leq 3}(||f(u)\nab^i(K,\sigmac)||_{L^\infty_u L^2_{\ub}L^2(S)}^2+||f(u)\nab^i\betab||_{L^\infty_{\ub} L^2_{u}L^2(S)}^2)\leq C(\mathcal O_{ini},\mathcal R_{ini})+C(\mathcal O_{ini})\Delta_3(1+\Delta_3)\ep,$$
which implies the conclusion of the proposition after taking $\ep$ to be sufficiently small.
\end{proof}

Notice that the schematic equations are symmetric under the change $\nab_3\leftrightarrow\nab_4$, $u\leftrightarrow\ub$ and $\psi_H\leftrightarrow\psi_{\Hb}$. Since the conditions for the initial data are also symmetric, we also have the following analogous energy estimates for $\nab^i\beta$ on $H_u$ and $\nab^i(K,\sigmac)$ on $\Hb_{\ub}$:
\begin{proposition}\label{R2}
There exists $\epsilon_0=\ep_0(\mathcal O_{ini},\mathcal R_{ini},\Delta_3)$ sufficiently small such that whenever $\epsilon\leq\epsilon_0$,
$$\sum_{i\leq 3}(||f(\ub)\nab^i\beta||_{L^\infty_u L^2_{\ub}L^2(S)}+||f(\ub)\nab^i(K,\sigmac)||_{L^\infty_{\ub} L^2_{u}L^2(S)}) \leq C(\mathcal O_{ini},\mathcal R_{ini}).$$
\end{proposition}

Propositions \ref{R1} and \ref{R2} together imply
\begin{proposition}\label{R.final}
There exists $\epsilon_0=\epsilon_0(\mathcal O_{ini},\mathcal R_{ini})$ such that whenever $\epsilon\leq \epsilon_0$,
$$\mathcal R\leq C(\mathcal O_{ini},\mathcal R_{ini}).$$
\end{proposition}
\begin{proof}
Let 
$$\Delta_3 \gg C(\mathcal O_{ini},\mathcal R_{ini}),$$
where $C(\mathcal O_{ini},\mathcal R_{ini})$ is taken to be the maximum of the {upper} bounds in Propositions \ref{R1} and \ref{R2}. Hence, the choice of $\Delta_3$ depends only on $\mathcal O_{ini}$ and $\mathcal R_{ini}$. Thus, by Propositions \ref{R1} and \ref{R2}, the bootstrap assumption (\ref{BA3}) can be improved by choosing $\epsilon$ sufficiently small depending on $\mathcal O_{ini}$ and $\mathcal R_{ini}$.
\end{proof}
{Combining Propositions \ref{Ricci}, \ref{Ricci32} and \ref{R.final}, we conclude} the proof of Theorem \ref{aprioriestimates}. {As mentioned previously, standard methods then imply Theorem \ref{extthm}.}

\section{Nature of the Singular Boundary}\label{secnature}

As described by Theorems \ref{C0extthm} and \ref{blowupthm}, we will also prove the regularity and singularity of the boundary $H_{u_*}$ and $\Hb_{\ub_*}$. We first prove the regularity of the boundary asserted in Theorem \ref{C0extthm}.
\begin{proof}[Proof of Theorem \ref{C0extthm}]
The fact that $(\mathcal M, g)$ can be {extended continuously} up to and beyond $H_{u_*}$ and $\Hb_{\ub_*}$ simply follows from the continuity of the metric components $\Omega$, $\gamma$ and $b$  proved in Propositions \ref{Omega}-\ref{b}. To obtain the higher regularity for $\gamma$, we recall the equations \eqref{Omegatransport}, \eqref{1st.var} and \eqref{btrans}:
\beaa
\frac{\partial}{\partial \ub}\Omega^{-1}=2\omega,\quad {\frac{\partial}{\partial \ub}\gamma_{AB}=2\Omega\chi_{AB}},\quad \frac{\partial}{\partial \ub}b^A=-4\Omega^2\zeta^A.
\eeaa
Commuting these equations with $(\frac{\partial}{\partial\th})^i$ and using the bounds\footnote{Notice that by controlling $\gamma$ and its coordinate angular derivatives $(\frac{\partial}{\partial \th})^i\gamma$, we can show also that $\frac{\partial}{\partial \th}$ and $\nab$ are comparable up to lower order terms, which allows us to apply the estimates for $\nab^i\trch$, $\nab^i\chih$, $\nab^i\eta$ and $\nab^i\etab$ to bound the coordinate angular derivatives of the metric components.} for the Ricci coefficients obtained in the proof of Theorem \ref{aprioriestimates}, we conclude that 
$$\sum_{i_1+i_2\leq 4} \sup_{0\leq u\leq {u_*}}\sup_{0\leq \ub\leq\ub_*}{\|(\pr{\th^1})^{i_1}(\pr{\th^2})^{i_2}} (\gamma,b,\Omega){\|}_{L^2(U_i(u,\ub))}{\leq C}. $$
The boundedness of $\psi$ and its angular derivatives 
$$\sum_{i\leq 3}\|\nab^i\psi\|_{L^\infty_u L^\infty_{\ub} L^2(S)}\leq C.$$
are already proved in Theorem \ref{aprioriestimates}. To control $\psi_{\Hb}$ and its angular derivatives on the singular boundary $\Hb_{\ub_*}$, we first note that by the smoothness assumption on the \emph{interior} of {the initial hypersurface} $\Hb_0$, we have that for every fixed $U\in [0,u_*)$, 
$${\sum_{i\leq 5}}\sup_{0\leq u\leq U}\|\nab^i\psi_{\Hb}\|_{L^2(S_{u,0})}\leq C_{{U}}$$
for some finite $C_{U}$.
We {now} revisit the proof of Proposition \ref{psiHb} {to bound $\nab^i\psi_{\Hb}$ up to $i\leq 3$ for $u\in [0,U]$}. Restricting to $[0,U]$, $f(u)^{-1}$ is bounded. Therefore, the estimates in \eqref{psihb1}, \eqref{psihb2} and \eqref{psihb3} are bounded uniformly in $u$. Finally, \eqref{psihb4} can be replaced by the estimate
\begin{equation*}
\begin{split}
&\int_0^{\ub}\sum_{i_1+i_2+i_3+i_4\leq 3}{\|}\nabla^{i_1}\psi^{i_2}\nabla^{i_3}\psi_H\nabla^{i_4}\psi_{\Hb}||_{L^2(S_{u,\ub'})}d\ub' \\
\leq &C\int_0^{\ub}(\sum_{i_1\leq 3}||\nab^{i_1}\psi_H||_{L^2(S_{u,\ub'})})(\sum_{i_2\leq 3}\sup_{0\leq \ub''\leq \ub}||\nab^{i_2}\psi_{\Hb}||_{L^2(S_{u,\ub''})})d\ub'.
\end{split}
\end{equation*}
Putting these bounds together, we have
\begin{equation*}
\begin{split}
&\sum_{i\leq 3}\sup_{\substack{0\leq u\leq U\\0\leq \ub'\leq \ub}}\|\nab^i\psi_{\Hb}\|_{L^2(S_{u,\ub'})}\\
\leq &C_{{U}}+C\int_0^{\ub}(\sum_{i_1\leq 3}||\nab^{i_1}\psi_H||_{L^2(S_{u,\ub'})})(\sum_{i_2\leq 3}\sup_{\substack{0\leq u'\leq U\\0\leq \ub''\leq {\ub'}}}||\nab^{i_2}\psi_{\Hb}||_{L^2(S_{u,\ub''})})d\ub',
\end{split}
\end{equation*}
which implies
\begin{equation}\label{higher.reg.pHb}
\sum_{i\leq 3}\sup_{\substack{0\leq u\leq U\\0\leq \ub\leq \ub_*}}\|\nab^i\psi_{\Hb}\|_{L^2(S_{u,\ub})}\leq {C_U}
\end{equation}
after applying Gronwall's inequality.

{To conclude the proof, it remains to control $\nab_3\nab^i\psi$ and $\nab_3\nab^i\psi_{\Hb}$ for $i\leq 2$. Since $\etab$ obeys a $\nab_3$ equation (see \eqref{null.str2}), by directly controlling the right hand side of the null structure equation (commuted with angular derivatives) and using the bounds in Theorem \ref{aprioriestimates}, we get
$$
\sum_{i\leq 2}{\sup_{\substack{0\leq u\leq U\\0\leq \ub\leq \ub_*}}}\|\nab_3\nab^i\etab\|_{L^2(S_{u,\ub})}\leq C_U.
$$
To control the term $\nab_3\nab^i\eta$, notice that combining the $\nab_3\etab$ equation in \eqref{null.str2} and the equations in \eqref{RC.relation}, we have
$$\nab_3\eta=-\nab_3\etab+2\nab_3\nab(\log\Omega)=\chib\cdot (\etab-\eta)-\bb-\nab\omb-4\omegab\nab(\log\Om)-2\chib\cdot\nab(\log\Om).$$
Upon expressing $\bb$ in terms of $\psi_{\Hb}$ using the Codazzi equation in \eqref{null.str3}, commuting the equation with $\nab^i$ and using the bound \eqref{higher.reg.pHb}, we get
\begin{equation}\label{3ieta}
\sum_{i\leq 2}{\sup_{\substack{0\leq u\leq U\\0\leq \ub\leq \ub_*}}}\|\nab_3\nab^i\eta\|_{L^2(S_{u,\ub})}\leq C_{{U}}.
\end{equation}}
Finally, we control the terms $\nab_3\nab^i\psi_{\Hb}$. Commuting the null structure equations {for $\nab_4\psi_{\Hb}$ in \eqref{null.str1} and \eqref{null.str2}} with $\nab_3\nab^i$, we have
\begin{equation*}
\begin{split}
&\nab_4\nab_3\nab^i\psi_{\Hb}\\
=&\sum_{\substack{j_1+j_2+j_3+j_4=1\\i_1+i_2+i_3+i_4=i}}(\nab^{i_1}\psi_{\Hb}^{j_1}\nab_3^{j_2}\nab^{i_2}\psi^{i_3}\nab_3^{j_3}\nab^{i_4}K+\nab^{i_1}\psi_{\Hb}^{j_1}\nab_3^{j_2}\nab^{i_2}\psi^{i_3}\nab_3^{j_3}\nab^{i_4}\nab\psi)\\
&+\sum_{\substack{j_1+j_2+j_3=1\\i_1+i_2+i_3+i_4+i_5=i}}\nab^{i_1}\psi_{\Hb}^{j_1}\nab_3^{j_2}\nab^{i_2}\psi^{i_3}\nab_3^{j_3}\nab^{i_4}\psi\nab^{i_5}\psi\\
&+\sum_{\substack{j_1+j_2+j_3+j_4=1\\i_1+i_2+i_3+i_4+i_5=i}}\nab^{i_1}\psi_{\Hb}^{j_1}\nab_3^{j_2}\nab^{i_2}\psi^{i_3}\nab_3^{j_3}\nab^{i_4}\psi_{\Hb}\nab_3^{j_4}\nab^{i_5}\psi_H.
\end{split}
\end{equation*}
Estimating directly the right hand side of the null structure equations or the Bianchi equations, we can easily show that
$$\sum_{i\leq 2}{\sup_{0\leq u\leq U}}\|\nab_3(\nab^iK,\nab^i\etab,\nab^i\psi_H)\|_{L^1_{\ub}L^2(S)}\leq C_{{U}}.$$
Using also \eqref{3ieta}, we thus have
$$\sum_{i\leq 2}{\sup_{0\leq u\leq U}}\|\nab_3\nab^i\psi_{\Hb}\|_{L^2(S_{u,\ub})}\leq C_{{U}}+C_{{U}}\int_0^{\ub}\sum_{i_1+i_2+i_3+i_4\leq 2}\|\nab^{i_1}\psi^{i_2}\nab_3\nab^{i_3}\psi_{\Hb}\nab^{i_4}\psi_H\|_{L^2(S_{u,\ub'})} d\ub'.$$
Using Gronwall's inequality, we get
$$\sum_{i\leq 2}{\sup_{\substack{0\leq u\leq U\\0\leq \ub\leq \ub_*}}}\|\nab_3\nab^i\psi_{\Hb}\|_{L^2(S_{u,\ub})}\leq C_{{U}}.$$
In particular, {combining} the above estimates{, we obtain} 
$$\displaystyle\sum_{\substack{i\leq 3-j,\,j\leq 1}}\sup_{0\leq u\leq U}\|\nab^j\nab^i(\psi_{\Hb},\psi)\|_{L^2(S_{u,\ub_*})}\leq C_{{U}}$$ on $\Hb_{\ub_*}$, as desired.
\end{proof}
Finally, we move to the proof of Theorem \ref{blowupthm}. First, we prove
\begin{proposition}\label{prop.blow.up}
Suppose, in addition to the assumptions in Theorem \ref{extthm}, $\chih$ initially obeys
$$\int_0^{\ub_*} |\chih\restriction_{\gamma}(\ub')|^2 d\ub' =\infty,$$
along an outgoing null generator $\gamma$ of $H_0$.
Let $\Phi_u(\gamma)$ be the image of $\gamma$ under the 1-parameter family of diffeomorphism generated by $\Lb$. Then 
$$\int_0^{\ub_*} (\trch\restriction_{\Phi_u(\gamma)}(\ub'))^2+|\chih\restriction_{\Phi_u(\gamma)}(\ub')|^2 d\ub' =\infty,$$
holds for every $0\leq u <u_*$.

Similarly suppose, in addition to the assumptions in Theorem \ref{extthm}, $\chibh$ initially obeys
$$\int_0^{\ub_*} |\chibh\restriction_{\gamma}(u')|^2 du' =\infty,$$
along an outgoing null generator $\gamma$ of $\Hb_0$.
Let ${\underline{\Phi}}_{\ub}(\gamma)$ be the image of $\gamma$ under the 1-parameter family of diffeomorphism generated by $L$. Then 
$$\int_0^{u_*} (\trchb\restriction_{\Phi_{\ub}(\gamma)}(u'))^2+|\chibh\restriction_{\Phi_{\ub}(\gamma)}(u')|^2 du' =\infty,$$
holds for every $0\leq \ub <\ub_*$.
\end{proposition}
\begin{proof}
Fix $U\in (0,u_*{)}$. Suppose
\begin{equation}\label{assumption}
\int_0^{\ub_*} (\trch\restriction_{\Phi_U(\gamma)}(\ub'))^2 d\ub' <\infty.
\end{equation}
We want to show that {under the assumption \eqref{assumption}, we have}
$$\int_0^{\ub_*} |\chih\restriction_{\Phi_U(\gamma)}(\ub')|^2 d\ub' =\infty{,}$$
{which will then imply the desired conclusion.}

Using \eqref{assumption}, define $h:[0,\ub_*)\to \mathbb R$ by
$$h(\ub)=|\trch\restriction_{\Phi_U(\gamma)}(\ub)|$$
such that 
$$\int_0^{\ub_*} h(\ub')^2 d\ub'<\infty.$$
Consider the following null structure equation for $\trch$:
$$\nab_3 \trch+ \trchb \trch =2\omegab \trch-2K+2\div \eta+2|\eta|^2$$
Along the integral curve of $-e_3$ emanating from $\Phi_u(\gamma)$, we thus have
$$\frac{d}{du} (e^{\int^u_U ({\Om{tr}\underline\chi}-2{\Om}\omegab)\restriction_{\Phi_{u'}(\gamma)} (\ub) du'}\trch \restriction_{\Phi_{u}(\gamma)} (\ub))=e^{\int^u_U ({\Om{tr}\underline\chi}-2{\Om}\omegab)\restriction_{\Phi_{u'}(\gamma)} (\ub) du'}(-2K+2\div \eta+2|\eta|^2).$$
By the estimates derived in the proof of Theorem \ref{aprioriestimates}, $K$, $\nab\eta$, $\eta$ are bounded and $\trchb$, $\omegab$ are in $L^1_u L^\infty(S)$. Therefore, 
\begin{equation}\label{assumption.1}
|\trch\restriction_{\Phi_u(\gamma)}(\ub)|\leq Ch(\ub)\quad\mbox{for all }u.
\end{equation}
Consider the following null structure equation for $\chih$:
$$\nab_3\chih+\frac 1 2 \trchb \chih =\nab\widehat{\otimes} \eta+2\omegab \chih-\frac 12 \trch \chibh +\eta\widehat{\otimes} \eta.$$
Contract this equation with $\chih$ to get
$$\frac 12 \nab_3|\chih|^2+\frac 1 2 \trchb |\chih|^2 -2\omegab|\chih|^2=(\nab\widehat{\otimes} \eta-\frac 12 \trch \chibh +\eta\widehat{\otimes} \eta)\cdot\chih,$$
which implies
$$ |\nab_3|\chih|+\frac 1 2 \trchb |\chih| -2\omegab|\chih||\leq |\nab\widehat{\otimes} \eta|+|\frac 12 \trch \chibh| +|\eta\widehat{\otimes} \eta|.$$
This implies that along the integral curve of $e_3$, we have
\begin{equation*}
\begin{split}
&|\frac{d}{du} (e^{\int^u_U (\frac 12{\Om{tr}\underline\chi}-2{\Om}\omegab)\restriction_{\Phi_{u'}(\gamma)} (\ub) du'} |\chih{|} \restriction_{\Phi_{u}(\gamma)} (\ub))|\\
\leq &2e^{\int^u_U (\frac 12{\Om{tr}\underline\chi}-2{\Om}\omegab)\restriction_{\Phi_{u'}(\gamma)} (\ub) du'}(|\nab\widehat{\otimes} \eta|+|\frac 12 \trch \chibh| +|\eta\widehat{\otimes} \eta|).
\end{split}
\end{equation*}
Using again the fact that $K$, $\nab\eta$, $\eta${, $\trchb$, $\chibh$, $\omegab$} are bounded {for $u\leq U$}, as well as the estimate \eqref{assumption.1}, we have
\begin{equation*}
\begin{split}
&|(e^{\int^u_U (\frac 12{\Om{tr}\underline\chi}-2{\Om}\omegab)\restriction_{\Phi_{u'}(\gamma)} (\ub) du'} |\chih{|} \restriction_{\Phi_{u}(\gamma)} (\ub))-(e^{\int^u_U ({\f 12\Om{tr}\underline\chi}-2{\Om}\omegab)\restriction_{\gamma} (\ub) du'} |\chih{|} \restriction_{\gamma} (\ub))|
\leq C_U(1+h(\ub)).
\end{split}
\end{equation*}
Notice that $e^{\int^u_U ({\f 12\Om{tr}\underline\chi}-2{\Om}\omegab)\restriction_{\Phi_{u'}(\gamma)} (\ub) du'}$ is bounded above and below uniformly in $\ub$. Taking the $L^2_{\ub}$ norm implies that {for $u\leq U$, we have}
\begin{equation*}
\int_0^{\ub_*}|\chih\restriction_{\Phi_u(\gamma)}(\ub')|^2 d\ub'\geq c\int_0^{\ub_*} |\chih\restriction_{\gamma}(\ub')|^2 d\ub'-C-C\int_0^{\ub_*} h^2(\ub')d\ub' =\infty
\end{equation*}
by the assumption of the proposition. The blow up for $\chib$ can be proved in a similar manner.
\end{proof}
This implies
\begin{proposition}\label{prop.blow.up.2}
Suppose the assumptions of Theorem \ref{blowupthm} hold. Then, in a neighborhood of any point on $\Hb_{\ub_*}$, $|{\chi}|^2$ is not integrable with respect to the spacetime volume form. Similarly, in a neighborhood of any point on $H_{u_*}$, $|{\chib}|^2$ is not integrable with respect to the spacetime volume form. 
\end{proposition}
\begin{proof}
{
We begin with $|\chib|^2$ near $H_{u_*}$. By definition, the image of the initial incoming null generator under the map $\underline\Phi_{\ub}$ defined in Proposition \ref{prop.blow.up} has constant $\ub$}, $\th^1$ and $\th^2$ values. Also, by Propositions \ref{Omega} and \ref{gamma}, the spacetime volume element $2\Omega^2\sqrt{\det \gamma}$ is bounded {uniformly} above and below. Therefore, for any neighborhood $\mathcal N$ of $p=(u,\ub_*,\th^1,\th^2)\in \Hb_{\ub_*}$, we have
\begin{equation*}
\begin{split}
&\int_{\mathcal N} ((\trchb)^2+|\chibh|^2) \\
\geq &c \int_{\th^2-\delta}^{{\th^2}+\delta}\int_{\th^1-\delta}^{\th^1+\delta}\int_{u-\delta}^{u+\delta}\int_{\ub_*-\delta}^{\ub_*}((\trchb)^2+|\chibh|^2)(u',\ub',(\th^1)',(\th^2)') d\ub'\, du'\,d(\th^1)'\, d(\th^2)'=\infty,
\end{split}
\end{equation*}
by Proposition \ref{prop.blow.up}.

To prove the corresponding statement for $|{\chi}|^2$ near ${\Hb_{\ub_*}}$, we first change to the coordinate system $(u,\ub,\tilde{\th}^1(\ub;u,\th),\tilde{\th}^2(\ub;u,\th))$ such that {$\Lb=\f{\rd}{\rd u}$}. This coordinate system can be {constructed} by solving the ordinary differential {equations}
$${\frac{d}{du}\tilde{\th}^A(u;\ub,\th)=-b^A(u,\ub,\tilde\th^1,\tilde\th^2),}$$
with initial {condition}\footnote{{We note that since we do not have a global coordinate chart on $S_{0,0}$, the above ODE only makes sense in $(\Phi_u\circ \underline{\Phi}_{\ub})(U_i)\cap (\underline{\Phi}_{\ub}\circ\Phi_u)(U_j)$, where $U_i$, $U_j$ are coordinate charts on $S_{0,0}$ and $\Phi_u$ and $\underline{\Phi}_{\ub}$ are as defined in Proposition \ref{prop.blow.up}. Nevertheless, since $\Phi_u\circ \underline{\Phi}_{\ub}$ and $\underline{\Phi}_{\ub}\circ\Phi_u$ are both diffeomorphisms between $S_{0,0}$ and $S_{u,\ub}$, for every point $p\in S_{u,\ub}$, there exists $i$ and $j$ such that $p\in (\Phi_u\circ \underline{\Phi}_{\ub})(U_i)\cap (\underline{\Phi}_{\ub}\circ\Phi_u)(U_j)$, where this change of coordinates makes sense.}}
$${\tilde{\th}^A(0;\ub,\th)}=\th^A.$$
By \eqref{btrans}, as well as the estimates for $\zeta$, $\Omega$ and their derivatives, {$b^A$ and the following first derivatives of $b^A$ are uniformly bounded:
$$|b^A|,\,|\frac{\partial b^A}{\partial \ub}|,\,|\frac{\partial b^A}{\partial \th^B}|\leq C. $$
T}herefore, 
$$|\frac{\partial \tilde{\th}^A}{\partial u}|,\,|\frac{\partial \tilde{\th}^A}{\partial \ub}|,\,|\frac{\partial \tilde{\th}^A}{\partial \th^B}|\leq C.$$
In the new coordinate system, we apply the same argument as in the case for $|{\chib}|^2$ near $H_{u_*}$ and have the estimate
$$\int_{\mathcal N} {((\trch)^2+}|\chih|^2{)} =\infty$$
for {any} neighborhood $\mathcal N$ of any point $p\in {\Hb_{\ub_*}}$, as desired.
\end{proof}
Finally, this allows us to conclude that the Christoffel symbols do not belong to $L^2$:
\begin{proposition}
Suppose the assumptions of Theorem \ref{blowupthm} hold. Then, the Christoffel symbols in the $(u,\ub,\th^1,\th^2)$ coordinate system are not in $L^2$ in a neighborhood of any point on $H_{u_*}$ or $\Hb_{\ub_*}$.
\end{proposition}
\begin{proof}
Recall that the metric in the $(u,\ub,\th^1,\th^2)$ coordinates takes the form
$$g=-2\Omega^2(du\otimes d\ub+d\ub \otimes du)+\gamma_{AB} (d\th^A-b^A du)\otimes (d\th^B-b^B du).$$
Note that
$$g^{u\ub}=-\frac 12 \Omega^{-2},\quad g^{u\alpha}=0\mbox{ for }\alpha\neq \ub.$$
One computes that 
\begin{equation*}
\begin{split}
\Gamma^{u}_{AB}=&-\frac 12 g^{u \ub} \frac{\partial}{\partial \ub}g_{AB}
=\frac{1}{4 \Omega^2} \frac{\partial}{\partial \ub}\gamma_{AB}
=\frac{1}{2 \Omega} \chi_{AB}.
\end{split}
\end{equation*}
Since $\frac 12\leq \Om\leq 2$ and $\gamma$ is uniformly bounded and positive definite, $\Gamma^u_{AB}$ is not in $L^2$ in a neighborhood of any point on the singular boundary $\Hb_{\ub_*}$ in the $(u,\ub,\th^1,\th^2)$ coordinate system.

To show that the incoming hypersurface $H_{u_*}$ is singular, first notice that
$$g^{\ub u}=-\frac 12 \Omega^{-2},\quad g^{\ub A}=-\frac 12 \Omega^{-2}b^A,\quad g^{\ub\ub}=0.$$
We then compute
\begin{equation*}
\begin{split}
\Gamma^{\ub}_{AB}=&\frac 12 g^{\ub u}(\frac{\partial}{\partial \th^A}g_{Bu}+\frac{\partial}{\partial \th^B}g_{Au} -\frac{\partial}{\partial u}g_{AB})+\frac 12 g^{\ub C}(\frac{\partial}{\partial \th^B}g_{AC}+\frac{\partial}{\partial \th^A}g_{BC}-\frac{\partial}{\partial \th^C}g_{AB})\\
=&\frac{1}{4 \Omega^2} (\frac{\partial}{\partial u}\gamma_{AB}-\frac{\partial}{\partial \th^B}(\gamma_{AC} b^C)-\frac{\partial}{\partial \th^A}(\gamma_{BC} b^C)-b^C(\frac{\partial}{\partial \th^B}\gamma_{AC}+\frac{\partial}{\partial \th^A}\gamma_{BC}-\frac{\partial}{\partial \th^C}\gamma_{AB} ))\\
=&\frac{1}{2 \Omega} \chib_{AB}+\mbox{regular terms},
\end{split}
\end{equation*}
where the regular terms denote metric components and their derivatives that are uniformly bounded by the estimates proved in the previous sections. By the same reasoning as in the case near $\Hb_{\ub_*}$, $\Gamma^{\ub}_{AB}$ is not in $L^2$ in a neighborhood of any point on the singular boundary $H_{u_*}$ in the $(u,\ub,\th^1,\th^2)$ coordinate system.

\end{proof}
This concludes the proof of Theorem \ref{blowupthm}.

\section{Acknowledgments}

The author thanks Mihalis Dafermos for suggesting the problem and sharing many insights from the works \cite{D1, D2, D3}, as well as offering valuable comments on an earlier version of the manuscript. He thanks Igor Rodnianski for very helpful suggestions. He also thanks Spyros Alexakis, Amos Ori and Yakov Shlapentokh-Rothman for stimulating discussions. Finally, he is grateful for the suggestions given by the anonymous referees.

Most of the work was carried out when the author was at Princeton University and University of Pennsylvania. This work is supported by the NSF Postdoctoral Fellowship DMS-1204493.

\bibliographystyle{hplain}
\bibliography{weaknull18}

\end{document}